%% file: postedlectures.tex
\documentclass{article}
\usepackage{latexsym}
\usepackage{amsmath,amsthm}
\usepackage{amssymb}
\usepackage{upgreek}
\usepackage{MnSymbol, wasysym}
\usepackage{graphics}
\usepackage{marvosym}
\usepackage{color}
\def\tr{\mbox{tr}}

\def\f12{\frac 1 2}

\def\a{\alpha}
\def\b{\beta}

\def\de{\delta}

\def\si{\sigma}

\def\pa{\partial}

\def\f12{\frac 1 2}

\newcommand{\nabb}{\mbox{$\nabla \mkern-13mu /$\,}}
\newcommand{\lessflat}
{\reflectbox{\mbox{$\flat$}}}

\newtheorem{theorem}{Theorem}[section]
\newtheorem{proposition}{Proposition}[subsection]
\newtheorem{conjecture}{Conjecture}[section]

\newtheorem{corollary}{Corollary}[section]
\newtheorem{problem}{Open problem}
\title{Lectures on black holes and linear waves}
\author{Mihalis Dafermos\thanks{University of Cambridge,
Department of Pure Mathematics and Mathematical Statistics,
Wilberforce Road, Cambridge CB3 0WB United Kingdom}
\and
 Igor Rodnianski\thanks{Princeton University,
Department of Mathematics, Fine Hall, Washington Road,
Princeton, NJ 08544 United States}
}
\begin{document}
\maketitle
\begin{abstract}
These lecture notes,  based on a course 
given at the Z\"urich Clay Summer School (June 23--July 18 2008),
review our current mathematical understanding of 
the global behaviour of waves on black hole exterior backgrounds.
Interest in this problem stems from its relationship to the non-linear stability of
the black hole spacetimes themselves
as solutions to the Einstein equations, one of the central open problems
of general relativity. 
After an introductory discussion of the Schwarzschild geometry and the black hole
concept,
the classical theorem of Kay and Wald  on the
boundedness of scalar waves on the exterior region of
Schwarzschild is reviewed. The original proof is
presented, followed by
a new more robust proof of a stronger  boundedness
statement. The problem of decay of scalar waves on Schwarzschild
is then addressed, and a theorem proving quantitative decay is stated and its proof sketched.
This decay statement is carefully contrasted with the type of statements derived
heuristically in the physics literature for the asymptotic tails
of individual spherical harmonics.
Following this, our recent proof of the boundedness of solutions to the wave equation on 
axisymmetric stationary backgrounds (including slowly-rotating Kerr and Kerr-Newman)
is reviewed and a new decay result for slowly-rotating 
Kerr spacetimes is stated and proved. This last result was announced
at the summer school and appears in print here for the first time. 
A discussion of the analogue of these problems for spacetimes with a positive
cosmological constant $\Lambda>0$ follows. Finally, a general framework
is given for capturing the red-shift effect for non-extremal black holes.
This unifies and extends some of the analysis of the previous sections.
The notes end with a collection of open problems.
\end{abstract}
\tableofcontents
\section{Introduction: General relativity and evolution}
Black holes are one of the fundamental predictions of general relativity. 
At the same time, they are one of its least understood (and most
often misunderstood) aspects.
These lectures intend to introduce the black hole concept
and the analysis of waves on black hole backgrounds $(\mathcal{M},g)$ 
by means of the example of the scalar wave equation
\begin{equation}
\label{fromintro}
\Box_g\psi=0.
\end{equation}

We do not assume the reader is familiar with general relativity, only basic analysis
and differential geometry. In this introductory
section, we briefly describe general relativity in outline form, taking
from the beginning
the evolutionary point of view which puts the Cauchy problem for the 
\emph{Einstein equations}--the system 
of nonlinear partial differential equations (see $(\ref{Eeq})$ below) governing the
theory--at the centre.  
The problem
 $(\ref{fromintro})$ can be viewed as a poor man's linearisation for the
Einstein equations. Study of $(\ref{fromintro})$ is then intimately related to the problem of
the dynamic
stability of the black hole spacetimes $(\mathcal{M},g)$ themselves.
Thus, one should view the subject of these lectures as intimately connected to the very tenability
of the black hole concept in the theory.

\subsection{General relativity and the Einstein equations}
General relativity postulates a
$4$-dimensional Lorentzian manifold $(\mathcal{M}, g)$--\emph{space-time}--which is
to satisfy the \emph{Einstein equations}
\begin{equation}
\label{Eeq}
R_{\mu\nu}-\frac12 g_{\mu\nu} R = 8\pi T_{\mu\nu}.
\end{equation}
Here,
$R_{\mu\nu}$, $R$ denote the \emph{Ricci} and \emph{scalar} curvature of
$g$, respectively, and $T_{\mu\nu}$ denotes a symmetric $2$-tensor on $\mathcal{M}$
termed the \emph{stress-energy-momentum tensor}
of matter. (Necessary background on Lorentzian geometry to understand the
above notation is given in Appendix~\ref{Lorge}.)
The equations $(\ref{Eeq})$ in of themselves do not close, but must be
coupled to ``matter equations'' satisfied by a collection $\{\Psi_i\}$ 
of matter fields defined on $\mathcal{M}$, 
together with a constitutive relation
determining $T_{\mu\nu}$ from $\{g,\Psi_i\}$. These equations and relations are stipulated
by the relevant continuum field theory (electromagnetism, fluid dynamics, etc.)~describing the matter.
The formulation of general relativity represents the culmination of 
the classical field-theoretic world-view where physics is governed by
a closed
system of partial differential equations.

Einstein was led to the system $(\ref{Eeq})$ in 1915,
after a $7$-year struggle to incorporate gravity into his earlier principle of relativity.
In the field-theoretic formulation of the ``Newtonian'' theory, gravity was described
by the \emph{Newtonian potential} $\phi$ satisfying
the \emph{Poisson equation}
\begin{equation}
\label{Poisson}
\triangle \phi =4\pi \mu,
\end{equation}
where $\mu$ denotes the \emph{mass-density} of matter. It is truly remarkable that
the constraints of consistency were so rigid that encorporating gravitation 
required finally a complete reworking of the principle of relativity, leading
to a theory where Newtonian gravity, special relativity and Euclidean geometry each
emerge as
limiting aspects of one dynamic geometrical structure--the Lorentzian metric--naturally living
on a $4$-dimensional spacetime continuum.
A second remarkable aspect of general relativity is that, in contrast to its Newtonian
predecessor, the theory is non-trivial even in the absence of matter.
In that case, we set $T_{\mu\nu}=0$ and the system  $(\ref{Eeq})$
takes the form
\begin{equation}
\label{Evac}
R_{\mu\nu}=0.
\end{equation}
The equations $(\ref{Evac})$ are known as the \emph{Einstein vacuum equations}.
Whereas $(\ref{Poisson})$ is a linear elliptic equation, $(\ref{Evac})$
can be seen to form a closed system of non-linear (but \emph{quasilinear})
wave equations. 
Essentially all of the characteristic features of the dynamics of the Einstein equations
are already present in the study of the vacuum equations $(\ref{Evac})$.

\subsection{Special solutions: Minkowski, Schwarzschild, Kerr}
To understand a theory like general relativity where the fundamental equations 
$(\ref{Evac})$ are nonlinear, the first goal often is to
identify and study important \emph{explicit solutions}, i.e.,
 solutions which can be written in closed form.\footnote{The traditional 
terminology in
general relativity for such solutions is \emph{exact solutions}.}
Much of the early history of general relativity centred around the discovery and interpretation of such
solutions. The simplest explicit solution to the Einstein vacuum equations $(\ref{Evac})$ 
is \emph{Minkowski space}
$\mathbb R^{3+1}$.  The next simplest solution of $(\ref{Evac})$ is the so-called \emph{Schwarzschild solution}, written down~\cite{Schwp} already in 1916.
This is in fact a one-parameter family of solutions $(\mathcal{M},g_M)$, the parameter
$M$ identified with \emph{mass}. See~$(\ref{incords})$ below for the metric form. 
The Schwarzschild family lives as a  subfamily
in a larger two-parameter family of explicit solutions $(\mathcal{M}, g_{M,a})$
 known as the \emph{Kerr solutions}, discussed in Section~\ref{thekerrmetric}.
These were discovered only much later~\cite{kerr} (1963).

When the Schwarzschild solution was first written down in local coordinates, the
necessary concepts to understand its geometry had not yet been developed. 
It took nearly 50 years from the time
when Schwarzschild was first discovered for its global geometry to be 
sufficiently well understood so as to be given a
suitable name: Schwarzschild and Kerr were examples
of what came to be known as \emph{black hole} spacetimes\footnote{This
name is due to  John Wheeler.}. The
Schwarzschild solution also illustrates another feature of the Einstein equations,
namely, the presence of singularities.  

We will spend Section~\ref{introsec} telling the story of the emergence of the black hole notion
and sorting out what the distinct notions of  ``black hole'' and ``singularity'' mean. 
For the purpose of the present introductory section, 
let us take the notion of ``black hole'' as a ``black box''
and make some general remarks on the role
of explicit solutions, whatever might be their properties. These remarks are relevant for any 
physical theory governed by an evolution equation.

\subsection{Dynamics and the stability problem}
Explicit solutions are indeed suggestive as to how
general solutions behave, but only if they are appropriately  ``stable''. 
In general relativity, this notion can in turn only be understood
after the problem of dynamics for $(\ref{Evac})$ has been formulated, that is to say,
the \emph{Cauchy problem}. 

In contrast to other non-linear field theories arising in physics,
in the case of general relativity, even formulating the Cauchy problem requires addressing 
several conceptual issues (e.g.~in what sense is $(\ref{Evac})$ hyperbolic?), and these took
a long time to be correctly sorted out.  
Important advances in this process include
the identification of the harmonic gauge by de Donder~\cite{deDonder},
the existence and uniqueness theorems for general quasilinear wave equations
in the 1930's based on work of Friedrichs, Schauder, Sobolev, Petrovsky, Leray and others,
and  Leray's notion of global hyperbolicity~\cite{leray}.
The well-posedness of the appropriate Cauchy problem for the vacuum equations $(\ref{Evac})$
was finally formulated and proven in celebrated work of
Choquet-Bruhat~\cite{choquet2} (1952) and 
Choquet-Bruhat--Geroch~\cite{chge:givp} (1969). 
See  Appendix~\ref{cauchyproblem} for a concise survey of these developments and
the precise statement of the existence and uniqueness
theorems and some comments on their proof. 

In retrospect, much of the confusion in early discussions
of the Schwarzschild solution can be traced to the lack of  a dynamic framework to understand
the theory.
It is only in 
the context of the language provided by~\cite{chge:givp} that one can then formulate
the dynamical stability problem and examine the relevance of various explicit
solutions.

The stability of Minkowski space
 was first proven in the monumental work of Christodoulou and Klainerman~\cite{book}. 
 See Appendix~\ref{stabsect} for a formulation of this result.
The dynamical stability of the Kerr family as  a family of solutions to the Cauchy 
problem for the Einstein equations, even restricted to parameter values near 
Schwarzschild, i.e.~$|a|\ll M$,\footnote{Note that without symmetry assumptions
one cannot study the stability 
problem for Schwarzschild per se. Only the larger Kerr family can be stable.}
is yet to be understood and poses an important challenge for the mathematical study of
general relativity in the coming years. 
See Section~\ref{formulation} for a formulation of this problem.
In fact, even the most basic
linear properties of waves (e.g.~solutions of $(\ref{fromintro})$)
on Kerr spacetime backgrounds (or more generally, backgrounds near Kerr) have only recently
been understood. 
In view of the wave-like features of the Einstein
 equations $(\ref{Evac})$ (see in particular Appendix~\ref{harmonic}),
this latter problem should be thought of as a prerequisite for understanding
the non-linear stability problem.

\subsection{Outline of the lectures}
The above linear problem will be the main topic of these lectures: We shall here
develop from the beginning the study of the linear homogeneous wave equation
$(\ref{fromintro})$
on fixed black hole spacetime backgrounds $(\mathcal{M},g)$. We have already
referred in passing to the content of some of the later sections. Let us give here
a complete outline:
Section~\ref{introsec} will introduce the black hole concept and the 
Schwarzschild geometry in the wider context of open problems in general relativity.
Section~\ref{S1} will concern the basic boundedness properties for solutions $\psi$
of $(\ref{fromintro})$ on Schwarzschild
exterior backgrounds.
Section~\ref{S2} will concern quantitative decay properties for $\psi$. 
Section~\ref{pertsec} will move on to spacetimes $(\mathcal{M},g)$
 ``near'' Schwarzschild, including
slowly rotating Kerr, 
discussing boundedness and decay properties
for solutions to $(\ref{fromintro})$ on 
such $(\mathcal{M},g)$, and ending in Section~\ref{formulation} with a formulation
of the non-linear stability problem for Kerr, the open problem which in some sense
provides the central motivation for these notes.
Section~\ref{cosmolosec} will consider the analogues of these
problems in spacetimes  with a positive cosmological
constant $\Lambda$,
Section~\ref{epilogue} will give a multiplier-type estimate valid for general
non-degenerate Killing horizons which quantifies the classical red-shift effect.
The importance of the red-shift effect as a stabilising mechanism for the analysis
of waves on black hole backgrounds will be a common theme
throughout these lectures. The notes end with a collection of
open problems in Section~\ref{acik}.

The proof of Theorem~\ref{DT} of Section~\ref{pertsec} as well as
all results of Section~\ref{epilogue} appear in print in these notes for the first time.
The discussion of Section~\ref{newproofofb} as well as the
proof of Theorem~\ref{Schdec} 
have also been streamlined in comparision with previous presentations.
We have given a guide to background literature in 
Sections~\ref{comsecs},~\ref{neocoms},~\ref{neocoms2} and~\ref{neocoms3}.

We have tried to strike a balance in these notes between making the discussion self-containted
and providing the necessary background to appreciate the place of the problem~$(\ref{fromintro})$
in the context of the current state of the art of the Cauchy problem for the Einstein equations 
$(\ref{Eeq})$ or $(\ref{Evac})$  and
the main open problems and conjectures which will guide this subject in the future.
Our solution has been to use the history of the Schwarzschild solution
as a starting point  in Section~\ref{introsec} for a number of digressions
into the study of gravitational collapse, singularities, and the weak and strong cosmic censorship
conjectures,
deferring, however, formal development of various important notions 
relating to Lorentzian geometry and the well-posedness of the Einstein equations  to a series of Appendices. We have already referred to these appendices in the text.
The informal nature of Section~\ref{introsec}
should make it clear   that the discussion is not intended as a proper survey, but merely
to expose the reader to important open problems
in the field and point to some references for further study.
The impatient reader is encouraged 
to move quickly through Section~\ref{introsec} at a first reading.
The problem~$(\ref{fromintro})$ is itself rather self-contained,
requiring only basic analysis and differential geometry, together with a good 
understanding of the black hole spacetimes, in particular, their so-called causal geometry.
The discussion of Section~\ref{introsec} should be more than enough for the latter, although
the reader may want to supplement this with a more general discussion, for instance~\cite{chru}.

These notes accompanied a series of lectures at
a summer school on ``Evolution Equations'' organized by
the Clay Mathematics Institute, June--July 2008. 
The centrality of the evolutionary point of view in general relativity
is often absent from textbook discussions. 
(See however the recent~\cite{rendall}.)
We hope that these notes contribute to the point of view that puts general relativity  at the centre of
modern developments in 
partial differential equations of
evolution.

 %

\section{The Schwarzschild metric and black holes}
\label{introsec}
Practically all concepts in the development of general relativity and much of its history
can be told from the point of view of the Schwarzschild solution.
We now readily associate this solution with
the  black hole concept. It is important to remember, however, that
the Schwarzschild solution was first discovered in a thoroughly classical
astrophysical setting: it was to represent the vacuum region outside a star.
The black hole interpretation--though in some sense inevitable--historically
only emerged much later.

The most efficient way to present the Schwarzschild solution is 
to begin at the onset with  Kruskal's maximal extension as a point of departure. Instead, we 
 shall take advantage of the informal nature of the present
notes to attempt a  more conversational and  ``historical'' presentation of the 
Schwarzschild metric and its interpretation.\footnote{This in no way 
should be considered as a true attempt at the history of the
solution, simply a pedagogical approach to its study. See for example~\cite{eisen}.}
Although certainly not the quickest route, this approach has the advantage
of highlighting the themes which have become so important in the 
subject--in particular, singularities, black holes and their event horizons--with 
the excitement of their step-by-step unravelling from their origin in a model
for the simplest of general relativistic stars. The Schwarzschild solution
will naturally lead to discussions of the Oppenheimer-Snyder collapse
model, the cosmic censorship conjectures, trapped surfaces and Penrose's incompleteness
theorems, and recent work of Christodoulou on trapped surface formation in vacuum 
collapse, and we elaborate on these topics in Sections~\ref{negmassec}--\ref{vaccol}.
(The discussion in these three
last sections was not included in the lectures, however, and is not
necessary for understanding the rest of the notes.)

\subsection{Schwarzschild's stars}
\label{stars}
The most basic self-gravitating
objects are stars. 
In the most primitive stellar models, dating from the 19th century, stars are modeled
by a self-gravitating fluid surrounded by vacuum.
Moreover, to  a first approximation, classically 
stars are \emph{spherically symmetric} and \emph{static}.

It should not be surprising then
that early research on the
Einstein equations  $(\ref{Eeq})$ would address the question of the existence and
structure of general relativistic stars in the new theory.
In view of our above discussion, the most basic  problem is to understand 
spherically symmetric, static metrics,
represented
 in coordinates $(t,r,\theta,\phi)$, such that the spacetime has two regions: In the region
 $r\le R_0$--the interior of the star--the metric should solve a suitable Einstein-matter system
$(\ref{Eeq})$ with appropriate matter, 
and in the region $r\ge R_0$--the exterior of the star--the spacetime should be vacuum, i.e.~the metric
should solve $(\ref{Evac})$.
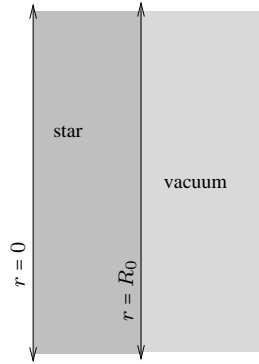
\[
\input{schw0.pstex_t}
\]

This is the problem first addressed by Schwarzschild~\cite{Schwp, Schwp2}, already
in 1916.
Schwarzschild considered the vacuum region first~\cite{Schwp} and
arrived\footnote{As is often the case, the actual history
is more complicated. Schwarzschild
 based his work on an earlier version of Einstein's theory which, while 
obtaining the correct vacuum equations, imposed a condition on admissible
coordinate systems which would in fact
exclude the coordinates of $(\ref{incords})$.
Thus he had to use a rescaled $r$ as a coordinate. Once this condition was removed
from the theory, there is no reason not to take $r$ itself as the coordinate.
It is in this sense that these coordinates can reasonably be
called ``Schwarzschild coordinates''.} at the one-parameter family of solutions:
\begin{equation}
\label{incords}
g=
-\left(1-\frac{2M}r\right)dt^2+\left(1-\frac{2M}r\right)^{-1}dr^2+r^2(d\theta^2+\sin^2\theta\, d\phi^2).
\end{equation}
Every student of this subject should
 explicitly check that this solves $(\ref{Evac})$ ({\bf Exercise}).

In~\cite{Schwp2}, Schwarzschild found interior metrics for the darker shaded region
$r\le R_0$ above. In this region, matter is described by a perfect fluid.
We shall not write down explicitly such metrics here, as this would require a long digression
into fluids,
their equations of state, etc. See~\cite{2phase}.  Suffice it to say here that
the existence of such solutions required that one take the constant $M$ positive, and 
the value $R_0$ marking the boundary of the star always satisfied $R_0>2M$.
The constant $M$ could then be identified with the
total mass of the star as measured by considering the orbits of far-away test
particles.\footnote{Test particles in general relativity follow timelike geodesics of the spacetime
metric. {\bf Exercise}: Explain the statement claimed about far-away test particles.
See also Appendix~\ref{asymptflat}.} 
In fact, for most reasonable matter models, static solutions of the type described
above  only exist under a stronger restriction on $R_0$ (namely $R_0\ge 9M /4$)
now known as the \emph{Buchdahl inequality}. See~\cite{buchdahl, hakan, Buch}.

The restriction on $R_0$ necessary for the existence of Schwarzschild's stars appears
quite fortuitous:
It is manifest from the form $(\ref{incords})$ that
the components of $g$ are singular if the $(t,r,\theta,\phi)$ coordinate
system for the vacuum region is extended to $r=2M$. 
But a natural (if perhaps seemingly of only academic interest) question arises, namely, 
what happens if one does away completely with the star and tries simply to consider 
the expression $(\ref{incords})$ for all values of $r$?
This at first glance would appear to be the problem of understanding the gravitational
field of a ``point particle'' with the particle removed.\footnote{Hence the title of~\cite{Schwp}.}

For much of the history of general relativity,
the degeneration of the metric functions at $r=2M$, when written in these coordinates,
was understood
as meaning that 
the gravitational field should be considered singular there.
This was the famous Schwarzschild ``singularity''.\footnote{Let the reader keep in mind that there
is a good reason for the quotation marks here and for those that follow.} 
 Since ``singularities'' were
considered ``bad'' by most
 pioneers of the theory, various arguments were concocted 
to show that the behaviour of $g$ where $r=2M$ is to be thought
of as ``pathological'', ``unstable'',
``unphysical''
and thus, the solution should not be considered there. 
The constraint on $R_0$ related to the 
Buchdahl inequality seemed to give support  
to this point of view. See also~\cite{Einstein}.

With the benefit of hindsight, we now know that
the interpretation of the previous paragraph is incorrect, on essentially every level:
neither is $r=2M$ a singularity, nor are singularities--which do in fact occur!--necessarily to 
be discarded!
Nor is it true that non-existence of static stars renders the behaviour at $r=2M$--whatever
it is--``unstable''
or ``unphysical''; on the contrary, it was an early hint of 
gravitational collapse!
Let us put aside this hindsight for now and try to discover for ourselves
the geometry and ``true'' singularities hidden in~$(\ref{incords})$, as well as
the correct framework for identifying ``physical'' solutions.
In so doing, we are
retracing in part the steps of early pioneers who studied these issues
without the benefit of the global geometric framework we now have at our disposal.
All the notions referred to above will reveal themselves in the next subsections.

\subsection{Extensions beyond the horizon} 
\label{orextsec}
The fact that the behaviour of the metric at $r=2M$ is not singular,
but simply akin to the well-known breakdown of the coordinates $(\ref{incords})$ at $\theta=0,\pi$ 
(this latter breakdown having never confused anyone\ldots), is 
actually quite easy to see, and there is no better way to appreciate
this than by doing the actual calculations. Let us see how to proceed.

First of all, before even attempting a change of coordinates, the following is
already suggestive:
Consider say a future-directed\footnote{We time-orient the metric by $\partial_t$.
See Appendix~\ref{Lorge}.}
ingoing radial null geodesic.
The image of such a null ray is in fact depicted below:
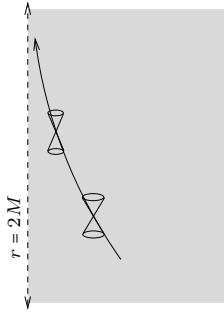
\[
\input{Schw1.pstex_t}
\]
One can compute  that this has finite affine length to the future,
i.e.~these null geodesics are \emph{future-incomplete}, while
 scalar curvature invariants remain bounded as $s\to \infty$.
It is an amusing {\bf exercise} to put oneself in this point of view 
and carry out the above
computations in these coordinates.

Of course, as such the above doesn't show anything.\footnote{Consider for instance
a cone with the vertex removed\ldots}
But it turns out that
indeed the metric can be extended to be defined on a ``bigger'' manifold.
One defines a new coordinate 
\[
t^*=  t+2M\log (r-2M).
\]
This metric then takes the form
\begin{equation}
\label{incords2}
g=-\left(1-\frac{2M}r\right)(dt^*)^2 +\frac{4M}r dt^*\, dr +
\left(1+\frac{2M}r\right) dr^2+r^2 d\sigma_{\mathbb S^2}
\end{equation}
on $r>2M$.
Note that $\frac{\partial}{\partial t^*}=\frac{\partial}{\partial t}$, each interpreted
in its respective coordinate system.
But now $(\ref{incords2})$ can clearly
be defined in the region $r>0$, $-\infty<t^*<\infty$, and, by explicit
computation or better, by analytic continuation, the metric $(\ref{incords2})$
must satisfy $(\ref{Evac})$ for all
$r>0$.

Transformations similar to the above were already known
to Eddington and Lemaitre~\cite{lemaitre} in the early 1930's. 
Nonetheless, from the point of view of that time, it was difficult to 
interpret their significance. The formalisation of the
manifold concept and associated language had not yet become common knowledge
to physicists (or most mathematicians for that matter), 
and in any case, there was no selection principle as to what 
should the underlying manifold $\mathcal{M}$ 
be on which a solution $g$ to $(\ref{Evac})$ should live, or, to put it another way,
the domain of $g$ in $(\ref{Evac})$ is not specified a priori by the theory.
So, even if the solutions $(\ref{incords2})$ exist, how do we know that they are
``physical''?

This problem can in fact only be clarifed in the context of the \emph{Cauchy problem}
for $(\ref{Eeq})$ coupled to appropriate matter.
Once the Cauchy problem for $(\ref{Evac})$ is formulated correctly, then 
one can assign a \emph{unique} spacetime to an appropriate notion of
\emph{initial data set}.  This is the \emph{maximal development}
 of Appendix~\ref{cauchyproblem}.
It is only the initial data set, and the matter model, which can be judged for
``physicality''. One cannot throw away the resulting maximal development just
because one does not like its properties!

From this point of view,  the question of whether the extension $(\ref{incords2})$
was ``physical'' was resolved in 1939 by Oppenheimer and Snyder~\cite{os}.
Specifically, they showed that the extension $(\ref{incords2})$ for $t\ge 0$
arose as a subset of the solution to the Einstein equations coupled
to a reasonable (to a first approximation at least) matter model,  evolving from physically 
plausible initial data. With hindsight, the notion of black hole was born
in that paper.

Had history proceeded differently, we could base our futher discussion on~\cite{os}.
Unfortunately, the model~\cite{os} was ahead of its time.
As mentioned in the introduction, the proper language 
to formulate the Cauchy problem
in general only came in 1969~\cite{chge:givp}.
The interpretation of explicit solutions remained the
main route to understanding the theory.
We will follow thus this route to the black hole concept--via the geometric study of
so-called
 maximally extended Schwarzschild--even though this spacetime is not
to be regarded as ``physical''. It was through the study of this spacetime that the 
relevant notions were first understood and the important Penrose diagrammatic notation
was developed.
We shall return to~\cite{os} only in Section~\ref{OSsection}.

\subsection{The maximal extension of Synge and Kruskal}
Let us for now avoid the question of what the underlying manifold ``should'' be, 
a question whose answer requires physical input (see paragraphs above), 
and simply ask the purely mathematical question of
how big the underlying manifold ``can'' be. This leads to the notion of a ``maximally extended''
solution. In the case of Schwarzschild, this will be a spacetime which, although
 not to be taken as a model for anything \emph{per se}, 
can serve as a reference for the formulation of all important concepts in the subject.

To motivate this notion of ``maximally extended'' solution, 
let us examine our first extension a little more closely.
The light cones can be drawn as follows:
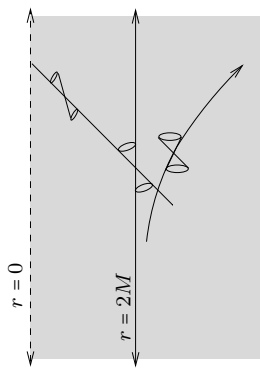
\[
\input{Schw2.pstex_t}
\]

Let us look say at null geodesics.
One can see ({\bf Exercise}) that future directed null geodesics either approach $r=0$ or
are future-complete. In the former case, scalar invariants of the curvature blow up in the limit
as the affine parameter approaches its supremum ({\bf Exercise}). 
The spacetime is thus ``singular'' in this sense.
It  thus follows from the above properties
that the above spacetime is \emph{future null geodesically incomplete}, 
but also \emph{future null geodesically inextendible}
as a $C^2$ Lorentzian metric, i.e.~there does not exist a larger $4$-dimensional
Lorentzian manifold with $C^2$ metric
such that the spacetime above embeds isometrically into the larger one
such that a future null geodesic passes into the
extension.

On the other hand, one can see that
past-directed null geodesics are not all complete, yet no curvature quantity
blows up along them ({\bf Exercise}).
Again, this suggests that something may still be missing!

Synge was the first to consider these issues systematically and construct ``maximal extensions'' 
of the original Schwarzschild metric in a paper~\cite{Synge} of
1950. A more concise approach to
such  a construction was given in a celebrated 1960 paper~\cite{kruskal}
of Kruskal.
Indeed, let $\mathcal{M}$ be the manifold with differentiable structure given by
 $\mathcal{U}\times\mathbb S^2$ where $\mathcal{U}$ is the open subset
 $T^2-R^2<1$ of the $(T,R)$-plane.
 Consider the metric  $g$
\[
g= \frac{32M^3}r e^{-r/2M}(-dT^2+dR^2)+r^2 d\sigma_{\mathbb S}^2
\]
where $r$ is defined implicitly by
\[
T^2-R^2=\left(1-\frac{r}{2M}\right)e^{r/2M}.
\]
The region $\mathcal{U}$ is depicted below:
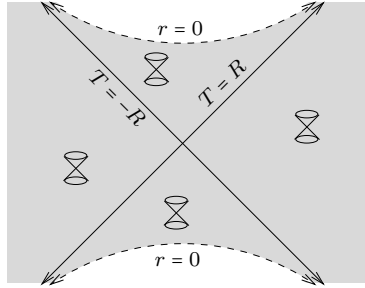
\[
\input{krusk.pstex_t}
\]
This is a spherically symmetric $4$-dimensional Lorentzian manifold satisfying
$(\ref{Evac})$ such that
the original Schwarzschild metric is isometric to the region
$R> |T|$ (where $t$ is given by $\tanh\left(\frac{t}{4M}\right)=T/R$), and our previous
partial extension is isometric to the region $T> -R$ ({\bf Exercise}). It
can be shown now ({\bf Exercise})
that $(\mathcal{M},g)$ is \emph{inextendible as a $C^2$ (in fact $C^0$)
Lorentzian manifold}, that is to say,
if
\[
i: (\mathcal{M},g)\to (\widetilde{\mathcal{M}},\tilde{g})
\]
is an isometric embedding, where $(\widetilde{M}, \tilde{g})$ is a $C^2$ (in fact $C^0$)
$4$-dimensional
Lorentzian manifold, then necessarily $i(\mathcal{M})=\widetilde{\mathcal{M}}$.

The above property defines the sense in which our spacetime is ``maximally'' extended,
and thus, $(\mathcal{M},g)$  is called sometimes \emph{maximally-extended
Schwarzschild}. In later sections,
we will often just call it ``the Schwarzschild solution''.

Note that the form of the metric is such that the light cones are as depicted. Thus,
one can read off much of the causal structure by sight.

It may come as a surprise that in maximally-extended Schwarzschild, there are
\emph{two} regions which are isometric to the original $r>2M$ Schwarzschild region.
Alternatively, a Cauchy surface\footnote{See Appendix~\ref{Lorge}.}
 will have topology $\mathbb S^2\times \mathbb R$ with
\emph{two} asymptotically flat ends. This suggests that 
this spacetime is not to be taken as a physical model. We will discuss this later on.
For now, let us simply try to understand better the global geometry of the metric.

\subsection{The Penrose diagram of Schwarzschild}
\label{pendiag}
There is an even more useful way to represent the above spacetime.
First, let us define null coordinates $U=T-R$, $V=T+R$. These coordinates have infinite
range. We may rescale them by $u=u(U)$, $v=v(V)$ to have finite range.
(Note the freedom in the choice of $u$ and $v$!)
The domain of $(u,v)$ coordinates, when represented in the plane 
where the axes are at $45$ and $135$ degrees with the horizontal,
 is known as a \emph{Penrose diagram} of
Schwarzschild.
Such a Penrose diagram is depicted below\footnote{How can $(u,v)$ be chosen
so that the $r=0$ boundaries
are horizontal lines? ({\bf Exercise})}:
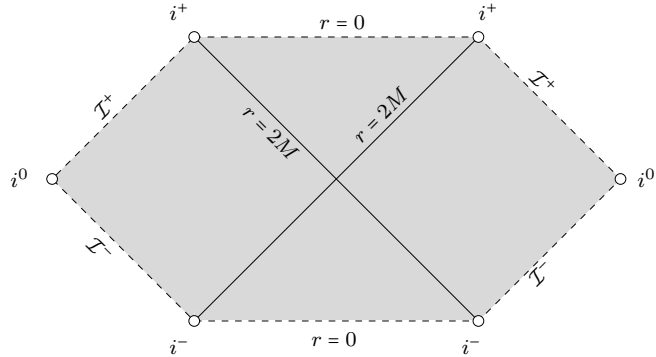
\[
\input{schw3.pstex_t}
\]

In more geometric language, one says that a Penrose diagram
corresponds to the image of a bounded conformal map
\[
\mathcal{M}/{\rm SO}(3)=\mathcal{Q}\to \mathbb R^{1+1},
\]
where one makes the identification $v=t+x$, $u=t-x$ where $(t,x)$ are now the
standard coordinates $\mathbb R^{1+1}$ represented in the standard way on the plane.
We further assume that the map preserves the time orientation, where Minkowski space
is oriented by $\partial_t$. (In our application, this is a fancy way of saying
that $u'(U), v'(V)>0$). It follows that the map preserves the causal 
structure of $\mathcal{Q}$.
In particular, we can ``read off'' the radial null geodesics of $\mathcal{M}$ from the depiction.

Now we may turn to
the boundary induced by the causal embedding. We define $\mathcal{I}^\pm$ to be
the boundary components as depicted.\footnote{Our convention is that open endpoint circles
are not contained in the intervals they bound, and dotted lines are not contained in
the regions they bound, whereas solid lines are.}
 These are characterized geometrically
as follows: $\mathcal{I}^+$ are
limit points of future-directed null rays in $\mathcal{Q}$ along which $r\to\infty$.
Similarly, $\mathcal{I}^-$  are limit points of past-directed
null rays for which $r\to \infty$. We call 
$\mathcal{I}^+$ \emph{future null infinity} and $\mathcal{I}^-$ \emph{past null infinity}.
The remaining boundary components  $i^0$ and
$i^\pm$ depicted
are often given the names \emph{spacelike} infinity and \emph{future (past) timelike} infinity, 
respectively.

In the physical application, it is important to remember that
asymptotically flat\footnote{See Appendix~\ref{asymptflat} for a definition.} spacetimes 
like our $(\mathcal{M},g)$ are not meant
to represent the whole universe\footnote{The
study of that problem is what is known as ``cosmology''. See Section~\ref{cosmolosec}.}, but rather,
 the gravitational field in the vicinity of
an isolated self-gravitating system. $\mathcal{I}^+$ is an idealization of
far away observers who can receive radiation from the system. In this sense,
``we''--as astrophysical observers of stellar collapse, say--are located at $\mathcal{I}^+$.
The ambient causal structure of $\mathbb R^{1+1}$ allows us to talk about
$J^-(p)\cap \mathcal{Q}$ for $p\in \mathcal{I}^+$\footnote{Refer to Appendix~\ref{Lorge} for
$J^\pm$.} and this will lead us to the black hole concept. Therein lies the use of the Penrose
diagram representation.

The  systematic use of the conformal point of view
 to represent the global geometry of spacetimes 
is one of the many great contributions of Penrose to general relativity.
These representations can be traced back to the well-known ``spacetime diagrams''  
of special relativity, promoted especially by Synge~\cite{synspe}. The ``formal'' use of Penrose
diagrams in the sense above goes back to Carter~\cite{carter}, in whose hands
these diagrams became a powerful tool for determining the global structure of all classical
black hole spacetimes. It is hard to overemphasize how important it is for the
student of this subject to become comfortable with these representations.

\subsection{The black hole concept}
\label{concep} 
With Penrose diagram notation, we may now explain the black hole concept.
\subsubsection{The definitions for Schwarzschild}
First an important remark:
In Schwarzschild, the boundary component $\mathcal{I}^+$ enjoys a limiting 
affine completeness. More specifically, normalising
a sequence of ingoing radial null vectors by parallel transport along an outgoing geodesic
meeting $\mathcal{I}^+$,
the affine length of the  null geodesics generated by these vectors, 
parametrized by their parallel
transport (restricted to $J^-(\mathcal{I}^+)$), tends to infinity:
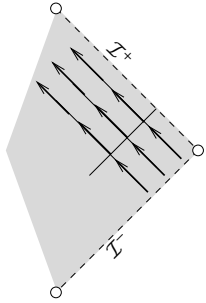
\[
\input{compl.pstex_t}
\]
This has the interpretation that far-away observers in the radiation zone can observe
for all time. (This is in some sense
related to the presence of timelike geodesics near infinity of infinite length, but
the completeness is best formulated with respect to $\mathcal{I}^+$.)
A similar statement clearly holds for $\mathcal{I}^-$.

Given this completeness property, we define now the \emph{black hole} region 
to be $\mathcal{Q}\setminus J^-(\mathcal{I}^+)$,
and the \emph{white hole} region to be  $\mathcal{Q}\setminus J^+(\mathcal{I}^-)$.
Thus, the black hole corresponds to those points of spacetime which cannot
``send signals'' to future null infinity, or, in the physical interpetation, 
to far-away observers who (in view of the completeness property!)
nonetheless can
observe radiation for infinite time.

The future boundary of $J^-(\mathcal{I}^+)$ in $\mathcal{Q}$ (alternatively
characterized as the past boundary
of the black hole region) is a null hypersurface known as the \emph{future event horizon},
and is denoted by $\mathcal{H}^+$. Exchanging past and future, we obtain the
\emph{past event horizon} $\mathcal{H}^-$. In maximal Schwarzschild,
$\{r=2M\}=\mathcal{H}^+\cup \mathcal{H}^-$. 
The subset $J^-(\mathcal{I}^+)\cap J^+(\mathcal{I}^-)$ is known as the \emph{domain of
outer communications}.

\subsubsection{Minkowski space}
\label{MSsec}
Note that in the case of Minkowski space,
$\mathcal{Q}=\mathbb R^{3+1}/{\rm SO(3)}$ is a manifold with boundary since the ${\rm SO}(3)$ action has
a locus of fixed points, the centre of symmetry. 
A Penrose diagram of Minkowski space is easily seen to be:
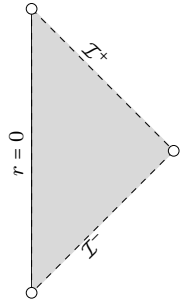
\[
\input{Minko.pstex_t}
\]
Here $\mathcal{I}^+$ and $\mathcal{I}^-$ are characterized as before, and enjoy the
same completeness property as in Schwarzschild.
One reads off immediately that $J^-(\mathcal{I}^+)\cap\mathcal{Q}=\mathcal{Q}$,
i.e.~$\mathbb R^{3+1}$ does not contain a black hole under the above definitions.

\subsubsection{Oppenheimer-Snyder}
\label{OSsection}
Having now the notation of Penrose diagrams, we can concisely describe
the geometry of the Oppenheimer-Snyder
solutions referred to earlier, without giving explicit forms of the metric. 
Like Schwarzschild's original picture of the gravitational field of a spherically
symmetric star, these solutions
involve a region $r\le R_0$ solving $(\ref{Eeq})$ and $r\ge R_0$ satisfying
$(\ref{Evac})$. The matter is described now by a pressureless fluid which is
initially assumed homogeneous in addition to being spherically symmetric.
The assumption of staticity is however dropped, 
and for appropriate initial conditions, it follows that $R_0(t^*)\to 0$
with respect to a suitable time coordinate $t^*$.
(In fact, the Einstein equations can be reduced to an o.d.e.~for $R_0(t^*)$.) 
We say that the star ``collapses''.\footnote{Note that $R_0(t^*)\to0$ does not
mean that the star collapses to ``a point'', merely that the spheres which foliate
the interior of the star shrink to $0$ area. The limiting singular boundary is
 a spacelike hypersurface as depicted.}
A Penrose diagram of such a solution (to the future of a Cauchy hypersurface)
can be seen to be of the form:
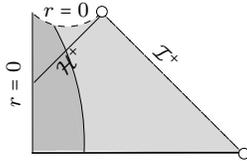
\[
\input{collapse.pstex_t}
\]
The lighter shaded region is isometric to a subset of maximal Schwarzschild,
in fact a subset of the original extension of Section~\ref{orextsec}. In particular, the completeness
property of $\mathcal{I}^+$ holds, and 
as before, we identify the black hole region to be $\mathcal{Q}\setminus
J^-(\mathcal{I}^+)$.

In contrast to maximal Schwarzschild, where the initial configuration is
unphysical (the Cauchy surface has two ends and topology $\mathbb R\times \mathbb S^2$), 
here the initial configuration is entirely plausible: the Cauchy surface is 
topologically $\mathbb R^3$, and its geometry is not far from Euclidean space.
The Oppenheimer-Snyder model~\cite{os}  should be viewed
as the most basic black hole solution arising from physically plausible regular initial
 data.\footnote{\label{heroic}
 Note however the end of Section~\ref{nakedsings}.}

It is traditional in general relativity to ``think'' Oppenheimer-Snyder but
``write'' maximally-extended Schwarzschild. In particular, one often imports
terminology like ``collapse'' in discussing Schwarzschild, and one often
 reformulates our definitions 
 replacing $\mathcal{I}^+$ with one of its connected components,
 that is to say, we will often write $J^-(\mathcal{I}^+)\cap J^+(\mathcal{I}^+)$
 meaning $J^-(\mathcal{I}^+_A)\cap J^+(\mathcal{I}^-_A)$, etc.
In any case, the precise relation between the two solutions should be clear from
the above discussion. In view of Cauchy stability results~\cite{he:lssst},
sufficiently general
theorems about the Cauchy problem on maximal Schwarzschild lead immediately
to such results on Oppenheimer-Snyder. (See for instance the exercise in Section~\ref{discrete}.) 
 One should always keep this relation in mind.

\subsubsection{General definitions?}
\label{general?}
The above definition of black hole for the Schwarzschild metric should be thought of as  a blueprint for 
how to define the notion of black hole region in general. That is to say, to define the black
hole region,
one needs
\begin{enumerate}
\item
some notion of future null infinity $\mathcal{I}^+$, 
\item 
a way of identifying
$J^-(\mathcal{I}^+)$,
and 
\item
some characterization of the ``completeness'' of $\mathcal{I}^+$.\footnote{The characterization
of completeness can be formulated for general asymptotically flat vacuum space times using
the results of~\cite{book}. This formulation is due to Christodoulou~\cite{sings}.
Previous attempts to formalise these notions rested on ``asymptotic simplicity'' and
``weak asymptotic simplicity''. See~\cite{he:lssst}. Although the qualitative picture suggested by
these notions appears plausible, the detailed asymptotic behaviour of solutions
to the Einstein equations turns out to be
much more subtle, and Christodoulou has
proven~\cite{rome}
that  these notions cannot capture even the simplest generic physically interesting systems.}
\end{enumerate}
If $\mathcal{I}^+$ is indeed complete, we can define the black hole region as

\begin{center}
``the complement in $\mathcal{M}$ of $J^-(\mathcal{I}^+)$''.
\end{center}

For spherically symmetric spacetimes arising as solutions of the Cauchy problem
for $(\ref{Eeq})$, 
one can show that there always exists a Penrose diagram,
and thus, a definition can be formalised along precisely these lines (see~\cite{trapped}).
For spacetimes without symmetry, however,
even defining the relevant asymptotic structure so that this structure
is compatible with the theorems one is to prove is a main part of the problem. 
This has
been accomplished definitively only in the case of perturbations of Minkowski space.
In particular, Christodoulou and Klainerman~\cite{book} have shown that 
spacetimes arising from 
perturbations of Minkowski initial data have a complete $\mathcal{I}^+$ in a well defined
sense, whose
past can be identified and is indeed
the whole spacetime. See Appendix~\ref{stabsect}.
That is to say, small perturbations of Minkowski space
cannot form black holes.

\subsection{Birkhoff's theorem}
\label{negmassec}
Formal Penrose diagrams are a powerful tool for understanding the
global causal structure of spherically symmetric spacetimes.
Unfortunately,  however, it turns out that
the study of spherically symmetric \emph{vacuum} spacetimes is
not that rich.
In fact, the Schwarzschild family parametrizes all spherically symmetric
vacuum spacetimes in a sense to be explained in this section.

\subsubsection{Schwarzschild for $M<0$}
Before stating the theorem, recall that in discussing Schwarzschild we have previously restricted to parameter value $M>0$. For the uniqueness statement, we must enlarge the family to 
include all parameter values.

If we set $M=0$ in $(\ref{incords})$, we of course
obtain Minkowski space in spherical polar coordinates.
A suitable maximal extension is Minkowski space as we know it, represented
by the Penrose diagram of Section~\ref{MSsec}.

On the other hand, we may also take $M<0$ in $(\ref{incords})$. 
This is so-called \emph{negative mass
Schwarzschild}. The metric element $(\ref{incords})$
for such $M$ is now regular
for all $r>0$. The limiting singular behaviour of the metric at $r=0$ is in fact essential,
i.e.~one can show that along inextendible incomplete geodesics  the curvature blows
up. Thus, one immediately arrives at a
maximally extended solution which can be seen to have Penrose diagram:
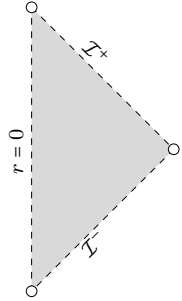
\[
\input{negmass.pstex_t}
\]
Note that in contrast to the case of $\mathbb R^{3+1}$, the boundary $r=0$ is here depicted
by a dotted line denoting (according to our conventions) that it is not part of $\mathcal{Q}$!

\subsubsection{Naked singularities and weak cosmic censorship}
\label{nakedsings}
The above spacetime is interpreted as having a ``naked singularity''. 
The traditional way
of describing this in the physics literature is to remark that
the ``singularity'' $\mathcal{B}=\{r=0\}$ is ``visible'' to $\mathcal{I}^+$, i.e.,
$J^-(\mathcal{I}^+)\cap\mathcal{B}\ne\emptyset$.
From the point of view of the Cauchy problem, however, this characterization
is meaningless because the above maximal extension is not globally
hyperbolic, i.e.~it is not uniquely characterized by an appropriate notion of initial 
data.\footnote{See Appendix~\ref{Lorge} for the definition of global hyperbolicity.}
From the point of view of the Cauchy problem, one must not consider 
maximal extensions but the \emph{maximal Cauchy development of initial data}, 
which by definition is 
 globally hyperbolic (see Theorem~\ref{Maxdev} of Appendix~\ref{cauchyproblem}). Considering an inextendible
spacelike hypersurface $\Sigma$ as a Cauchy surface, 
the maximal Cauchy development of $\Sigma$
would be the darker
shaded region depicted below:
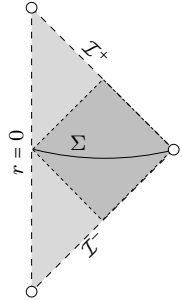
\[
\input{Cauchy.pstex_t}
\]
The proper characterization of ``having a naked singularity'', from the point of view
of the darker shaded spacetime, is that its
$\mathcal{I}^+$ is
incomplete. Of course, this example does not say anything about the dynamic formation
of naked singularities, because the inital data hypersurface $\Sigma$ is already
in some sense ``singular'', for instance, it is geodesically incomplete, and the curvature
blows up along incomplete geodesics.
The dynamic formation of a naked singularity from regular, complete initial data
would be pictured by:
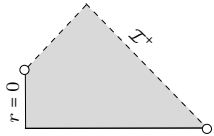
\[
\input{nsing.pstex_t}
\]
where we are to understand also in the above that $\mathcal{I}^+$ is incomplete. 
The conjecture that for generic asympotically flat\footnote{See Appendix~\ref{asymptflat}
for a formulation of this notion. Note that asymptotically flat data are in particular complete.} initial data for ``reasonable'' Einstein-matter systems,
the maximal Cauchy development ``possesses a complete $\mathcal{I}^+$'' is known as
\emph{weak cosmic censorship}.\footnote{This conjecture is originally due to
Penrose~\cite{Penr2}. The present formulation is taken from Christodoulou~\cite{sings}.}
 
 In light of the above conjecture, the story of
the Oppenheimer-Snyder solution and its role in the emergence of the black hole
concept does
have an interesting epilogue. Recall that in the Oppenheimer-Snyder
solutions, the region $r\le R_0$, in addition to being spherically symmetric, is homogeneous.
It turns out that by considering spherically symmetric initial data for which
the ``star'' is no longer homogeneous, Christodoulou has proven that
one can arrive at spacetimes for which
``naked singularities'' form~\cite{dust!} with Penrose diagram as above and with $\mathcal{I}^+$
incomplete.
Moreover, it is shown
in~\cite{dust!} that this occurs for \emph{an open subset}
of initial data within spherical symmetry, with respect to  a suitable topology on the set
of spherically symmetric initial data. Thus, weak cosmic 
censorship is violated in  this model, at least if the conjecture is restricted to
spherically symmetric data.

The fact that in the Oppenheimer-Snyder solutions black holes formed appears thus
to be a rather fortuitous accident! Nonetheless, we should note that
the failure of weak cosmic censorship in this context is believed to be 
due to the inappropriateness of the  pressureless model, not as indicative of actual phenomena. 
Hence, the restriction on  the matter 
model to be ``reasonable'' in the formulation of the conjecture.
In a remarkable series of papers,
Christodoulou~\cite{wcs, sings} has shown weak cosmic censorship to be true for
the Einstein-scalar field system under spherical symmetry.
On the other hand, he has also shown~\cite{examples}
that the assumption of genericity is still necessary by
explicitly constructing solutions  of this system
with incomplete $\mathcal{I}^+$ and Penrose diagram
as depicted above.\footnote{The discovery~\cite{examples} of these naked singularities led
to the discovery of so-called critical collapse phenomena~\cite{chopt} which
has since become a popular topic of investigation~\cite{gund}.}

\subsubsection{Birkhoff's theorem}
Let us understand now by ``Schwarzschild solution with parameter $M$'' (where $M\in \mathbb R$) 
the maximally extended Schwarzschild metrics described above.

We have the so-called \emph{Birkhoff's theorem}:
\begin{theorem}
\label{Birkhoff}
Let $(\mathcal{M},g)$ be a spherically symmetric solution to the vacuum equations
$(\ref{Evac})$. 
Then it is locally
isometric to a Schwarzschild solution with parameter $M$, for some $M\in\mathbb R$.
\end{theorem}
In particular, spherically symmetric solutions to $(\ref{Evac})$
possess an additional Killing field not in the Lie algebra ${\rm so}(3)$.
({\bf Exercise}: Prove Theorem~\ref{Birkhoff}. Formulate and prove a global version
of the result.)

\subsubsection{Higher dimensions}
\label{higherhigherdim}
In $3+1$ dimensions, spherical symmetry is the only symmetry assumption compatible
with asymptotic flatness (see Appendix~\ref{asymptflat}), such that moreover the symmetry group acts
transitively on $2$-dimensional orbits. Thus,
Birkhoff's theorem means that  vacuum gravitational collapse cannot be studied
in a $1+1$ dimensional setting by imposing symmetry. The simplest models
for dynamic gravitational collapse thus necessarily involve matter, as in the
Oppenheimer-Snyder model~\cite{os} or the Einstein-scalar field system studied 
by Christodoulou \cite{forma,wcs}

Moving, however, to
$4+1$ dimensions, asymptotically flat manifolds can admit a more
general $SU(2)$ symmetry acting transitively on $3$-dimensional group orbits. The Einstein
vacuum equations $(\ref{Evac})$
under this symmetry
admit $2$ dynamical degrees of freedom and can be written as a
nonlinear system on a $1+1$-dimensional Lorentzian quotient $\mathcal{Q}=
\mathcal{M}/SU(2)$, where 
the dynamical degrees of freedom of the metric are reflected by two 
nonlinear scalar fields on $\mathcal{Q}$. 
This symmetry--known as ``Triaxial Bianchi IX''--was first identified by Bizon, Chmaj and Schmidt~\cite{bizon1, bizon2}~who derived the equations on $\mathcal{Q}$ and studied the resulting
system numerically. 
The symmetry includes spherical symmetry as a special case,
and thus, is admitted in particular by
$4+1$-dimensional Schwarzschild\footnote{{\bf Exercise}: Work out explicitly the higher dimensional
analogue of the Schwarzschild solution for all dimensions.}.
The 
nonlinear stability of the Schwarzschild family  as solutions of the vacuum equations
$(\ref{Evac})$
can then be studied--within the class of Triaxial Bianchi IX initial data--as
a $1+1$ dimensional problem.  Asymptotic stability for the 
Schwarzschild spacetime
in this setting has been recently shown
in the thesis of Holzegel~\cite{kostakis, kostakis1, kostakis2},
adapting vector field multiplier estimates 
similar to Section~\ref{S2} to a situation where the metric is not known a priori. 
The construction of the relevant mutipliers is then quite subtle, as they must be normalised
``from the future'' in a bootstrap setting.
The thesis~\cite{kostakis} is a good reference
for understanding the relation of the linear theory to the non-linear black hole stability problem.
See also Open problem~\ref{kostakis;} in Section~\ref{notlinear}.

\subsection{Geodesic incompleteness and ``singularities''}
\label{trapsec}
Is the picture of gravitational collapse as exhibited by Schwarzschild (or better,
Oppenheimer-Snyder) stable? This question is behind the later chapters in the notes,
where essentially the considerations hope to be part of a future understanding of the stability
of the exterior region up to the event horizon, i.e.~the closure of the past of null infinity to the
future of a Cauchy surface. (See Section~\ref{formulation} for a formulation of this open
problem.)
What is remarkable, however, is that there is a feature of Schwarzschild which can easily be shown
to be ``stable'', without understanding the p.d.e.~aspects of $(\ref{Eeq})$: its geodesic incompleteness.

\subsubsection{Trapped surfaces}
First a definition:
Let $(\mathcal{M},g)$ be a time-oriented Lorentzian manifold, and $S$ a closed 
spacelike $2$-surface.
For any point $p\in S$, we may define two null mean curvatures $\tr \chi$ and $\tr \bar\chi$,
corresponding to the two future-directed null vectors $n(x)$, $\bar n(x)$, where
$n$, $\bar n$ are normal to $S$ at $x$. 
We say that $S$ is \emph{trapped} if $\tr \chi<0$, $\tr \bar\chi<0$.

{\bf Exercise}: Show that points $p\in \mathcal{Q}\setminus {\rm clos}(J^-(\mathcal{I}^+))$ 
correspond to trapped surfaces of $\mathcal{M}$. 
Can there be other trapped surfaces?  (Refer
also for instance to~\cite{senov}.)

\subsubsection{Penrose's incompleteness theorem}
\begin{theorem} (Penrose 1965~\cite{Penr})
\label{incthe}
Let $(\mathcal{M},g)$ be globally hyperbolic\footnote{See Appendix~\ref{Lorge}.} 
with non-compact Cauchy
surface $\Sigma$, where $g$ is a $C^2$ metric, and let 
\begin{equation}
\label{nulconv}
R_{\mu\nu}V^\mu V^\nu\ge0
\end{equation}
for all null vectors $V$.
Then if $\mathcal{M}$ contains a closed trapped two-surface $S$, it follows
that $(\mathcal{M},g)$
is future causally geodesically incomplete. 
\end{theorem}
This is the celebrated 
\emph{Penrose incompleteness theorem}. 

Note that solutions of the Einstein vacuum equations $(\ref{Evac})$ satisfy $(\ref{nulconv})$.
(Inequality $(\ref{nulconv})$, known as the null convergence condition, 
is also satisfied for
solutions to the Einstein equations $(\ref{Eeq})$ coupled to most
plausible matter models, specifically, if the energy momentum tensor $T_{\mu\nu}$ satisfies
$T_{\mu\nu}V^\mu V^\nu\ge 0$ for all null $V^\mu$.) 
On the other hand, by definition, the unique solution
to the Cauchy problem (the so-called \emph{maximal Cauchy development of initial data})
is globally hyperbolic (see Appendix~\ref{maxdevy}). 
Thus, the theorem applies to the maximal development
of (say) asymptotically flat (see Appendix~\ref{asymptflat})
vacuum initial data containing a trapped surface. 
Note finally that by Cauchy stability~\cite{he:lssst}, 
the presence of a trapped surface
in $\mathcal{M}$ is clearly ``stable'' to perturbation of initial data.

From the point of view of gravitational collapse, it is more appropriate to define
a slightly different notion of trapped. We restrict to $S\subset \Sigma$ a Cauchy surface
such that $S$ bounds a disc in $\Sigma$. We then can define a unique outward 
null vector field $n$ along $S$, and we say that $S$ is trapped if $\tr \chi<0$
and antitrapped\footnote{Note that there exist other conventions
in the literature for this terminology. See~\cite{senov}.}
 if $\tr \bar\chi<0$, where $\tr \bar\chi$ denotes the mean curvature
with respect to a conjugate ``inward'' null vector field.
The analogue of Penrose's incompleteness theorem holds under this definition.
One may also prove the interesting result
that antitrapped surface cannot not form if they are not present
initially.
See~\cite{notes}.

Note finally that there
are related incompleteness statements due to Penrose and Hawking~\cite{he:lssst}
relevant in cosmological (see Section~\ref{cosmolosec}) settings.

\subsubsection{``Singularities'' and strong cosmic censorship}
\label{experience}
Following~\cite{notes}, we have 
called Theorem~\ref{incthe} an ``incompleteness theorem'' and not a ``singularity theorem''. This is of course an issue of semantics, but
let us further discuss this point briefly as it may serve to clarify various issues. The term 
``singularity'' has had a tortuous history in the context of general relativity. As we have seen, its
first appearance was 
to describe something that turned out not to be a singularity at all--the ``Schwarzschild 
singularity''. It was later realised that behaviour which
could indeed reasonably be described
by the word ``singularity'' did in fact
occur in solutions, as exemplified by the $r=0$ singular ``boundary''
of Schwarzschild towards which curvature scalars blow up. 
The presense of this singular behaviour ``coincides'' in Schwarzschild with
the fact that the spacetime is future causally geodesically incomplete--in fact,
the curvature blows up along all incomplete causal geodesics. 
In view of the fact that it is the incompleteness property which can be inferred from Theorem~\ref{incthe},
it was tempting to redefine ``singularity'' as geodesic incompleteness (see~\cite{he:lssst}) 
and to call
Theorem~\ref{incthe} a ``singularity theorem''.

This is of course a perfectly valid point of view. But is it correct then 
to associate the incompleteness of Theorem~\ref{incthe} to ``singularity'' in the
sense of ``breakdown'' of the metric? Breakdown of the metric
is most easily understood with curvature blowup as above, but more generally, it is captured
by the notion of ``inextendibility'' of the Lorentzian manifold in some regularity
class. We have already remarked that maximally-extended Schwarzschild is inextendible
in the strongest of senses, i.e.~as 
a $C^0$ Lorentzian metric. It turns out, however, that the statement of
Theorem~\ref{incthe}, even when applied to the maximal development of
complete initial data for $(\ref{Evac})$, is compatible with the solution being 
extendible as a $C^\infty$ Lorentzian metric such that \emph{every}
 incomplete causal geodesic of the original spacetime enter the extension!
This is in fact what happens in the case of Kerr initial data. (See Section~\ref{thekerrmetric}
for a discussion of the Kerr metric.) The reason that
the existence of such extensions does not contradict the ``maximality'' of the ``maximal development''
is that these extensions fail to be globally hyperbolic, while the ``maximal development'' is
``maximal'' in the class of globally hyperbolic spacetimes (see Theorem~\ref{Maxdev} of
Appendix~\ref{cauchyproblem}). In the context of Kerr initial data,
Theorem~\ref{incthe} is thus not saying that
breakdown of the metric occurs,
merely that globally hyperbolicity breaks down, and thus further extensions 
cease to be predictable from initial data.\footnote{Further
confusion can arise from the fact that ``maximal extensions'' of Kerr constructed with
the help of analyticity are still geodesically incomplete and inextendible, in particular, with
the curvature blowing up along all incomplete causal geodesics. Thus, one often
talks of the ``singularities'' of Kerr, referring to the ideal singular boundaries one can attach
to such extensions. One must remember, however, that these
extensions are of no relevance from the point of view of the Cauchy problem, and
in any case, their singular behaviour in principle has nothing to do with Theorem~\ref{incthe}.}
 
 A similar phenomenon is exhibited by the Reissner-Nordstr\"om solution
 of the Einstein-Maxwell equations~\cite{he:lssst}, which, unlike
 Kerr, is spherically symmetric and thus admits a Penrose diagram representation:
 \[
\input{kerr2.pstex_t}
 \]
What is drawn above is the maximal development of $\Sigma$. The spacetime
is future causally geodesically incomplete, but can be extended smoothly
to a $(\tilde{\mathcal{M}},\tilde g)$ such that all inextendible geodesics leave the original
spacetime. The boundary of $(\mathcal{M},g)$ in the extension corresponds to $\mathcal{CH}^+$
above. Such boundaries are known as \emph{Cauchy horizons}.

The \emph{strong cosmic censorship conjecture} says that  the maximal development
of \emph{generic} asymptotically flat initial data for the vacuum Einstein equations is inextendible
as a suitably regular Lorentzian metric.\footnote{As with weak cosmic censorship,
the original formulation of this conjecture
is due to Penrose~\cite{Penr3}. The formulation given here is from~\cite{sings}.
Related formulations are given in~\cite{SCCc, monc}.
One can also pose the conjecture for compact initial data, and for various Einstein-matter systems.
It should be emphasized that ``strong cosmic censorship'' does not imply ``weak cosmic censorship''.
For instance, one can imagine a spacetime with Penrose diagram as in the last diagram of
Section~\ref{nakedsings}, 
with incomplete $\mathcal{I}^+$, but
still  inextendible across the null ``boundary''
emerging from the centre.}
One can view this conjecture as saying that
whenever one has geodesic incompleteness, it is due to breakdown of 
the metric in the sense discussed above. (In view of the 
above comments, for this conjecture to be true, the behaviour of the Kerr metric described
above would have to be \emph{unstable} to perturbation.\footnote{Note that the instability
concerns a region ``far inside'' the black hole interior. The black hole exterior
is expected to be stable (as in 
the formulation of Section~\ref{formulation}), hence these notes. See~\cite{the, cbh} for the resolution
of a spherically symmetric version of this problem, where the role of the Kerr metric is
played by Reissner-Nordstr\"om metrics.})
Thus, if by the term ``singularity'' one wants to suggest  ``breakdown of the metric'',
it is only a positive resolution of the strong cosmic censorship conjecture that would
in particular
(generically) make Theorem~\ref{incthe} into a true ``singularity theorem''.

\subsection{Christodoulou's work on trapped surface formation in vacuum}
\label{vaccol}
These notes would not be complete without a brief discussion of
the recent breakthrough by Christodoulou~\cite{megalo}
on the understanding of trapped surface formation for the vacuum.

The story begins with Christodoulou's earlier~\cite{forma}, where 
a condition is given ensuring that trapped surfaces form for spherically symmetric
solutions of the Einstein-scalar field system. The condition is that the difference
in so-called Hawking mass $m$ of two concentric spheres on an outgoing null hupersurface
be sufficiently large with respect to the difference in area radius $r$ of the spheres.
This is a surprising result as it shows that trapped surface formation can arise from
initial conditions which are as close to dispersed as possible, in the sense that the 
supremum of the quantity
$2m/r$ can be taken arbitrarily small initially.

The results of~\cite{forma} lead immediately (see for instance~\cite{regpast}) to the existence
of smooth spherically symmetric solutions of the Einstein-scalar field system with Penrose diagram 
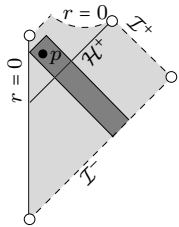
\[
\input{regpast.pstex_t} 
\]
where the point $p$ depicted corresponds to a trapped surface, and the spacetime
is past geodesically complete with a complete past null infinity, whose future is the entire
spacetime, i.e., the spacetime contains no white holes.\footnote{The
triangle ``under'' the darker shaded region can in fact be taken to be Minkowski.}
Thus, black hole formation can arise from spacetimes with a complete 
regular past.\footnote{The singular boundary
in general consists of a possibly empty null component emanating from the regular
centre, and a spacelike component where $r=0$
in the limit and across which the spacetime is inextendible as a $C^0$ Lorentzian metric.
(This boundary could ``bite off'' the top corner of the darker shaded rectangle.)
The null component arising from the centre can be shown to be empty generically after
passing to a slightly less regular class of solutions, for which well-posedness still holds.
See Christodoulou's proof of the cosmic censorship
conjectures~\cite{wcs} for the Einstein-scalar
field system.}

In~\cite{megalo}, Christodoulou constructs 
vacuum solutions by prescribing a characteristic initial value
problem with data on (what will be) $\mathcal{I}^-$. This $\mathcal{I}^-$ is taken to be
past complete, and in fact, the data is taken to be trivial to the past of a sphere on $\mathcal{I}^-$.
Thus, the development will include a region where the metric is Minkowski, corresponding
precisely to the lower lighter shaded triangle above.
It is shown that--as long as the incoming energy per unit solid angle in
all directions\footnote{This is defined in terms of the shear of $\mathcal{I}^-$.} is sufficiently
large in a strip of $\mathcal{I}^-$ right after the trivial part, where
sufficiently large is taken in comparison with the affine length of the generators of $\mathcal{I}^-$--a 
trapped surface  arises in the domain of development of the data restricted to the past of
this strip.
Comparing with the 
spherically symmetric picture above, 
this trapped surface would arise precisely as before
in the analogue of the darker shaded region depicted.

In contrast to the spherically symmetric case, where given the lower triangle,
existence of the solution in the darker shaded region (at least 
as far as trapped surface
formation) follows immediately, 
for vacuum collapse, showing the existence of a sufficiently ``big'' spacetime
is a major difficulty. For this, the results of~\cite{megalo} exploit a
hierarchy in the 
Einstein equations $(\ref{Evac})$ in the context of what is there called the ``short pulse
method''. This method may have many other applications for nonlinear problems.

One could in principle hope to extend~\cite{megalo} to show the formation of black hole
spacetimes in the sense described previously. 
For this, one must first extend the initial data suitably,
for instance so that 
$\mathcal{I}^-$ is complete. If the resulting spacetime can be shown to possess 
a complete future null infinity $\mathcal{I}^+$,
then, since the trapped surface shown to form can be proven (using the methods of
the proof of Theorem~\ref{incthe}) \emph{not} to be in the past
of null infinity, the spacetime will indeed contain a black hole region.\footnote{In spherical
symmetry, the completeness of null infinity follows immediately once a single trapped surface
has formed, for the Einstein equations coupled to a wide class of matter models. 
See for instance~\cite{trapped}.
For vacuum collapse, Christodoulou has formulated a statement on trapped surface
formation that would imply weak cosmic censorship. See~\cite{sings}.} 
Of course, resolution of this problem would
appear comparable in difficulty to
the stability problem for the Kerr family (see the formulation of
Section~\ref{formulation}).

\section{The wave equation on Schwarzschild I: uniform boundedness}
\label{S1}

In the remainder of these lectures, we will concern ourselves solely
with linear wave equations on 
black hole backgrounds, specifically, the scalar linear homogeneous wave equation
$(\ref{fromintro})$. 
As explained in the introduction, the study of the solutions to such equations is
motivated by the stability problem for the black hole spacetimes themselves as solutions
to $(\ref{Evac})$.   The equation $(\ref{fromintro})$ can be viewed as a poor man's linearisation
of $(\ref{Evac})$, neglecting tensorial structure. Other linear problems with 
a much closer relationship to the study of the Einstein equations will be discussed
in Section~\ref{acik}.

\subsection{Preliminaries}
\label{prelims}
Let $(\mathcal{M},g)$ denote (maximally-extended) Schwarzschild with parameter
$M>0$.  Let $\Sigma$ be an arbitrary \emph{Cauchy surface}, that is to say, 
a hypersurface
with the property that every inextendible causal geodesic in $\mathcal{M}$ intersects
$\Sigma$ precisely once. (See Appendix~\ref{Lorge}.)
\begin{proposition}
\label{prelim}
If 
$\uppsi\in H^2_{\rm loc}(\Sigma)$, $\uppsi' \in H^1_{\rm loc}(\Sigma)$, then there is a 
unique $\psi$ with $\psi|_{\mathcal{S}}\in H^2_{\rm loc}(\mathcal{S})$,
$n_{\mathcal{S}}\psi|_{\mathcal{S}} \in H^1_{\rm loc}(\mathcal{S})$,
for all spacelike $\mathcal{S}\subset\mathcal{M}$,
satisfying
\[
\Box_g\psi=0, \qquad \psi|_{\Sigma}=\uppsi, \qquad
n_{\Sigma} \psi|_{\Sigma}=\uppsi',
\]
where $n_{\Sigma}$ denotes the future unit normal of $\Sigma$.
For $m\ge 1$, if $\uppsi\in H^{m+1}_{\rm loc}$, $\uppsi' \in H^m_{\rm loc}$, then 
$\psi|_{\mathcal{S}} \in H^{m+1}_{\rm loc}(\mathcal{S})$,
 $n_{\mathcal{S}}\psi|_{\mathcal{S}} \in H^m_{\rm loc}(\mathcal{S})$.
Moreover, if $\uppsi_1, \uppsi'_1$, and $\uppsi_2, \uppsi'_2$ are as above
and $\uppsi_1=\uppsi_2$, $\uppsi_1'=\uppsi_2'$ in an open set $\mathcal{U}\subset
\Sigma$,
then $\psi_1=\psi_2$ in $\mathcal{M}\setminus J^\pm(\Sigma\setminus {\rm clos}(\mathcal{U}))$.
\end{proposition}

 We will be interested in understanding the behaviour of $\psi$
in the exterior of the black hole and white hole regions, up to and including the
horizons.
It is enough of course to understand the behaviour in the region
\[
\mathcal{D} \doteq {\rm clos}\left(J^-(\mathcal{I}_A^+)\cap J^+(\mathcal{I}_A^-)\right)\cap\mathcal{Q}
\]
where $\mathcal{I}^\pm_A$ denote a pair 
of connected components of $\mathcal{I}^\pm$, respectively,  with a 
common limit point.\footnote{We will sometimes be sloppy with   
distinguishing between $\pi^{-1}(p)$ and $p$, where $\pi:\mathcal{M}\to \mathcal{Q}$
denotes the natural projection, distinguishing
 $J^-(p)$ and $J^-(p)\cap \mathcal{Q}$, etc. The context
should make clear what is meant.}

Moreover, it suffices ({\bf Exercise}: Why?) to assume that 
$\Sigma\cap \mathcal{H}^-=\emptyset$,
and that we are interested in the behaviour in $J^-(\mathcal{I}^+)\cap J^+(\Sigma)$. Note that in this
case, by the domain of dependence property of the above proposition,
we have that    the solution in this region is determined by $\uppsi|_{\mathcal{D}\cap
\Sigma}$, $\uppsi'|_{\mathcal{D}\cap\Sigma}$.
In the case where $\Sigma$ itself is spherically symmetric, then its projection to $\mathcal{Q}$
will look like:
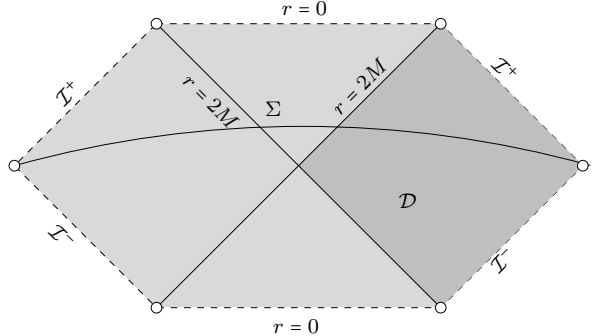
\[
\input{wCauchy.pstex_t}
\]
If $\Sigma$ is not itself spherically symmetric, then its projection to $\mathcal{Q}$ will in general
have open interior. Nonetheless,
we shall always depict $\Sigma$ as above.

\subsection{The Kay--Wald boundedness theorem}
\label{KWsect}
The most basic problem is to obtain uniform boundedness for $\psi$. This is
resolved in the celebrated:
\begin{theorem}
\label{KW}
Let $\psi$, $\uppsi$, $\uppsi'$ be as in Proposition~\ref{prelim}, 
with $\uppsi\in H^{m+1}_{\rm loc}(\Sigma)$, $\uppsi'\in 
H^m_{\rm loc}(\Sigma)$
for a sufficiently high $m $, and such that $\uppsi$, $\uppsi'$ decay suitably at
$i^0$. Then there is a constant $D$ depending on $\uppsi$, $\uppsi'$ such that
\[
|\psi| \le D
\]
in $\mathcal{D}$.
\end{theorem}

The proof of this theorem is due to Wald~\cite{drimos} and Kay--Wald~\cite{kw:lss}.
The ``easy part'' of the proof (Section~\ref{eazy}) is a classic
application of vector field commutators and multipliers, together with elliptic
estimates and the Sobolev inequality.
The main difficulties arise at the horizon, and these are overcome by what
is essentially a clever trick. In this section, we will go through the original argument,
as it is a nice introduction to vector field
multiplier and commutator techniques, as well as to the geometry of Schwarzschild.
We will then point out (Section~\ref{critical}) various disadvantages of the method of proof.
Afterwards, we give a new proof that in fact achieves a stronger result
(Theorem~\ref{boundedn}).
As we shall see, the techniques of this proof
will be essential for future applications.
 
\subsubsection{The Killing fields of Schwarzschild}
Recall the symmetries of $(\mathcal{M},g)$:
$(\mathcal{M},g)$ is spherically symmetric, i.e.~there is a basis of Killing vectors
$\{\Omega_i\}_{i=1}^3$ spanning the Lie algebra ${\rm so}(3)$. These are sometimes
known as \emph{angular momentum operators}.
In addition, there is another Killing field $T$ (equal to $\partial_t$ in the coordinates
$(\ref{incords})$) which is
hypersurface orthogonal and future directed timelike near $i_0$. 
This Killing field is in fact 
timelike everywhere in $J^-(\mathcal{I}^+)\cap J^+(\mathcal{I}^-)$,
becoming null and tangent to the
horizon, vanishing at $\mathcal{H}^+\cap \mathcal{H}^-$.
We say that the Schwarzschild metric in $J^-(\mathcal{I}^+)\cap J^+(\mathcal{I}^-)$
is \emph{static}.
$T$ is spacelike in the black hole and white hole regions.

Note that whereas in Minkowski space $\mathbb R^{3+1}$,
the Killing fields at any point span the tangent space, this is no longer the case for Schwarzschild. We shall
return to this point later.

\subsubsection{The current $J^T$ and its energy estimate}
Let $\varphi_t$ denote the $1$-parameter group of diffeomorphisms generated by
the Killing field $T$. Define $\Sigma_\tau=\varphi_t (\Sigma\cap\mathcal{D})$. We have
that
$\{\Sigma_\tau\}_{\tau\ge 0}$
defines a spacelike foliation of 
\[
\mathcal{R}\doteq \cup_{\tau\ge 0}\Sigma_\tau.
\]
Define
\[
\mathcal{H}^+(0,\tau)\doteq
\mathcal{H}^+\cap J^+(\Sigma_0)\cap J^-(\Sigma_\tau),
\]
and 
\[
\mathcal{R}(0,\tau)\doteq
\cup_{0\le\bar\tau\le \tau} \Sigma_{\bar\tau}.
\]
Let $n^\mu_{\Sigma}$ denote the future directed unit normal of $\Sigma$, and
let  $n^\mu_{\mathcal{H}}$ define a null generator of $\mathcal{H}^+$, and give
$\mathcal{H}^+$ the associated  volume form.\footnote{Recall that for null surfaces,
the definition of a volume form relies on the choice of a normal. All integrals in what follow
will always be with respect to the natural volume form, and in the case of a null hypersurface,
with respect to the volume form related to the given choice of normal. See Appendix~\ref{divthesec}.}

Let $J^T_\mu(\psi)$
 denote the energy current defined by applying the vector field $T$ as a multiplier,
 i.e.
 \[
 J^T_\mu(\psi) = T_{\mu\nu} (\psi) T^\nu= (\partial_\mu\psi\partial_\nu\psi-\frac12g_{\mu\nu}
 \partial^\alpha\psi \partial_\alpha\psi)T^\nu
 \]
 with its associated 
current $K^T(\psi)$,
\[
K^T(\psi)= {}^T \pi^{\mu\nu}T_{\mu\nu}(\psi)= \nabla^\mu J^T_\mu(\psi),
\]
where $T_{\mu\nu}$ denotes the standard energy momentum tensor
of $\psi$ (see 
Appendix~\ref{VFM}). Since $T$ is Killing, and $\nabla^\mu T_{\mu\nu}=0$,
 it follows that $K^T(\psi)=0$,
and the divergence theorem
(See Appendix~\ref{divthesec}) applied to $J_\mu^T$ in the region
$\mathcal{R}(0,\tau)$ yields
\begin{equation}
\label{Enest}
\int_{\Sigma_\tau} J^T_\mu(\psi)n^\mu_{\Sigma_\tau} 
  + \int_{\mathcal{H}^+(0,\tau)} J^T_\mu(\psi)n^\mu_{\mathcal{H}} = 
\int_{\Sigma_0} J^T_\mu(\psi)n^\mu_{\Sigma_0} .
\end{equation}
See
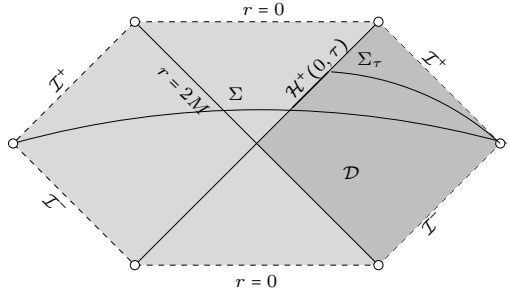
\[
\input{wfol.pstex_t}
\]
Since $T$ is future-directed causal in $\mathcal{D}$, we have
\begin{equation}
\label{ineqs}
J^T_\mu(\psi) n^\mu_{\Sigma} \ge 0, \qquad J^T_\mu(\psi) n^\mu_{\mathcal{H}} \ge 0.
\end{equation}
Let us fix an $r_0>2M$. 
It follows from $(\ref{Enest})$, $(\ref{ineqs})$ that
\[
\int_{\Sigma_\tau\cap\{r\ge r_0\}} J^T_\mu(\psi)n^\mu_{\Sigma_\tau} 
\le
\int_{\Sigma_0} J^T_\mu(\psi)n^\mu_{\Sigma_0}.
\]
As long as $-g(T,n_{\Sigma_0})\le B$ for some constant $B$,\footnote{For definiteness,
one could choose $\Sigma$ to be a surface of constant $t^*$ defined in
Section~\ref{orextsec}, or alternatively,
require that it be of constant $t$ for large $r$.} we have
\[
B(r_0, \Sigma) ((\partial_t \psi)^2 +(\partial_r \psi)^2 + |\nabb\psi|^2)\ge 
J^T_\mu(\psi) n^\mu \ge b(r_0,\Sigma) ((\partial_t \psi)^2 +(\partial_r \psi)^2 + |\nabb\psi|^2).
\]
Here, $|\nabb \psi|^2$ denotes the induced norm on the group orbits
of the $SO(3)$ action, with $\nabb$ the
gradient of the induced metric on the group orbits.
We thus have
\[
\int_{\Sigma_\tau\cap\{r\ge r_0\}} 
(\partial_t \psi)^2 +(\partial_r \psi)^2 + |\nabb\psi|^2 \le
B(r_0,\Sigma) \int_{\Sigma_0} J^T_\mu(\psi)n^\mu_{\Sigma_0} .
\]

\subsubsection{$T$ as a commutator and pointwise estimates away from the horizon}
\label{eazy}
We may now commute the equation with $T$ (See Appendix~\ref{commutat}), 
i.e., since $[\Box, T]=0$, 
if $\Box_g\psi=0$ then $\Box_g (T\psi)=0$. We thus obtain an estimate
\begin{equation}
\label{oresti}
\int_{\Sigma_\tau\cap\{r\ge r_0\}} 
(\partial_t^2 \psi)^2 +(\partial_r\partial_t \psi)^2 + |\nabb\partial_t\psi|^2 \le
B (r_0,\Sigma)\int_{\Sigma_0} J^T_\mu(T\psi)n^\mu_{\Sigma_0} .
\end{equation}
{\bf 
Exercise}: By elliptic estimates and a Sobolev estimate show that
if $\uppsi(x)\to 0$ as $x\to i_0$, then $(\ref{oresti})$ implies that
for $r\ge r_0$,
\begin{equation}
\label{ellipt}
|\psi|^2 \le B(r_0,\Sigma)\left(\int_{\Sigma_0}J^T_\mu(\psi) n^\mu_{\Sigma_0} +\int_{\Sigma_0}J^T_\mu(T\psi) n^\mu_{\Sigma_0} \right),
\end{equation}
for solutions $\psi$ of $\Box_g\psi=0$.

The right hand side of $(\ref{ellipt})$ is finite under the assumptions of Theorem~\ref{KW},
for $m=1$.
Thus, proving the estimate of Theorem~\ref{KW} away from the horizon poses no
difficulty. The difficulty of Theorem~\ref{KW} is obtaining estimates which hold
up to the horizon.

{\bf Remark}: The above argument via elliptic estimates 
clearly also holds for Minkowski space.
But in that case, there is
an alternative ``easier'' argument,
namely, to commute with all translations.\footnote{Easier, but not necessarily better\ldots}
We see thus already that the lack of Killing fields in Schwarzschild makes things
more difficult. We shall again return to this point later.

\subsubsection{Degeneration at the horizon}
\label{whowasthatdegenerate?}
As one takes $r_0\to 2M$, the constant $B(r_0,\Sigma)$ provided by
the estimate $(\ref{ellipt})$ blows up. This is precisely because 
$T$ becomes null on $\mathcal{H}^+$
and thus its control over derivatives of $\psi$ degenerates.
Thus, one cannot prove uniform boundedness holding up to the horizon by the above.

Let us examine more carefully this degeneration on various hypersurfaces.

On $\Sigma_\tau$, we have only
\begin{equation}
\label{YDZ}
J^T_\mu(\psi) n^\mu_{\Sigma_\tau}  \ge  
B(\Sigma_\tau)((\partial_{t^*} \psi)^2 + (1-2M/r) (\partial_r \psi)^2 + |\nabb\psi|^2).
\end{equation}
We see the degeneration in the presence of the factor $(1-2M/r)$.
Note that ({\bf Exercise}) $1-2M/r$ vanishes to first order on $\mathcal{H}^+\setminus \mathcal{H}^-$.
Alternatively, one can examine the flux on the horizon $\mathcal{H}^+$ itself.
For definiteness, let us choose $n_{\mathcal{H}^+}= T$ in $\mathcal{R}\cap \mathcal{H}^+$. We have
\begin{equation}
\label{YDZYDZ}
J^T_\mu(\psi) T^\mu  = (T\psi)^2.
\end{equation}
Comparing with the analogous computation on a null cone in Minkowski space, one sees that 
a term $|\nabb\psi|^2$ is ``missing''.

Are estimates of the terms $(\ref{YDZ})$, $(\ref{YDZYDZ})$
enough to control $\psi$? It is a good idea to play with these
estimates on your own, allowing yourself to commute the equation
with $T$ and $\Omega_i$ to obtain higher order estimates. 
{\bf Exercise}: Why does this not lead to an estimate as in~(\ref{ellipt})?

It turns out that there is a way around this problem and 
the degeneration on the horizon is suggestive.
For suppose  there existed a $\tilde\psi$ such that
\begin{equation}
\label{suchthat}
\Box_g\tilde\psi=0,\qquad T\tilde\psi=\psi.
\end{equation}
Let us see immediately how one can obtain estimates on the horizon itself.
For this,
we note that
\[
J^T_\mu(\tilde{\psi})T^\mu  +J^T_\mu(\psi)T^\mu = \psi^2+(T\psi)^2.
\]
Commuting
now with the whole Lie algebra of isometries, we obtain
\begin{align*}
J^T_\mu(\tilde{\psi})T^\mu  +J^T_\mu(\psi)T^\mu 
+ \sum_i J^T_\mu(\Omega_i\tilde\psi)T^\mu +J^T_\mu(T\psi)T^\mu \cdots\\
 = \psi^2+(T\psi)^2
+\sum_i (\Omega_i\psi)^2+ (T^2 \psi)^2 + \cdots.
\end{align*}
Clearly, by a Sobolev estimate applied on the horizon, together
with the estimate
\[
\int_{\mathcal{H}^+\cap\mathcal{R}}
J^T_\mu(\Gamma^{(\alpha)} \tilde\psi) n^\mu_{\mathcal{H}}
\le
\int_{\Sigma_0}J^T_\mu(\Gamma^{(\alpha)} \tilde\psi) n^\mu_{\Sigma_0}
\]
for $\Gamma=T,\Omega_i$ (here $(\alpha)$ denotes a multi-index of arbitrary order),
we would obtain
\begin{equation}
\label{onthehor}
|\psi|^2\le B \sum_{\Gamma=T,\Omega_i}\sum_{|(\alpha)|\le 2}
\int_{\Sigma_0}J^T_\mu(\Gamma^{(\alpha)} \tilde\psi) n^\mu_{\Sigma_0}
\end{equation}
on $\mathcal{H}^+\cap\mathcal{R}$.

It turns out that the estimate $(\ref{onthehor})$ can be extended to points not on the horizon by
considering $t=c$ surfaces. 
Note that these hypersurfaces all meet at $\mathcal{H}^+\cap \mathcal{H}^-$.
It is an informative calculation to examine the nature
of the degeneration of estimates on such hypersurfaces because it is of a double nature, since,
in addition to $T$ becoming null,
the limit of (subsets of) these spacelike hypersurfaces approaches the null horizon $\mathcal{H}^+$.
We leave the details as an {\bf exercise}.

\subsubsection{Inverting an elliptic operator}
\label{invertin}
So can a $\tilde\psi$ satisfying $(\ref{suchthat})$ actually be constructed?
We have
\begin{proposition}
\label{existof}
Suppose $m$ is sufficiently high, $\uppsi$, $\uppsi'$ decay suitably at $i^0$,
and $\psi|_{\mathcal{H}^+\cap\mathcal{H}^-}=0$, $\Xi \psi|_{\mathcal{H}^+\cap\mathcal{H}^-}=0$
for some spherically symmetric timelike vector field $\Xi$ defined
along $\mathcal{H}^+\cap\mathcal{H}^-$.
Then there exists a $\tilde\psi$ satisfying $\Box_g \tilde\psi=0$ with 
$T\tilde\psi =\psi$ in $\mathcal{D}$,
and moreover, the right hand side of $(\ref{onthehor})$ is finite.
\end{proposition}
Formally, one sees that on $t=c$ say, if we let $\bar{g}$ denote the induced
Riemannian metric, and if we impose
initial data
\[
T \tilde\psi|_{t=c} = \psi,
\]
\[
\tilde{\psi}|_{t=c} = \triangle_{(1-2M/r)^{-1}\bar g}^{-1} T\psi,
\]
and let $\tilde\psi$ solve the wave equation with this data,
then
\[
T\tilde{\psi} = \psi
\]
as desired.

So to use the above, it suffices to ask whether the initial data for $\tilde\psi$ above can be constructed
and have sufficient regularity so as for the right hand side of $(\ref{onthehor})$ to be defined.
To impose the first condition, since $T=0$ along $\mathcal{H}^+\cap\mathcal{H}^-$,
one must have that $\psi$ vanish there to some order. 
For the second condition, note first that
the metric
$(1-2M/r)^{-1}\bar g$ has an asymptotically hyperbolic end and an asymptotically
flat end. Thus,
to  construct $ \triangle_{(1-2M/r)^{-1}\bar g}^{-1} T\psi$ suitably well-behaved\footnote{so that
we may apply to this quantity the arguments of 
Section \ref{whowasthatdegenerate?}.}, one must 
have that $T\psi$ decays appropriately towards the ends.
We leave to the reader the task of verifying that the assumptions of the Proposition are
sufficient.

\subsubsection{The discrete isometry}
\label{discrete}
Proposition~\ref{existof}, together with estimates $(\ref{onthehor})$ and  $(\ref{ellipt})$,
yield the proof of Theorem~\ref{KW} in the special case that the conditions of
Proposition~\ref{existof} happen to be satisfied. In the original paper of Wald~\cite{drimos},
one took $\Sigma_0$ to coincide with $t=0$ and
restricted to data $\uppsi$, $\uppsi'$ which were supported in a compact region not containing
$\mathcal{H}^+\cap\mathcal{H}^-$. Clearly, however, this is a deficiency, as general 
solutions will be supported in $\mathcal{H}^+\cap\mathcal{H}^-$. (See also the
last exercise below.)

It turns out, however, that one can overcome the restriction on the support by
the following trick: Note that the previous proposition produces a $\tilde\psi$ such that
$T\tilde\psi =\psi$ on all of $\mathcal{D}$. We only require however that $T\tilde\psi=\psi$
on $\mathcal{R}$.
The idea is to
define a new $\bar\uppsi$, $\bar\uppsi'$ on $\Sigma$, such that
$\bar\uppsi=\uppsi$, $\bar\uppsi'=\uppsi'$ on $\Sigma_0$ and, denoting
by $\bar\psi$ the solution to the Cauchy problem with the new data,
$\bar\psi|_{\mathcal{H}^+\cap\mathcal{H}^-}=0$, 
$\Xi \bar\psi|_{\mathcal{H}^+\cap\mathcal{H}^-}=0$.
By the previous proposition and the 
 domain of dependence property of Proposition~\ref{prelim}, we will have indeed constructed a
$\tilde{\psi}$ with $T\tilde\psi=\psi$ in $\mathcal{R}$ for which the right hand side
of $(\ref{onthehor})$ is finite.

Remark that Schwarzschild admits a discrete symmetry generated by
the map $X\to -X$ in Kruskal coordinates. Define $\bar\uppsi$, $\bar\uppsi'$ so that
$\bar\uppsi(X, \cdot)=-\bar\uppsi(-X,\cdot)$, $\bar\uppsi'(X,\cdot)=-\bar\uppsi'(-X,\cdot)$.
\begin{proposition}
Under the above assumptions, it follows that
\[
\bar\psi(X,\cdot)=-\bar\psi(-X,\cdot).
\]
\end{proposition}
The proof of the above is left as an exercise in preservation of symmetry
for solutions of the wave equation.
It follows immediately that
\[
\bar\psi|_{\mathcal{H}^+\cap\mathcal{H}^-}=0
\]
and that
\[
\partial_U \bar\psi = -\partial_V \bar\psi,
\]
and thus $(\partial_U+\partial_V )\bar\psi=0$.
In view of the above remarks and Proposition~\ref{existof}
with $\Xi=\partial_U+\partial_V$,
we have shown the full statement of Theorem~\ref{KW}.

{\bf Exercise}: Work out explicit regularity assumptions and
quantitative dependence on initial data in Theorem~\ref{KW}, describing in particular
decay assumptions necessary at $i_0$.

{\bf Exercise}: Prove the analogue of Theorem~\ref{KW} on the Oppenheimer-Snyder
spacetime discussed previously. \emph{Hint: One need not know the explicit
form of the metric, the statement given about the Penrose diagram suffices.}
Convince yourself that the original restricted version of Theorem~\ref{KW} due
to Wald~\cite{drimos}, where the support of $\psi$ is restricted near $\mathcal{H}^+\cap\mathcal{H}^-$,
is not sufficient to yield this result.

\subsubsection{Remarks}
\label{critical}
The clever proof described above successfully obtains pointwise boundedness for $\psi$
up to the horizon $\mathcal{H}^+$. Does this really close the book,
however, on the boundedness question? From various 
points of view, it may be desirable to go further. 
\begin{enumerate}
\item
\label{eventhough}
Even though one obtains the ``correct'' pointwise result, one does
not obtain boundedness at the horizon for the energy measured by a local observer,
that is to say, bounds for 
\[
\int_{\Sigma_\tau} J^{n_{\Sigma_\tau}}_\mu(\psi) 
n^\mu_{\Sigma_\tau}.
\]
This indicates that it would be difficult to use this result even for the simplest
non-linear problems.
\item
\label{transv}
One does not obtain boundedness for transverse derivatives to the horizon,
i.e.~in $(t^*,r)$ coordinates, $\partial_r \psi$, $\partial^2_r\psi$, etc.
({\bf Exercise}: Why not?)
\item
\label{unnat}
The dependence on initial data is somewhat unnatural. ({\bf Exercise}: Work out explicitly
what it is.)
\vskip1pc
\emph{As far as the method of proof is concerned, there are additional shortcomings when the
proof is viewed  from the standpoint of possible future generalisations:}
\item
\label{usess}
To obtain control at the horizon, one
 must commute (see~$(\ref{onthehor})$) with all angular momentum operators $\Omega_i$. Thus
the spherical symmetry of Schwarzschild is used in a fundamental way.
\item
\label{exat}
The exact staticity is fundamental for the construction of $\tilde\psi$. It is not clear
how to generalise this argument in the case say where $T$ is not hypersurface orthogonal 
and Killing but one assumes merely that its
deformation tensor $^T\pi_{\mu\nu}$ decays. This would be the situation
in a bootstrap setting of a non-linear stability problem.
\item
\label{disct}
The construction of $\bar\psi$ requires the discrete isometry of Schwarzschild, which again,
cannot be expected to be stable.
\end{enumerate}

\subsection{The red-shift and a new proof of boundedness}
\label{newproofofb}
We give in this section a new proof of boundedness which
overcomes the shortcomings outlined above.
In essence, the previous proof limited itself by relying solely 
on Killing fields as multipliers and commutators. 
It turns out that there is an important physical aspect of Schwarzschild
which can be captured by other vector-field multipliers and commutators
which are not however Killing.
This is related to the celebrated \emph{red-shift effect}.

\subsubsection{The classical red-shift}
The \emph{red-shift effect} is one of the most
celebrated aspects of black holes. It is classically described as follows:
 Suppose two observers, $A$ and $B$ are such that
$A$ crosses the event horizon and $B$ does not. If $A$ emits a signal at constant
frequency as he measures it, then the frequency at which it is received by
$B$ is ``shifted to the red''. 
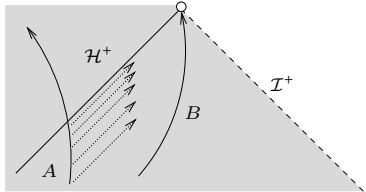
\[
\input{redshift.pstex_t}
\]
The consequences of this for the appearance of a  collapsing star  to far-away observers
were first explored in the seminal paper of Oppenheimer-Snyder~\cite{os} referred to at length
in Section~\ref{introsec}. For a nice discussion, see also the classic textbook~\cite{MTW}.

The red-shift effect as described above is a global one, and essentially depends only
on the fact that the proper time of $B$ is infinite whereas the proper time of $A$ before crossing
$\mathcal{H}^+$ is finite.
In the case of the Schwarzschild black hole, 
there is a ``local'' version of this red-shift: If $B$ also crosses the event 
horizon but at advanced time later than $A$:
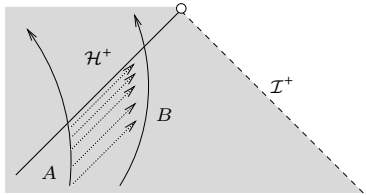
\[
\input{locredshift.pstex_t}
\]
 then the 
frequency at which $B$ receives at his horizon crossing time
is shifted to the red by a factor depending exponentially on the advanced time difference
of the crossing points of $A$ and $B$.

The exponential factor is determined by the so-called \emph{surface gravity},
a quantity that can in fact be defined for all so-called Killing horizons. This localised red-shift
effect depends only on the positivity of this quantity.  We shall understand this more general
situation in Section~\ref{epilogue}. Let us for now simply explore how we can ``capture'' the
red-shift effect in the Schwarzschild geometry.

\subsubsection{The vector fields $N$, $Y$, and $\hat{Y}$}
\label{thevectorfields}
It turns out that a ``vector field multiplier'' version of this localised red-shift effect is
captured by
the following 
\begin{proposition}
\label{newproprs}
There exists a $\varphi_t$-invariant smooth future-directed
timelike vector field $N$ on $\mathcal{R}$
and a positive constant $b>0$ such that
\[
K^N(\psi) \ge b J^N_\mu (\psi) N^\mu
\]
on $\mathcal{H}^+$ for all solutions $\psi$ of $\Box_g\psi=0$.
\end{proposition}
(See Appendix~\ref{VFM} for the $J^N$, $K^N$ notation.)
\begin{proof}
Note first that since $T$ is tangent to $\mathcal{H}^+$, it follows that given any $\sigma<\infty$,
there clearly exists 
a  vector field $Y$ 
on $\mathcal{R}$ such that
\begin{enumerate}
\item
\label{newf}
$Y$ is $\varphi_t$ invariant and spherically symmetric.
\item
\label{newse}
$Y$ is future-directed null on $\mathcal{H}^+$ and transverse to $\mathcal{H}^+$,
say $g(T,Y)=-2$.
\item
On $\mathcal{H}^+$, 
\begin{equation}
\label{euko}
\nabla_YY=  -\sigma\, (Y+T).
\end{equation}
\end{enumerate}
Since $T$ is tangent to $\mathcal{H}^+$, along which $Y$ is null, we have
\begin{equation}
\label{fromkilling}
g(\nabla_TY,Y)=0.
\end{equation}
From properties~\ref{newf} and~\ref{newse}, and the form of the Schwarzschild
metric, one computes $({\bf Exercise})$
\begin{equation}
\label{KRUCIAL}
g(\nabla_TY,T) \doteq 2\kappa > 0
\end{equation}
on $\mathcal{H}^+$.
Defining a local frame $E_1$, $E_2$ for the ${\rm SO}(3)$ orbits,
we note
\[
g(\nabla_{E_i}Y, Y) = \frac12 E_ig (Y,Y)=0,
\]
\[
g(\nabla_{E_1}Y, E_2) = -g(Y, \nabla_{E_1}E_2)=-g(Y,\nabla_{E_2}E_1)
=g(\nabla_{E_2}Y,E_1).
\]
Writing  thus
\begin{equation}
\label{listten1}
\nabla_TY = -\kappa  Y + a^1\, E_1 +a^2\, E_2
\end{equation}
\begin{equation}
\label{listten2}
\nabla_YY = -\sigma\, T-\sigma\, Y
\end{equation}
\begin{equation}
\label{listten3}
\nabla_{E_1}Y = h^1_1 \, E_1 +  h^2_1\, E_2 -\frac12a^1\, Y 
\end{equation}
\begin{equation}
\label{listten4}
\nabla_{E_2}Y = h^1_2 \,E_1 + h^2_2 \, E_2 - \frac12a^2 \, Y
\end{equation}
with ($h^2_1=h^1_2$),
we now compute
\begin{eqnarray}
\label{Wecompute}
\nonumber
K^Y &=&  - \frac12 \left({\bf T}(Y,Y) (-\kappa)  +\frac12 {\bf T}(T,Y)(-\sigma )+
\frac12 {\bf T}(T,T)(-\sigma) \right)\\
\nonumber
&&\hbox{}-\frac12({\bf T}(E_1, Y)a^1+{\bf T}(E_2,Y)a^2) \\
\nonumber
&&\hbox{}+ {\bf T}(E_1,E_1)h_1^1 +{\bf T}(E_2, E_2) h^1_2+ {\bf T}(E_1,E_2)( h_1^2+h^1_2)
\end{eqnarray}
where we denote the energy momentum tensor by  ${\bf T}$, to prevent confusion with $T$.
(Note that, in view of the fact that $\mathcal{Q}$ imbeds as a totally geodesic submanifold
of $\mathcal{M}$, we have in fact $a^1=a^2=0$. This is  of no importance in our computations,
however.)
It follows immediately in view again of the 
algebraic properties of ${\bf T}$, that
\begin{eqnarray*}
K^Y  &\ge& \frac12\kappa \, {\bf T}(Y,Y) +  \frac14\sigma\, {\bf T}(T,Y+T)\\
&&\hbox{}
  					- c {\bf T}(T,Y+T) -c\sqrt{{\bf T} (T,Y+T){\bf T}(Y,Y)}
\end{eqnarray*}
where $c$ is independent of the choice of $\sigma$.
It follows that choosing $\sigma$ large enough, we have
\[
K^Y \ge b\, J_\mu^{T+Y}(T+Y)^\mu.
\]
So just set $N=T+Y$, noting that $K^N=K^T+K^Y=K^Y$.
\end{proof}

The computation $(\ref{KRUCIAL})$ represents a well known property
of stationary black holes holes
and the constant $\kappa$ is the so-called \emph{surface gravity}. (See~\cite{Townsend}.)
Note that since $Y$ is $\varphi_t$-invariant and $T$ is Killing, we have
\[
g(\nabla_TY,T)=  g(\nabla_YT , T)= - g(\nabla_TT,Y)
\]
on $\mathcal{H}^+$.
On the other hand
\[
g(\nabla_TT,E_i)= -g(\nabla_{E_i}T,T)=0,
\]
since $T$ is null on $\mathcal{H}^+$.
Thus, $\kappa$ is alternatively characterized by
\[
\nabla_TT = \kappa\, T
\]
on $\mathcal{H}^+$.
We will elaborate on this in Section~\ref{epilogue}, where a generalisation of 
Proposition~\ref{newproprs}
will be presented.

{\bf Exercise}: Relate the strength of the red-shift with the constant $\kappa$,
for the case where observers $A$ and $B$ both cross the horizon, but $B$ at advanced
time $v$ later than $A$.

If one desires an explicit form of the vector field, then one can
argue as follows:
Define first the vector field $\hat{Y}$ by 
\begin{equation}
\label{Yhatdef}
\hat{Y} = \frac{1}{1-2M/r}\partial_u.
\end{equation}
(See Appendix~\ref{computs}.)
Note that this vector field satisfies $g(\nabla_{\hat Y}\hat{Y}, T)=0$.
Define 
\[
Y= (1+\delta_1(r-2M))\hat{Y} +\delta_2(r-2M)T.
\]
It suffices to choose $\delta_1$, $\delta_2$ appropriately.

The behaviour of $N$ away from the horizon is of course irrelevant in the above proposition.
It will be useful for us to have the following:
\begin{corollary}
\label{Nproperties}
Let $\Sigma$ be as before.
There exists a $\varphi_t$-invariant smooth future-directed
timelike vector field $N$ on $\mathcal{R}$, constants
$b>0$, $B>0$, and two
values $2M<r_0<r_1<\infty$ such that
\begin{enumerate}
\item
\label{fir}
$K^N \ge  b\, J^N_\mu n^\mu_{\Sigma}$ for $r\le r_0$,
\item 
$N=T$ for $r\ge r_1$,
\item 
\label{trito}
$|K^N  |\le B J^T_\mu  n^\mu_{\Sigma}$, and
 $J^N_\mu n^\mu_{\Sigma}\sim J^Tn^\mu_{\Sigma}$ for $r_0\le r\le r_1$.
\end{enumerate}
\end{corollary}

\subsubsection{$N$ as a multiplier}
Recall the definition of $\mathcal{R}(0,\tau)$.
Applying the energy identity with the current $J^N$ in this region, we  obtain
\begin{align}
\label{Nesti}
\nonumber
\int_{\Sigma_\tau} J^N_\mu n^\mu_{\Sigma} +\int_{\mathcal{H}^+(0,\tau)}
J^N_\mu n^\mu_{\mathcal{H}} + \int_{\{r\le r_0\}\cap\mathcal{R}(0,\tau)} K^N\\
= \int_{\{r_0\le r\le r_1\}\cap\mathcal{R}(0,\tau)} (-K^N) + \int_{\Sigma_0}  J^N_\mu n^\mu_{\Sigma}.
\end{align}
The reason for writing the above identity in this form will become apparent in what
follows.
Note that  
since $N$ is timelike at $\mathcal{H}^+$, we see all the ``usual terms'' in 
the flux integrals, i.e.
\[
J^N_\mu n^\mu_{\mathcal{H}} \sim (\partial_{t^*}\psi)^2+ |\nabb \psi|^2,
\]
and 
\[
J^N_\mu n^\mu_{\Sigma_\tau} \sim (\partial_{t^*}\psi)^2+(\partial_r\psi)^2 
+ |\nabb \psi|^2.
\]
The constants in the $\sim$ depend as usual on the choice of the original $\Sigma_0$ and the
precise choice of $N$.

Now the identity $(\ref{Nesti})$ also holds where $\Sigma_0$ is replaced by
$\Sigma_{\tau'}$, $\mathcal{H}^+(0,\tau)$ is replaced by $\mathcal{H}^+(\tau',\tau)$,
and $\mathcal{R}(0,\tau)$ is replaced by $\mathcal{R}(\tau',\tau)$,
for an arbitrary $0\le \tau'\le \tau$.

We may add to both sides of $(\ref{Nesti})$ an arbitrary multiple of the spacetime integral
$\int_{\{r\ge r_0\}\cap \mathcal{R}(\tau',\tau)} J^T_\mu n^\mu_{\Sigma}$.
In view of the fact that
\[
\int_{\{r\ge r'\}\cap \mathcal{R}(\tau',\tau)} J^N_\mu n^\mu_{\Sigma}\sim
\int_{\tau'}^\tau\left(\int _{\{r\ge r'\}\cap \Sigma_{\bar \tau}} J^N_\mu n^\mu_{\Sigma}\right)d\bar\tau
\]
for any $r'\ge 2M$ (where $\sim$ depends on $\Sigma_0$, $N$),
from the inequalities  shown and property~\ref{trito}  of Corollary~\ref{Nproperties} we obtain
\begin{align*}
\int_{\Sigma_\tau} J^N_\mu n^\mu_{\Sigma} 
+b \int_{\tau'}^\tau\left(\int_{\Sigma_{\bar\tau}} J^N_\mu n^\mu_{\Sigma}\right)d\bar\tau \le
B \int_{\tau'}^\tau\left(\int_{\Sigma_{\bar\tau}} J^T_\mu n^\mu_{\Sigma}\right)d\bar\tau
 + \int_{\Sigma_{\tau'}} J^N_\mu n^\mu_{\Sigma}.
\end{align*}

On the other hand, in view of our previous $(\ref{Enest})$, $(\ref{ineqs})$, we have
\begin{equation}
\label{astera}
 \int_{\tau'}^\tau\left(\int_{\Sigma_{\bar\tau}} J^T_\mu n^\mu_{\Sigma}\right)d\bar\tau
\le (\tau-\tau') \int_{\Sigma_0} J^T_\mu n^\mu_{\Sigma}.
\end{equation}
Setting 
\[
f(\tau)=\int_{\Sigma_\tau} J^N_\mu n^\mu_{\Sigma} 
\]
we have that
\begin{equation}
\label{1d}
f(\tau ) +b \int_{\tau'}^\tau f(\bar\tau ) d\bar\tau \le  B D (\tau-\tau') + f(\tau') 
\end{equation}
for all $\tau\ge \tau'\ge0$,
from which it follows ({\bf Exercise}) that $f \le B(D+f(0))$.
(We use the inequality with
 $D=\int_{\Sigma_0}J^T_\mu n^\mu_{\Sigma_0}$.)
In view of the trivial inequality
\[
\int_{\Sigma_0}J^T_\mu n^\mu_{\Sigma_0}\le
B  \int_{\Sigma_0} J^N_\mu n^\mu_{\Sigma_0},
\]
we obtain
\begin{equation}
\label{finalL2}
\int_{\Sigma_\tau} J^N_\mu n^\mu_{\Sigma_\tau} 
\le
B \int_{\Sigma_0} J^N_\mu n^\mu_{\Sigma_0}.
\end{equation}

We have obtained a ``local observer's'' energy estimate.
This addresses point~\ref{eventhough} of Section~\ref{critical}.

\subsubsection{$\hat{Y}$ as a commutator}
\label{asacommutator}
It turns out $({\bf Exercise})$
that from  $(\ref{finalL2})$, 
one could obtain pointwise bounds as before on $\psi$ by commuting with angular momentum
operators $\Omega_i$. No construction of $\tilde\psi$, $\bar\psi$, etc., would be necessary,
and this would thus address points~\ref{unnat}, \ref{exat}, \ref{disct}  of Section~\ref{critical}.

Commuting with $\Omega_i$ clearly would not address however point~\ref{usess}.
Moreover, it would not address point~\ref{transv}. {\bf  Exercise}: Why not?

It turns out that one can resolve this problem by applying $N$ not only as a multiplier,
but also as a commutator. The calculations are slightly easier if we more simply
commute with $\hat{Y}$ defined in~$(\ref{Yhatdef})$.

\begin{proposition}
\label{poscompu}
Let $\psi$ satisfy $\Box_g\psi=0$. Then we may write
\begin{equation}
\label{newwave}
\Box_{g}(\hat{Y}\psi) 
 = \frac{2}r \hat{Y}(\hat{Y}(\psi))-\frac{4}{r}(\hat Y(T\psi))+  P_1\psi
 \end{equation}
where
$P_1$ is the first order operator $P_1\psi \doteq \frac{2}{r^2}(T\psi -\hat{Y}\psi)$.
\end{proposition}

This is proven easily with the help of Appendix~\ref{commutat}.
As we shall see, the
sign of the first term on the right hand side of $(\ref{newwave})$ is important.
We will interpret this computation geometrically in terms of the sign of the surface gravity in
Theorem~\ref{hocom} of Section~\ref{epilogue}.

 Let us first note that our boundedness result gives us in particular
 \begin{equation}
\label{forthis0}
\int_{\{r\le r_0\}\cap \mathcal{R}(0,\tau)} K^N(\psi)
\le B D\,\tau
\end{equation}
where $D$ comes from initial data.
({\bf Exercise}: Why?)
Commute now the wave equation with $T$ and apply the multiplier $N$.
See Appendix~\ref{commutat}.
One obtains in particular an estimate for
\begin{equation}
\label{forthis}
\int_{\{r\le r_0\}\cap \mathcal{R}(0,\tau)} (\hat{Y} T\psi)^2
\le B 
\int_{\{r\le r_0\}\cap \mathcal{R}(0,\tau) }K^{N}(T\psi)
\le B D\,\tau,
\end{equation}
where again $D$ refers to a quantity coming from initial data. 
Commuting now the wave equation with $\hat{Y}$ and applying the multiplier $N$,
one obtains an energy identity of the form
\begin{align*}
\int_{\Sigma_\tau}& J^N_\mu(\hat{Y}\psi) n^\mu_{\Sigma} +\int_{\mathcal{H}^+(0,\tau)}
J^N_\mu(\hat{Y}\psi) n^\mu_{\mathcal{H}} + \int_{\{r\le r_0\}\cap\mathcal{R}(0,\tau)}
 K^N(\hat{Y}\psi)\\
=& \int_{\{r_0\le r\le r_1\}\cap\mathcal{R}(0,\tau)} (-K^N(\hat{Y}(\psi))\\
& +
\int_{\{r\le r_0\}\cap\mathcal{R}(0,\tau)} \mathcal{E}^N (\hat{Y}\psi)
+\int_{\{r\ge r_0\}\cap\mathcal{R}(0,\tau)}\mathcal{E}^N(\hat{Y}\psi)\\
&+
 \int_{\Sigma_0}  J^N_\mu(\hat{Y}\psi) n^\mu_{\Sigma},
\end{align*}
where $J^N(\hat{Y}\psi)$, $K^N(\hat{Y}\psi)$ are defined by $(\ref{assoc1})$, $(\ref{assoc2})$,
respectively, with $\hat{Y}\psi$ replacing $\psi$, and
\begin{eqnarray*}
\mathcal{E}^N(\hat{Y}\psi)&=&
-(N\hat{Y}\psi)
\left( \frac{2}r \hat{Y}(\hat{Y}(\psi))-\frac{4}{r}(\hat Y(T\psi))+  P_1\psi\right)\\
&=&-
\frac{2}r( \hat{Y}(\hat{Y}(\psi)))^2\\
&&\hbox{}
 -\frac{2}r((N-\hat{Y})\hat{Y}\psi)(\hat{Y}\hat{Y}\psi)
+\frac{4}{r}(N\hat{Y}\psi)(\hat Y(T\psi))\\
&&\hbox{}- (N\hat{Y}\psi)  P_1\psi.
\end{eqnarray*}
The first term on the right hand side has a good sign! 
Applying Cauchy-Schwarz and the fact that
$N-\hat{Y}=T$ on $\mathcal{H}^+$, it follows that choosing $r_0$ accordingly,
one obtains that the second two terms can be bounded in $r\le r_0$ by
\[
\epsilon K^N(\hat{Y}\psi) + \epsilon^{-1} (\hat{Y} T\psi)^2
\]
whereas the last term can be bounded by
\[
\epsilon K^N(\hat{Y}\psi) + \epsilon^{-1} K^N(\psi).
\]
In view of $(\ref{forthis0})$ and $(\ref{forthis})$, one obtains
\[
\int_{\{r\le r_0\}\cap\mathcal{R}(0,\tau)} \mathcal{E}^N (\hat{Y}\psi)
\le\epsilon\int_{\{r\le r_0\}\cap\mathcal{R}(0,\tau)}  K^N(\hat{Y}\psi)
+B\epsilon^{-1} D\tau.
\]
{\bf Exercise}: Show how from this one
can arrive
again at an inequality $(\ref{1d})$.

Commuting repeatedly with $T$, $\hat{Y}$, 
the above scheme plus elliptic estimates yield natural $H^m$ estimates for all $m$.
Pointwise estimates for all derivatives then follow
by a standard Sobolev estimate.

\subsubsection{The statement of the boundedness theorem}
We obtain finally
\begin{theorem}
\label{boundedn}
Let $\Sigma$ be a 
Cauchy hypersurface for Schwarzschild
such that $\Sigma\cap\mathcal{H}^-=\emptyset$, let $\Sigma_0=\mathcal{D}
\cap \Sigma$, let $\Sigma_\tau$ denote the translation of $\Sigma_0$,
let $n_{\Sigma_\tau}$ denote the future normal of $\Sigma_\tau$, and let
$\mathcal{R}= \cup_{\tau\ge0}\Sigma_\tau$. Assume $-g(n_{\Sigma_0},T)$ is uniformly bounded. 
Then there
exists a constant $C$ depending only on $\Sigma_0$ such that the following holds.
Let $\psi$, $\uppsi$, $\uppsi'$ be as in Proposition~\ref{prelim}, with
$\uppsi \in H^{k+1}_{\rm loc}(\Sigma)$, $\uppsi'\in H^k_{\rm loc}(\Sigma)$,
and 
\[
\int_{\Sigma_0} J^T_\mu (T^m\psi) n^\mu_{\Sigma_0}<\infty
\]
for $0\le m\le k$.
Then 
\[
|\nabla^{\Sigma_\tau}\psi|_{H^{k}(\Sigma_\tau)} + |n\psi|_{H^{k}(\Sigma_\tau)} \le C\left(
|\nabla^{\Sigma_0}\uppsi|_{H^{k}(\Sigma_0)} + |\uppsi'|_{H^{k}(\Sigma_0)} \right).
\]
If $k\ge 1$, then we have 
\[
\sum_{0\le m\le k-1}\sum_{m_1+m_2=m, m_i\ge 0}
|(\nabla^\Sigma)^{m_1}n^{m_2} \psi|\le C\left(\lim_{x\to i^0}|\uppsi| +  |\nabla^{\Sigma(0)}\uppsi|_{H^{k}(\Sigma_0)} + |\uppsi'|_{H^{k}(\Sigma_0)} \right)
\]
in $\mathcal{R}$.
\end{theorem}
Note that $(\nabla^\Sigma)^{m_1}n^{m_2} \psi$ denotes an $m_1$-tensor on
the Riemannian manifold $\Sigma_\tau$, and $|\cdot|$ on the left hand side of the 
last inequality above just denotes the induced
norm on such tensors.

\subsection{Comments and further reading}
\label{comsecs}
The first discussion of the wave equation on Schwarzschild is
perhaps the work of Regge and Wheeler~\cite{RW}, but the true
mathematical study of this problem was initiated by Wald~\cite{drimos},
who proved Theorem~\ref{KW} under the assumption that $\psi$ vanished
in a neighborhood of $\mathcal{H}^+\cap\mathcal{H}^-$. The full statement
of Theorem~\ref{KW}  and the proof presented in Section~\ref{KWsect}
is due to Kay and Wald~\cite{kw:lss}. The present notes owe a lot to
the geometric view point emphasized
in the works~\cite{drimos, kw:lss}.

Use of the vector field $Y$ as a multiplier was first introduced 
in our~\cite{dr3}, and
its use is central in~\cite{dr4} and~\cite{dr5}.
In particular, the property formalised by
Proposition~\ref{newproprs}
was discovered there. It appears that this
may be key to a stable understanding of black hole event horizons.
See
Section~\ref{pertsec0} below, as well as Section~\ref{epilogue}, for a generalisation
of Proposition~\ref{newproprs}.

It is interesting to note that
in~\cite{dr4,dr5}, $Y$ had always been used
in conjunction with vector fields $X$ of the type to be discussed in the next section (which require a more delicate global construction)
as well as $T$.  This meant that one always had to obtain \emph{more} than boundedness (i.e.~decay!)
in order to obtain the proper boundedness result at the horizon. Consequently, one
had to use many aspects of the structure of Schwarzschild, particularly, the trapping to 
be discussed in later lectures.
The argument given above, where boundedness is obtained
using only $N$ and $T$
as multipliers is presented for the first time in a self-contained fashion in these lectures. 
The argument can be read off, however, from the more general argument
of~\cite{dr6} concerning perturbations of Schwarzschild including Kerr.
The use of $\hat{Y}$ as a commutator to estimate higher order quantities also originates in~\cite{dr6}.
The geometry behind this computation is further discussed in Section~\ref{epilogue}.

Note that the use of $Y$ together with $T$ is of course equivalent to the
use of $N$ and $T$. We have chosen to give a name to the vector field $N=T+Y$ 
merely for convenience. Timelike vector fields are more convenient when perturbing\ldots

Another remark on the use of $\hat{Y}$ as a commutator:
Enlarging the choice of commutators has proven very important in previous work on
the global 
analysis of the wave equation. In a seminal paper, 
Klainerman~\cite{muchT} showed improved decay for the wave 
equation on Minkowski space in the interior region
by commutation with scaling and Lorentz boosts. This was a key step
for further developments for long time existence for quasilinear wave equations~\cite{santafe}.

The distinct role of multipliers and commutators and the geometric considerations
which enter into their construction is beautifully elaborated by 
Christodoulou~\cite{fluids}.

\subsection{Perturbing?}
\label{pertsec0}
Can the proof of Theorem~\ref{boundedn} be adapted to hold for spacetimes
``near'' Schwarzschild? 
To answer this, one must first decide what one means by the notion of ``near''. 
Perhaps the simplest class of perturbed metrics would be those that retain the same differentiable
structure of $\mathcal{R}$, 
retain $\mathcal{H}^+$ as a null hypersurface, and retain the Killing field $T$.
One infers (without computation!) that the statement of
Proposition~\ref{newproprs} and thus Corollary~\ref{Nproperties}
 is stable to such perturbations of
the metric.
Therein lies the power of that Proposition and of the multiplier $N$.
(In fact, see Section~\ref{epilogue}.)
Unfortunately, one easily sees that our argument proving Theorem~\ref{boundedn} is
still unstable, even in the
class of perturbations just described. The reason is the following:
Our argument relies essentially on an a priori estimate for
$\int_{\Sigma_\tau}J^T_\mu n^\mu$ (see~$(\ref{astera})$), 
which requires $T$ to be non-spacelike in $\mathcal{R}$. When one
perturbs,  $T$ will in general become spacelike in a region of $\mathcal{R}$. 
(As we shall see in Section~\ref{thekerrmetric}, 
this happens in particular in the case of Kerr. The region where $T$ is spacelike is known
as the ergoregion.)

There is a sense in which the above
is the only obstruction to perturbing the above argument, 
i.e.~one can solve the following

{\bf Exercise}:
Fix the differentiable structure of $\mathcal{R}$ and the vector field $T$. Let $g$
be a metric sufficiently close to Schwarzschild such that $\mathcal{H}^+$ is null, 
and suppose $T$ is Killing and 
non-spacelike in $\mathcal{R}$, and $T$ is null on $\mathcal{H}^+$. Then 
Theorem~\ref{boundedn} applies.
 (In fact, one need not assume that $T$ is non-spacelike in $\mathcal{R}$,
 only that $T$ is null on the horizon.) See also Section~\ref{epilogue}.

{\bf Exercise}:
Now do the above where $T$ is not assumed to be Killing, but $^T\pi_{\mu\nu}$ is
assumed to decay suitably. What precise assumptions must one impose?

This discussion may suggest that there is in fact no stable boundedness argument, that is to say,
a ``stable argument'' would of necessity need to prove more than boundedness, i.e.~decay. 
We shall see later that there is a sense in which this is true and a sense in which it is not!
But before exploring this, let us understand how one can go beyond boundedness
and prove quantitative decay for waves on Schwarzschild itself.
It is quantitative decay after all that we must understand if we are to understand nonlinear problems.

\section{The wave equation on Schwarzschild II: quantitative decay rates}
\label{S2}
Quantitative decay rates are central for our understanding of non-linear problems.
To discuss energy decay for solutions $\psi$ of $\Box_g\psi=0$ on Schwarzschild, 
one must consider a different foliation.
Let $\tilde{\Sigma}_0$ be a spacelike hypersurface terminating on null infinity
and define $\tilde{\Sigma}_\tau$ by translation. 
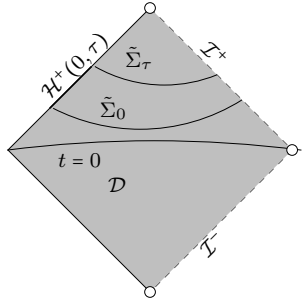
\[
\input{newfol.pstex_t}
\]
The main result of this section is the following
\begin{theorem}
\label{Schdec}
There exists a constant $C$ depending only on $\tilde\Sigma_0$
such that the following holds. 
Let $\uppsi\in H^{4}_{\rm loc}$, $\uppsi'\in H^{3}_{\rm loc}$, and 
suppose $\lim_{x\to i^0}\uppsi=0$ and
\begin{eqnarray*}
E_1= \sum_{|(\alpha)|\le 3 } \sum_{\Gamma=
\{\Omega_i\}}\int_{t=0} r^2 J^{n_0}_\mu (\Gamma^{(\alpha)} \psi)n_{0}^\mu
<\infty
\end{eqnarray*}
where $n_0$ denotes the unit normal of the hypersurface $\{t=0\}$.
Then    
\begin{equation}
\label{tote...}
 \int_{\tilde\Sigma_\tau} J^{N} _\mu(\psi) n^\mu_{\tilde\Sigma_\tau} \le 
C E_1 \tau^{-2},
\end{equation}
where $N$ is the vector field of Section~\ref{thevectorfields}.
Now let $\uppsi\in H^{7}_{\rm loc}$, $\uppsi'\in H^{6}_{\rm loc}$, $\lim_{x\to i^0}\uppsi=0$,
and suppose
\begin{eqnarray*}
E_2=\sum_{|(\alpha)|\le6 } \sum_{\Gamma=\{\Omega_i\}}\int_{t=0} r^2 J^{n_0}_\mu (\Gamma^{(\alpha)} \psi)
n_0^\mu
<\infty.
\end{eqnarray*}
Then
\begin{equation}
\label{ptb}
\sup_{\tilde\Sigma_{\tau}} \sqrt{r}|\psi| \le C\sqrt{E_2} \tau^{-1}, \qquad
\sup_{\tilde\Sigma_{\tau}} r|\psi| \le C\sqrt{E_2} \tau^{-1/2}.
\end{equation}
\end{theorem}

The fact that~$(\ref{tote...})$
``loses derivatives'' is a fundamental aspect of this problem related to the trapping
phenomenon, to be discussed in what follows, although
the precise number of derivatives lost above is
wasteful. Indeed, there are several aspects in which the
above results can
be improved. See Proposition~\ref{nowaste} and the exercise of Section~\ref{pointdec}.

We can also express the pointwise decay in terms of advanced and retarded null coordinates
$u$ and $v$.  Defining\footnote{The strange convention on the factor of $2$
is chosen simply to agree with~\cite{dr3}.} 
$v =2(t+r^*)=2(t+r+2M\log(r-2M))$, $u=2(t-r^*)=2(t-r-2M\log(r-2M))$,
it follows in particular from $(\ref{ptb})$ that
\begin{equation}
\label{innull}
|\psi|\le CE_2(|v|+1)^{-1}, \qquad
|r\psi| \le C(r_0)E_2 u^{-\frac12},
\end{equation}
where the first inequality applies in
$\mathcal{D}\cap {\rm clos} (\{t\ge 0\})$, whereas the second 
applies only in $\mathcal{D}\cap \{t\ge 0\}\cap \{r\ge r_0\}$,
with $C(r_0)\to\infty$ as $r_0\to 2M$. See also Appendix~\ref{computs}.
Note that, as in Minkowski space,
 the first inequality of $(\ref{innull})$ is sharp as a \emph{uniform} decay rate in $v$.

\subsection{A spacetime integral estimate}
\label{sie}
The zero'th step in the proof of Theorem~\ref{Schdec}
is an estimate for a spacetime integral
whose integrand should control the quantity
\begin{equation}
\label{shouldcontrol}
\chi\, J^{N}_\mu(\psi) n^\mu_{\tilde\Sigma_\tau}
\end{equation}
where $\chi$ is a $\varphi_t$-invariant weight function
such that $\chi$ degenerates only at infinity.
Estimates of the spacetime integral $(\ref{shouldcontrol})$ have their
origin in the classical virial theorem, which in Minkowski space
essentially arises from
applying the energy identity to the current $J^V$ with $V=\frac{\partial}{\partial r}$.

Naively, one might expect to be able to obtain an estimate of the form say
\begin{equation}
\label{naively}
\int_{\tilde{\mathcal{R}}(0,\tau)} \chi J^N_\mu(\psi) n^\mu_{\tilde\Sigma_\tau}
\le B\int_{\tilde\Sigma_0} J^N_\mu n^\mu_{\tilde\Sigma_0},
\end{equation}
for such a $\chi$.
It turns out that there is a well known
high-frequency obstruction for the existence of an estimate of the form 
$(\ref{naively})$ arising
from \emph{trapped null geodesics}. 
This problem has been long studied in the context of the wave equation in Minkowski space
outside of an obstacle, where the analogue of trapped null geodesics are straight lines
which reflect off the obstacle's boundary in such a way so as to remain in a compact subset
of space.
In Schwarzschild, one can easily infer from a continuity
argument the existence of a family of
null geodesics with $i^+$ as a limit point.\footnote{This can be thought of as a very weak notion of
what it would mean for a null geodesic to be trapped from the point of view of decay
results with respect to the foliation $\tilde\Sigma_\tau$.}
But in view of the integrability of geodesic flow, one can in fact understand all such geodesics
explicitly.

{\bf Exercise}: Show that the hypersurface $r=3M$ is spanned by null geodesics. Show that from every
point in $\mathcal{R}$, there is a codimension-one subset of
 future directed null directions whose corresponding
geodesics approach $r=3M$, and all other null geodesics either cross $\mathcal{H}^+$ or
meet $\mathcal{I}^+$. 

The timelike hypersurface $r=3M$ is traditionally called
the \emph{photon sphere}.
Let us first see how one can 
capture this high frequency obstruction.

\subsubsection{A multiplier $X$ for high angular frequencies}
\label{hayhay}
We look for a multiplier with the property that the spacetime integral it generates
is positive definite. Since in Minkowski space, this is provided by the vector field
$\partial_r$, we will look for simple generalizations.
Calculations are slightly easier when one considers $\partial_{r^*}$ associated
to Regge-Wheeler coordinates $(r^*,t)$. See Appendix~\ref{RWcs} for the definition of
this coordinate system.\footnote{Remember,
when considering coordinate vector fields, one has to specify the entire coordinate system. 
When considering $\partial_r$, it is here to be understood that
we are using Schwarzschild coordinates,
and when considering $\partial_{r^*}$, it is to be understood that
we are using Regge-Wheeler coordinates. The precise choice of the angular coordinates
is of course irrelevant.}
For $X=f(r^*)\partial_{r^*}$, where $f$ is a general function,
we obtain the formula
\[
K^X= \frac{f'}{1-2M/r}(\partial_{r^*}\psi)^2
+\frac{f}r\left(1-\frac{3M}r\right)|\nabb\psi|^2
-\frac14\left(2f'+4\frac{r-2M}{r^2}f\right)\nabla^\alpha \psi\nabla_\alpha\psi.
\]
Here $f'$ denotes $\frac{df }{dr^*}$.
We can now define a ``modified'' current 
\[
J^{X,w}_\mu=J^X_\mu (\psi) +\frac18w\partial_\mu(\psi^2)-\frac18(\partial_\mu w)\psi^2
\]
associated to the vector field $X$ and the function $w$.
Let 
\[
K^{X,w}= \nabla^\mu J^{X,w}_\mu.
\]
Choosing
\[
w=f'+2\frac{r-2M}{r^2}
f +
\frac{\delta(r-2M)}{r^5}\left(1-\frac{3M}r\right)f,
\]
we have
\begin{eqnarray*}
K^{X,w} &=&
\left( \frac{f'}{1-2M/r}-\frac{\delta f}{2r^4}\left(1-\frac{3M}r\right)\right)(\partial_{r^*}\psi)^2\\
&&\hbox{}
+\frac{f}r\left(1-\frac{3M}r\right)\left(\left(1-\frac{\delta(r-2M)}{2r^4}\right)|\nabb\psi|^2+\frac \delta{2r^3} (\partial_t \psi)^2\right)\\
&&\hbox{}
-\left(\frac18\Box_{g}\left(2f'+4\frac{r-2M}{r^2}f
 +
2\frac{\delta(r-2M)}{r^5}\left(1-\frac{3M}r\right)f
\right) \right)\psi^2.
\end{eqnarray*}

Recall that in view of the spherical symmetry of $\mathcal{M}$, we may
decompose
\[
\psi= \sum_{\ell\ge 0,|m|\le \ell} \psi_\ell(r,t) Y_{m,\ell}(\theta,\phi)
\]
where $Y_{m,\ell}$ are the so-called \emph{spherical 
harmonics}, each summand satisfies again the wave equation, and the convergence
is in $L^2$ of the ${\rm SO}(3)$ orbits.

Let us assume that $\psi_\ell=0$ for spherical harmonic number $\ell\le L$
for some $L$ to be determined. 
We look for $K^{X,w}$ such that $\int_{\mathbb S^2} {K^{X,w}}\ge 0$, but also
$|J^{X,w}_\mu n^\mu|\le B J^N_\mu n^\mu$.
Here $\int_{\mathbb S^2}$ denotes integration over group orbits of the ${\rm SO}(3)$ action.
For such $\psi$,
in view of  the resulting inequality
\[
 \frac{L(L+1)}{r^2}\int_{\mathbb S^2} \psi^2 \le \int_{\mathbb S^2} |\nabb\psi|^2,
\]
it follows that
taking $L$ sufficiently large and 
$0<\delta<1$ sufficiently small so that
$1-\frac{\delta(1-\mu)}{2r^3}\ge \frac12 $, 
 it clearly
suffices to construct an $f$
with the following properties:
\begin{enumerate}
\item
$|f|\le B$,
\item 
$f'\ge  B (1-2M/r)r^{-4}$, 
\item
$f(r=3M)=0$,
\item
$-\frac18\Box_{g}\left(2f'+4\frac{r-2M}{r^2}f
 +
2\frac{\delta(r-2M)}{r^5}\left(1-\frac{3M}r\right)f
\right)
(r=3M)> 0$,
\item
$\frac18\Box_{g}
\left(2f'+4\frac{r-2M}{r^2}f
 +
2\frac{\delta(r-2M)}{r^5}\left(1-\frac{3M}r\right)f
\right)
\le  Br^{-3}$
\end{enumerate}
for some constant $B$.
{\bf Exercise}. Show that one can construct such a function.

Note the significance of the photon sphere!

\subsubsection{A multiplier $X$ for all frequencies}
\label{oneforall}
Constructing a multiplier for all spherical harmonics,
so as to capture in addition ``low frequency'' effects, is more tricky. 
It turns out, however, that one can actually define a current which does not require 
spherical harmonic decomposition at all.
The current is of the form:
\begin{eqnarray*}
J^{\bf X}_\mu(\psi) &=&
e J^N_\mu(\psi) + J^{X^a}_\mu(\psi)+ \sum_iJ^{X^b,w^b}_\mu(\Omega_i\psi)\\
&&\hbox{}-
\frac12\frac{r(f^b)'}{f^b(r-2M)}\left(\frac{r-2M}{r^2}-\frac{(r^*-\alpha-\alpha^{1/2})}{
\alpha^2+(r^*-\alpha-\alpha^{1/2})^2}\right) X^b_\mu \psi^2.
\end{eqnarray*}
Here, $N$ is as in Section~\ref{thevectorfields}, $X^a=f^a \partial_{r^*}$, $X^b=f^b\partial_{r^*}$,
the warped current $J^{X,w}$ is defined as in Section~\ref{hayhay}, 
\[
f^a = -\frac{C_a}{\alpha r^2} + \frac{c_a}{r^{3}},
\]
\[
f^b= \frac{1}{\alpha}\left(\tan^{-1}\frac{r^*-\alpha-\alpha^{1/2}}\alpha -\tan^{-1}(-1-\alpha^{-1/2})\right),
\]
\[
w^b= 
\frac{1}{8}\left((f^b)'+2\frac{r-2M}{r^2}f^b\right),
\]
and $e$, $C_a$, $c_a$, $\alpha$ are positive parameters which must be chosen accordingly.
With these choices, one can show (after some computation)
that the divergence $K^{\bf X}=\nabla^\mu J_\mu^{\bf X}$ 
controls in particular
\begin{equation}
\label{cip}
\int_{\mathbb S^2} K^{\bf X}(\psi)  \ge
b\chi  \int_{\mathbb S^2} J^N_\mu(\psi) n^\mu,
\end{equation}
where $\chi$ is non-vanishing but decays (polynomially) as $r\to\infty$.
Note that in view of the normalisation $(\ref{burada})$ of the $r^*$ coordinate, 
$X^b=0$ precisely at $r=3M$. 
The left hand side of the inequality $(\ref{cip})$ controls also
second order derivatives which degenerate however at $r=3M$. We have dropped these
terms. It is actually useful for applications that the $J^{X^a}(\psi)$ part of the current
is not ``modified'' by a function $w^a$, and thus
$\psi$ itself does not occur in the boundary terms.
That is to say
\begin{equation}
\label{demek=}
|J_\mu^{\bf X}(\psi) n^\mu | \le B\left(J^N_\mu(\psi) n^\mu+ 
\sum_{i=1}^3J^N_\mu(\Omega_i\psi)n^\mu \right).
\end{equation}
On the event horizon $\mathcal{H}^+$, we have a better one-sided bound
\begin{equation}
\label{better1}
-J_\mu^{\bf X}(\psi) n^\mu_{\mathcal{H}^+} \le B \left(J^T_\mu(\psi)n^\mu_{\mathcal{H}^+}+
\sum_{i=1}^3 J^T_\mu(\Omega_i\psi)n^\mu_{\mathcal{H}^+}\right).
\end{equation}
For details of the construction, see~\cite{dr5}.

In view of $(\ref{cip})$, $(\ref{demek=})$ and $(\ref{better1})$, together
with the previous boundedness result Theorem~\ref{boundedn}, 
one obtains in particular the estimate
\begin{equation}
\label{finalestimate}
\int_{\tilde{\mathcal{R}}(\tau',\tau)} \chi J^N_\nu(\psi) n^\nu_{\tilde{\Sigma}}
      \le B \int_{\tilde\Sigma(\tau')}    \left(       J^N_\mu(\psi)  +         \sum_{i=1}^3 J^N_\mu
      (\Omega_i \psi)\right)
      n^\mu_{\tilde\Sigma_\tau},
\end{equation}
for some  nonvanishing $\varphi_t$-invariant function
$\chi$ which decays polynomially
as $r\to\infty$.

On the other hand, 
considering the current $J_\mu^{\bf X} (P_{\le L}\psi) + J^{X,w}_\mu((I-P_{\le L})\psi)$,
where $J_\mu^{X,w}$ is the current of Section~\ref{hayhay}  and $P_{\le L}\psi$ denotes the projection
to the space spanned by spherical harmonics with $\ell \le L$, we obtain the estimate
\begin{equation}
\label{finalotherestimate}
\int_{\tilde{\mathcal{R}}(\tau',\tau)} \chi h J^N_\nu(\psi) n^\nu_{\tilde{\Sigma}}
      \le B \int_{\tilde\Sigma(\tau')}     J^N_\mu(\psi)n^\mu_{\tilde\Sigma_\tau},
\end{equation}
where $h$ is any smooth nonnegative function $0\le h\le 1$ vanishing at $r=3M$, and $B$ depends
also on the choice of function $h$.

\subsection{The Morawetz conformal $Z$ multiplier and energy decay}
\label{Morawetz}
How does the estimate $(\ref{finalestimate})$ assist us to prove decay?

Recall that energy decay can be proven in Minkowski space with the help 
of the so-called \emph{Morawetz current}.
Let
\begin{equation}
\label{defofZ}
Z= u^2\partial_u +v^2\partial_v 
\end{equation}
and define
\[
J^{Z,w}_\mu(\psi)= J^{Z}_\mu(\psi) +\frac{tr^*(1-2M/r)}{2r} \psi\partial_\mu\psi
-\frac{r^*(1-2M/r)}{ 4r}\psi^2\partial_\mu t.
\]
(Here $(u,v)$, $(r^*,t)$ are the coordinate systems of Appendix~\ref{computs}.)
Setting $M=0$, this corresponds precisely to the current introduced by Morawetz~\cite{mora}
on Minkowski space.

It is a good {\bf exercise} to show that (for $M>0$!) the coefficients of this current are
$C^0$ but not $C^1$ across $\mathcal{H}^+\cup\mathcal{H}^-$.

To understand how one hopes to use this current, let us recall the situation in
Minkowski space. There, the signifance of $(\ref{defofZ})$ arises since it is a conformal 
Killing field. 
Setting $M=0$, $r^*=r$ in the above
one obtains\footnote{The reason for introducing the $0$'th order terms is because
the wave equation is not conformally invariant. It is remarkable that one can nonetheless
obtain positive definite boundary terms, although
a slightly unsettling feature is that this
positivity property $(\ref{remtrue})$ requires looking specifically at constant $t=\tau$ surfaces
and integrating.}
\begin{equation}
\label{remtrue}
\int_{t=\tau} J_\mu^{Z,w} n^\mu \ge 0,
\end{equation}
\begin{equation}
\label{doesn't}
K^{Z,w}=0.
\end{equation}

The inequality $(\ref{remtrue})$ remains true in the Schwarzschild case
and one can obtain exactly as before
\begin{equation}
\label{fLHS}
\int_{t=\tau} J_\mu^{Z,w} n^\mu \ge 
b \int_{t=\tau} u^2(\partial_u\psi)^2 +v^2(\partial_v\psi)^2 + \left(1-\frac{2M}r\right)(u^2+v^2)|\nabb\psi|^2.
\end{equation}
(In fact, we have dropped positive $0$'th order terms from the right hand side of $(\ref{fLHS})$,
which will be useful for us later on in Section~\ref{pointdec}.)
Note that away from the horizon, we have that
\begin{equation}
\label{f2LHS}
\int_{t=\tau} J_\mu^{Z,w}n^\mu \ge b(r_0,R)\tau ^{2} \int_{\{t=\tau\}\cap\{r_0\le r\le R\}} J^N_\mu n^\mu.
\end{equation}
Thus, if the left hand side of $(\ref{f2LHS})$ could be shown to be
bounded, this would prove the first statement
of Theorem~\ref{Schdec} where $\tilde{\Sigma}_\tau$ is replaced however with
 $\{t=\tau\}\cap \{r_0\le r\le R\}$.

In the case of Minkowski space,
the boundedness of the left hand side of $(\ref{fLHS})$ follows immediately by $(\ref{doesn't})$
and the energy identity
\begin{equation}
\label{Zei}
\int_{t=\tau} J_\mu^{Z,w} + \int_{0\le t \le \tau} K^{Z,w}
=
\int_{t=0} J_\mu^{Z,w} 
\end{equation}
as long as the data are suitably regular and decay so as for the right hand side to be bounded.
For Schwarzschild, 
one cannot expect $(\ref{doesn't})$ to hold,
and this is why we have introduced
the $X$-related currents.

First the good news:
There exist constants $r_0<R$ such that
\[
K^{Z,w}\ge 0
\]
for $r\le r_0$,
and in fact 
\begin{equation}
\label{usingalso}
K^{Z,w}\ge b \frac{t}{r^3}\psi^2
\end{equation}
for $r\ge R$ and some constant $b$.
These terms have the ``right sign'' in the energy identity $(\ref{Zei})$.
In $\{r_0\le r\le R\}$, however, the best we can do is
\[
-K^{Z,w} \le   B\,t\, ( |\nabb\psi|^2 +   |\psi|^2).
\]
This is the bad news, although, in view of the presence of trapping, it is
to be expected.
Using 
also $(\ref{usingalso})$,
we may
estimate
\begin{eqnarray}
\label{theestimateK}
\nonumber
\int_{0\le t\le \tau} -K^{Z,w} &\le&
 B\, \int_{\{0\le t\le \tau\}\cap \{r_0 \le r\le R\} } t\, J^N_\mu n^\mu\\
&\le&
 B\, \tau\, \int_{\{0\le t\le \tau\}\cap \{r_0 \le r\le R\} } J^N_\mu n^\mu.
\end{eqnarray}

In view of the fact that the first integral on the right hand side of $(\ref{theestimateK})$
is bounded by
$(\ref{finalestimate})$,
and the weight $\tau^2$ in $(\ref{f2LHS})$, applying
the energy
identity of the current $J^{Z,w}$ in the region
$0\le t\le \tau$, we  
 obtain immediately a preliminary version of the
first statement of the Theorem~\ref{Schdec}, but with
$\tau^2$ replaced by $\tau$, and the hypersurfaces
$\tilde\Sigma_\tau$ replaced by $\{t=\tau\}\cap \{ r'\le r\le R'\}$
for some constants $r'$, $R'$, but where $B$ depends on these constants.
(Note the geometry of this region. All $\{t=\rm constant\}$ hypersurfaces have common
boundary $\mathcal{H}^+\cap \mathcal{H}^-$. {\bf Exercise}: Justify the integration by
parts $(\ref{Zei})$, in view of the fact that $Z$ and $w$ are
only $C^0$ at $\mathcal{H}^+\cup\mathcal{H}^-$.)

Using the current $J^T$ and an easy geometric argument, 
it is not difficult to replace the hypersurfaces $\{t=\tau\}\cap \{r'\le r\le R'\}$ above
with $\tilde\Sigma_\tau\cap\{r \ge r'\}$,\footnote{Hint: Use $(\ref{fLHS})$ to estimate
the energy on $\{t=t_0\}\cap J^+(\tilde\Sigma_\tau)$ with weights in $\tau$. Send $t_0\to\infty$
and estimate backwards to $\tilde{\Sigma}_\tau$ using conservation of the $J^T$ flux.}
obtaining
\begin{equation}
\label{athalf0}
\int_{\tilde{\Sigma}_\tau\cap\{r\ge r'\}} J_\mu^{N}(\psi) n^\mu \le B\, \tau^{-1}\left(
  \int_{t=0} J_\mu^{Z,w}(\psi) n^\mu
 + \int_{\tilde\Sigma_0} J_\mu^N(\psi)n^\mu +\sum_{i=1}^3 J_\mu^N(\Omega_i\psi)n^\mu \right).
 \end{equation}

To obtain decay for the nondegenerate energy near the horizon, note
that by the pigeonhole principle in view of the boundedness of the left
hand side of $(\ref{finalestimate})$ and what has just been proven, 
there exists $({\bf exercise})$ a dyadic sequence $\tilde\Sigma_{\tau_i}$ 
for which the first statement of Theorem~\ref{Schdec} holds,
with $\tau^{-2}$ replaced by
$\tau^{-1}_i$. Finally, by Theorem~\ref{boundedn}, we immediately ({\bf exercise}: why?)
remove the restriction to the dyadic sequence.

We have thus obtained
\begin{equation}
\label{athalf}
\int_{\tilde{\Sigma}_\tau} J_\mu^{N}(\psi) n^\mu \le B\, \tau^{-1}\left(
  \int_{t=0} J_\mu^{Z,w}(\psi) n^\mu
 + \int_{\tilde\Sigma_0} J_\mu^N(\psi)n^\mu +\sum_{i=1}^3 J_\mu^N(\Omega_i\psi)n^\mu \right).
 \end{equation}

The statement $(\ref{athalf})$ loses one power of $\tau$ in comparison with 
the first statement of Theorem~\ref{boundedn}. How do we obtain the full result?
First of all, note that, commuting once again with $\Omega_j$,
it follows that $(\ref{athalf})$ holds for $\psi$ replaced with $\Omega_j\psi$.
Now we may partition $\tilde{\mathcal{R}}(0,\tau)$ dyadically into subregions
$\tilde{\mathcal{R}}(\tau_i,\tau_{i+1})$ and
revisit the $X$-estimate $(\ref{finalestimate})$
on each such region. 
In view of $(\ref{athalf})$ applied to both $\psi$ and $\Omega_j\psi$, 
the estimate $(\ref{finalestimate})$
gives
\begin{equation}
\label{fromfinalestimate}
\int_{\tilde{\mathcal{R}}(\tau_i,\tau_{i+1})} \chi J^N_\nu(\psi) n^\nu_{\tilde{\Sigma}}
      \le B D\tau_i^{-1},
\end{equation}
where $D$ is a quantity coming from data.
Summing over $i$, this gives us
that 
\[
\int_{\mathcal{R}(0,\tau)} t\, \chi J_\nu^N(\psi)n^\nu_{\tilde\Sigma} \le  BD(1+ \log |\tau+1|)
\]
This estimates in particular the first term on 
the right hand side of the first inequality of $(\ref{theestimateK})$. 
Applying this inequality, we obtain as before $(\ref{athalf0})$, but
with  $\tau^{-2}(1+\log |\tau+1|)$ replacing $\tau$.
Using $(\ref{fromfinalestimate})$ and a pigeonhole principle, one improves this to
$(\ref{athalf})$, with $\tau^{-2}(1+\log |\tau+1|)$ now replacing $\tau$. 
Iterating this argument again one removes the $\log$ ({\bf exercise}).

Note that this loss of derivatives in $(\ref{tote...})$ simply arises 
from the loss in $(\ref{finalestimate})$.
If $\Omega_i$ could be replaced by $\Omega_i^\epsilon$ in $(\ref{finalestimate})$, 
then the loss would be
$3\epsilon$.
The latter refinement can in fact be deduced from the original $(\ref{tote...})$
using in addition work of Blue-Soffer~\cite{newblue}.
Running the argument of this section
with the $\epsilon$-loss version of $(\ref{tote...})$, we obtain now
\begin{proposition}
\label{nowaste}
For any $\epsilon>0$,
statement $(\ref{tote...})$ holds with $3$ replaced by $\epsilon$ in the definition
of $E_1$ and $C$ replaced by $C_\epsilon$.
\end{proposition}

\subsection{Pointwise decay}
\label{pointdec}
To derive pointwise decay for $\psi$ itself,
we should remember that we have in fact
dropped a good $0$'th order term
from the estimate $(\ref{fLHS})$. In particular, we have also
\[
\int_{t=\tau}J_\mu^{Z,w}(\psi) n^\mu  \ge b
\int_{\{t=\tau\}\cap \{r\ge r_0\}} (\tau^2 r^{-2}+1) \psi^2.
\]
From this and the previously derived bounds, 
pointwise decay can be shown  easily by applying $\Omega_i$ as commutators
and Sobolev estimates.
See~\cite{dr3} for details.

{\bf Exercise}: Derive pointwise decay for all derivatives of $\psi$, including transverse
derivatives to the horizon of any order, by commuting in addition with $\hat{Y}$ as in the proof
of Theorem~\ref{boundedn}.

\subsection{Comments and further reading}
\label{neocoms}
\subsubsection{The $X$-estimate}
The origin of the use of vector field multipliers of the type $X$ (as in Section~\ref{sie})
for proving decay for solutions of
the wave equation goes back to Morawetz. (These identities are generalisations of the
classical virial identity, which has itself a long and complicated
history.) In the context of Schwarzschild black holes,
the first results in the direction of such estimates were in Laba and Soffer~\cite{labasoffer}
for a certain ``Schr\"odinger'' equation (related to the Schwarzschild $t$-function),
and, for the wave equation, in Blue and Soffer~\cite{BlueSof0}.  These results were incomplete
(see~\cite{BlueSof}), however,  and the
first estimate of this type was actually obtained in our~\cite{dr3}, motivated
by the original calculations of~\cite{BlueSof0, labasoffer}. 
This estimate required decomposition of $\psi$ into 
individual spherical harmonics $\psi_\ell$, and choosing 
the current $J^{X,w}$ separately for each $\psi_\ell$. 
A slightly different approach to this estimate is provided by~\cite{BlueSof}.
A somewhat simpler choice of current $J^{X,w}$ which provides an estimate
for all sufficiently high spherical
harmonics was first presented by Alinhac~\cite{alinhac}.  
Our Section~\ref{hayhay} is similar in spirit. 
The first estimate not requiring a spherical harmonic
decomposition was obtained in~\cite{dr5}. This is the current of Section~\ref{oneforall}.
The problem of reducing the loss of derivatives in~$(\ref{finalestimate})$ 
has been addressed in Blue-Soffer~\cite{newblue}.\footnote{A related refinement,
where $h$ of $(\ref{finalotherestimate})$ is replaced by a function vanishing logarithmically
at $3M$, follows from~\cite{marzuola} referred to below.}
The results of~\cite{newblue} in fact also apply to the Reissner-Nordstr\"om metric.

A slightly different construction of a current 
as in Section~\ref{oneforall} has been given by Marzuola and 
collaborators~\cite{metcalfe}. This current
does not require commuting with $\Omega_i$.
In their subsequent~\cite{marzuola},~the considerations of~\cite{metcalfe} are combined with
ideas from~\cite{dr3,dr5} to obtain an estimate 
which does not degenerate on the horizon: One includes 
a piece of the current $J^N$ of Section~\ref{thevectorfields}
and exploits Proposition~\ref{Nproperties}.

\subsubsection{The $Z$-estimate}
The use of vector-field multipliers of the type $Z$ also goes back to
celebrated work of Morawetz,
in the context of the wave equation outside convex obstacles~\cite{mora}.
The geometric interpretation of this estimate arose later,
and the use of $Z$ adapted to the causal geometry of a non-trivial metric first appears perhaps
in the proof of stability of Minkowski space~\cite{book}.
The decay result Theorem~\ref{Schdec} 
was obtained in our~\cite{dr3}. A result yielding similar decay
away from the horizon (but weaker decay along the horizon) was proven independently in a nice 
paper of Blue and Sterbenz~\cite{BlueSter}. 
Both~\cite{BlueSter} and~\cite{dr3} make use of a current
based on the vector field $Z$. 
In~\cite{BlueSter}, the error term analogous to
$K^{Z,w}$ of Section~\ref{Morawetz} 
was controlled with the help of an auxiliary collection of multipliers with linear weights in $t$,
chosen at the level of each spherical harmonic, whereas
in~\cite{dr3}, these error terms are controlled directly from $(\ref{finalestimate})$ 
by a dyadic iteration scheme similar to the one we have given here in Section~\ref{Morawetz}.
The paper~\cite{BlueSter} does not obtain estimates for
the non-degenerate energy flux $(\ref{tote...})$; moreover,
a slower pointwise decay rate near the horizon is achieved
in comparison to Theorem~\ref{Schdec}. Motivated by~\cite{dr3}, 
the authors of~\cite{BlueSter} have since given a 
different argument~\cite{BlueSter1} to obtain just the pointwise estimate
$(\ref{ptb})$ on the horizon, exploiting the ``good'' term in $K^{Z,w}$ near the horizon. 
The proof of Theorem~\ref{Schdec} presented in Section~\ref{Morawetz}
 is a slightly modified version of the scheme in~\cite{dr3}, avoiding
 spherical harmonic decompositions (for obtaining $(\ref{finalestimate})$) by using
 in particular the result of~\cite{dr5}.

\subsubsection{Other results}
Statement~$(\ref{ptb})$ of  Theorem~\ref{Schdec} has been generalised to the Maxwell case by
Blue~\cite{BlueMax}. In fact, the Maxwell case is  much ``cleaner'', as the current $J^Z$
need not be modified by a function $w$, and its flux is pointwise
positive through any spacelike hypersurface. The considerations near the horizon
follow~\cite{BlueSter1} and thus the analogue of $(\ref{tote...})$ is not in fact obtained, only
decay for the degenerate flux of $J^T$. Nevertheless, the
non-degenerate $(\ref{tote...})$ for Maxwell can be proven following
the methods of this section, using in particular currents associated to the vector field $Y$
({\bf Exercise}).

To our knowledge, the above discussion exhausts the quantitative pointwise and energy 
decay-type statements which are known for general solutions of the wave equation
on Schwarzschild.\footnote{For fixed spherical harmonic $\ell=0$, there is also the quantitative
result of~\cite{dr1}, to be mentioned in Section~\ref{heuristic}.}
The best previously known results on general solutions of the wave equation were non-quantitative decay type statements which we briefly mention. A pointwise decay without a rate  was first proven
in the thesis of Twainy~\cite{twainy}.
Scattering and asymptotic completeness
statements  for the wave, Klein-Gordon, Maxwell
and Dirac equations have been
obtained by~\cite{dimock, dimockkay, bachelot, bachelotm, nicolas}.  These type of statements are typically insensitive to the amount of trapping.  See the related discussion of
Section~\ref{heuristic}, where the statement of Theorem~\ref{Schdec} is compared
to non-quantitative statements heuristically derived in the physics literature.

\subsection{Perturbing?}
\label{pertsec01}
Use of the $J^N$ current
``stabilises'' the proof of Theorem~\ref{Schdec} with respect to considerations near the horizon.
There is, however, a sense in which the above argument is still  fundamentally attached
to Schwarzschild.
The approach taken to derive the multiplier estimate $(\ref{cip})$
depends on the structure of the trapping set, in particular,
the fact that trapped null geodesics approach
 a codimension-$1$ subset of spacetime, the photon sphere. 
Overcoming the restrictiveness of this approach
is the fundamental remaining difficulty in extending
these techniques to Kerr, as will be accomplished in Section~\ref{KDsec}.
Precise implications of this fact for
multiplier estimates are discussed further in~\cite{alinhac}.

\subsection{Aside: Quantitative vs.~non-quantitative results and the heuristic tradition}
\label{heuristic}
The study of wave equations on Schwarzschild has a long history in the 
physics literature, beginning with the pioneering Regge and Wheeler~\cite{RW}.
These studies have all been associated with showing ``stability''.

A seminal paper is that of Price~\cite{rpr:ns}. There, insightful heuristic arguments were
put forth deriving the asymptotic tail of each spherical harmonic $\psi_\ell$ evolving
from compactly supported initial data,
suggesting that for $r>2M$,
\begin{equation}
\label{showing...}
\psi_\ell(r, t)\sim C_\ell t^{-(3+2\ell)}.
\end{equation}
These arguments were later extended by Gundlach et al~\cite{gundlach}
to suggest
\begin{equation}
\label{showing....}
\psi_\ell|_{\mathcal{H}^+}\sim C_\ell v^{-(3+2\ell)},\qquad
r\psi_\ell|_{\mathcal{I}^+} \sim \bar{C}_\ell u^{-(2+\ell)}.
\end{equation}
Another approach to these heuristics
via the analytic continuation of the Green's function was followed by~\cite{ching}. 
The latter approach in principle could perhaps be turned into a rigorous proof, at least for
solutions not supported on $\mathcal{H}^+\cap\mathcal{H}^-$. See~\cite{st:pc, kronthaler}
for just $(\ref{showing...})$ for the $\ell=0$ case. 

Statements of the form~$(\ref{showing...})$ are interesting because, if  proven,
they would give the fine structure of the tail of the solution. However,
it is important to realise that statements like~$(\ref{showing...})$
in of themselves would not give quantitative
bounds for the size of the 
solution at all later times in terms of initial data. In fact, the above heuristics
do not even suggest what the best such quantitative result would be, they only give a
heuristic \emph{lower}
bound on the best possible quantitative decay rate in a theorem 
like Theorem~\ref{Schdec}.

Let us elaborate on this further. For fixed spherical harmonic, by compactness
a statement of the form $(\ref{showing...})$ would immediately yield some
bound
\begin{equation}
\label{ineqs?}
|\psi_{\ell}|(r,t)\le D(r,\psi_\ell )t^{-3},
\end{equation}
for some constant $D$ depending on $r$ \emph{and on the solution itself}. 
It is not clear, however, what the sharp such quantitative
inequality of the form $(\ref{ineqs?})$ is supposed to be
when the constant is to depend on a natural quantity associated to data.
It is the latter, however, which is important for the nonlinear stability problem.

There is a setting in which a quantitative version of $(\ref{ineqs?})$ has indeed
been obtained:
The results of~\cite{dr1} (which apply to
the nonlinear problem where the scalar field is coupled to the Einstein equation,
but which can be specialised to the decoupled case of the $\ell=0$ harmonic on
Schwarzschild) prove in particular that 
\begin{equation}
\label{proved}
|n_{\Sigma_\tau}\psi_0|+ |\psi_0|\le C_\epsilon D(\uppsi,\uppsi') \tau^{-3+\epsilon},
\qquad
|r\psi_0|\le C D(\uppsi,\uppsi') \tau^{-2}
\end{equation}
where $C_\epsilon$ depends only on $\epsilon$, and $D(\uppsi,\uppsi')$ is a quantity 
depending only  on the initial $J^T$ energy and a pointwise weighted $C^1$ norm.
In view of the relation between $\tau$, $u$, and $v$, $(\ref{proved})$ includes also
decay on the horizon and null infinity as in the heuristically derived
$(\ref{showing....})$.
The fact that the power $3$ indeed appears
in both in the quantitative $(\ref{proved})$ and in $(\ref{ineqs?})$ may 
be in part accidental. See also~\cite{bizon}.

For general solutions, i.e.~for the sum over spherical harmonics,
the situation is even worse. In fact, a statement
like $(\ref{showing...})$ a priori gives no information whatsoever of any sort,
even of the non-quantitative kind. It is in principle
compatible with $\limsup_{t\to\infty} \psi(r,t)=\infty$.\footnote{Of course, \emph{given} 
the quantitative result of Theorem~\ref{boundedn} and the 
statement~$(\ref{showing...})$,
one could then infer 
that for 
each $r>2M$, then $\lim_{t\to \infty}\phi(r,t)= 0$, without however a rate ({\bf exercise}).}
It is well known, moreover, that to understand quantitative decay rates for general solutions,
one must quantify trapping. This is not, however, captured by the heuristics
leading to $(\ref{showing...})$, essentially because for fixed $\ell$, the effects of trapping
concern an intermediate time interval not reflected in the tail.
It should thus not be surprising that these heuristics do not address the fundamental 
problem at hand.

Another direction for heuristic work has been the study of so-called
\emph{quasi-normal modes}. These are solutions with time dependence $e^{-i\omega t}$
for $\omega$ with negative imaginary part, and appropriate boundary conditions. 
These occur as poles of the analytic continuation
of the resolvent of an associated elliptic problem, 
and in the scattering theory literature are typically known as \emph{resonances}.
Quasinormal modes are discussed  in the nice survey article of 
Kokkotas and Schmidt~\cite{Kokkotas}. Rigorous results on the distribution of 
resonances have been achieved in 
Bachelot--Motet-Bachelot~\cite{2bachelots}
and S\'a Barreto-Zworski~\cite{SB-Zworski}. 
The asymptotic distribution of the quasi-normal modes
as $\ell\to\infty$ can be thought to reflect trapping. On the other hand,
these modes fo not reflect the ``low-frequency'' effects
giving rise to tails. Thus, they too tell only part of the story. See, however,  the case of
Schwarzschild-de Sitter in Section~\ref{cosmolosec}.

Finally, we should mention Stewart~\cite{Stewart}. This is to our knowledge the first clear
discussion in the physics literature of the  relevance of trapping on the Schwarzschild metric
in this context
and the difference between quantitative and non-quantitative decay rates.
It is interesting to compare Section 3 of~\cite{Stewart} with what has now been proven:
Although the predictions of~\cite{Stewart} do not quite match the situation in 
Schwarzschild (it is in particular incompatible with~$(\ref{showing...})$), they
apply well to the  Schwarzschild-de Sitter case
developed in Section~\ref{cosmolosec}.

The upshot of the present discussion is the following: Statements of the form~$(\ref{showing...})$, 
while interesting, may 
have little to do with the problem of non-linear stability of black holes, and are perhaps
more interesting for the lower bounds that they suggest.\footnote{See for instance the
relevance of this in~\cite{cbh}.}
In fact, in view of their non-quantitative nature,
these results are less relevant for the stability problem
than the quantitative boundedness theorem of Kay and Wald.
Even the statement of Section~\ref{eazy} cannot be derived as a corollary of the 
statement~$(\ref{showing...})$, nor would knowing~$(\ref{showing...})$ simplify in 
any way the proof of Section~\ref{eazy}.

\section{Perturbing Schwarzschild: Kerr and beyond}
\label{pertsec}
We now turn to the problem of perturbing  the Schwarzschild metric and proving boundedness
and decay for the wave equation on the backgrounds of such perturbed metrics. 
Let us recall our dilemma: The boundedness argument of Section~\ref{S1}
required that $T$ remains causal everwhere in the exterior. In view of the comments of
Section~\ref{pertsec0},
this is clearly unstable. On the other hand, the decay argument of Section~\ref{S2} requires
understanding the trapped set and in particular, uses the fact that in Schwarzschild,
a certain codimension-$1$ subset of spacetime--the 
photon sphere--plays a special role. 
Again, as discussed in Section~\ref{pertsec01},
this special structure is unstable.

It turns out that nonetheless, these issues can be addressed and both boundedness
(see Theorem~\ref{kerbnd}) and decay (see Theorem~\ref{DT}) can be proven for 
the wave equation on suitable
perturbations of Schwarzschild.
As we shall see, the boundedness  proof (See Section~\ref{BP}) turns out to be more
robust and can be applied to a larger class of metrics--but it too 
requires some insight from the Schwarzschild decay argument! The decay 
proof (See Section~\ref{KDsec}) 
will require us to restrict to exactly Kerr spacetimes.

Without further delay, perhaps it is time to introduce the Kerr family\ldots

\subsection{The Kerr metric}
\label{thekerrmetric}
The \emph{Kerr metric} is a $2$-parameter family of metrics first discovered~\cite{kerr} 
in 1963.
The parameters are called \emph{mass} $M$ and specific angular momentum $a$,
i.e.~angular momentum per unit mass. In so-called
Boyer-Lindquist local
coordinates, the metric element takes the form:
\begin{align*}
-\left(1-\frac{2M}
{r\left(1+ \frac{a^2\cos^2\theta }{r^2}\right) } \right)   \,dt^2
+\frac{1+\frac{a^2\cos^2\theta}{r^2}}{1-\frac{2M}{r}+
\frac{a^2}{r^2}} \, dr^2 
+ r^2\left(1+\frac{a^2\cos^2\theta}{r^2}\right) \, d\theta^2\\
+r^2 \left (1+\frac{a^2}{r^2}
+\left(\frac{2M}{r}\right)\frac{a^2\sin^2\theta}{r^2\left(1+\frac{a^2\cos^2
\theta}{r^2}
\right)}\right)
\sin^2\theta\, d\phi^2\\
-4M\frac{a\sin^2\theta}{r\left(1+\frac{a^2\cos^2\theta}{r^2}\right)}\, 
dt  \, d\phi.
\end{align*}
The vector fields $\partial_t$ and $\partial_\phi$ are Killing. We say that
the Kerr family is 
\emph{stationary} and \emph{axisymmetric}.\footnote{There are various conventions on 
the meaning of the words ``stationary'' and ``axisymmetric'' depending on the context. Let
us not worry about this here\ldots}
Traditionally, one denotes
\[
\Delta=r^2-2Mr+a^2.
\]
If $a=0$, the Kerr metric clearly reduces to Schwarzschild $(\ref{incords})$.

Maximal extensions of the Kerr metric were first constructed by Carter~\cite{cartersep}.
For parameter range $0\le |a|< M$, these maximal
extensions have black hole regions and white hole regions
bounded by future and past event horizons $\mathcal{H}^\pm$ meeting at a bifurcate sphere. 
The above coordinate system is defined in a
domain of outer communications, and the horizon will correspond
to the limit $r\to r_+$, where $r_+$ is the larger positive root of $\Delta=0$,
i.e.
\[
r_+=M+\sqrt{M^2-a^2}.
\]

Since the motivation of our study is the Cauchy problem for the Einstein equations,
it is more natural to consider not maximal extensions, but maximal developments of complete initial
data. (See Appendix~\ref{cauchyproblem}.)
In the Schwarzschild case, the maximal development of initial data on a Cauchy
surface $\Sigma$ as described previously coincides with maximally-extended Schwarzschild.
In Kerr, if we are to take an asymptotically flat (with two ends) hypersurface in a maximally extended
Kerr for parameter range $0< |a|<M$, then its maximal development will have a smooth
boundary in maximally-extended Kerr. This boundary is what is known as a \emph{Cauchy horizon}.
We have already discussed this phenomenon in Section~\ref{experience} 
in the context of strong
cosmic censorship.
The maximally extended Kerr solutions are quite bizarre, in particular,
they contain closed timelike curves.
 This is of no concern to us  here, however.
By definition, for us the term ``Kerr metric $(\mathcal{M},g_{M,a})$'' will
always denote the maximal development of a complete asymptotically
flat hypersurface $\Sigma$, as above, with two ends. One can depict the Penrose-diagramatic
representation of a suitable two-dimensional timelike slice of this solution as below:
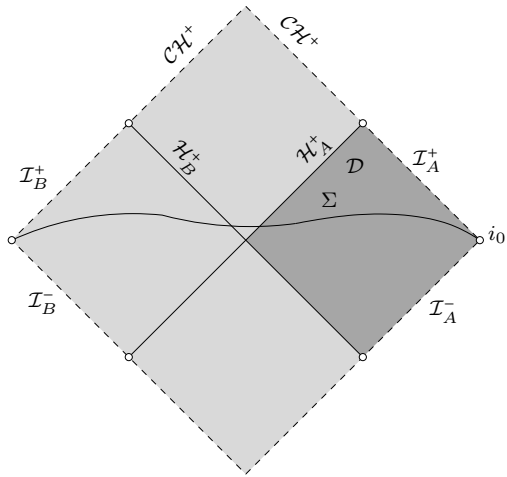
\[
\input{kerr.pstex_t}
\]
This depiction coincides with the standard  Penrose diagram of the spherically symmetric
Reissner-Nordstr\"om metric~\cite{he:lssst, Townsend}.

With this convention in mind, we note that the dependence
of $g_{M,a}$ on $a$ is smooth in the range $0\le |a| <M$.
In particular, Kerr solutions with small $|a|\ll M$ can be viewed as close
to Schwarzschild.

One can see this explicitly in the subregion of interest to us
by passing to
a new system of coordinates. Define
\[
t^*= t+  \bar t(r)
\]
\[
\phi^*= \phi + \bar \phi(r)
\]
where 
\[
\frac{d \bar t}{dr}(r)= (r^2+a^2)/\Delta^2, \qquad 
\frac{d \bar\phi}{dr}(r) = a/\Delta.
\]
(These coordinates are often known as \emph{Kerr-star coordinates}.)
These coordinates are regular
across $\mathcal{H}^+\setminus \mathcal{H}^-$.\footnote{Of course, one again
needs two coordinate systems in view of the breakdown of spherical coordinates. We shall
suppress this issue in the discussion that follows.}
We may finally define a coordinate $r_{\rm Schw}=r_{\rm Schw}(r, a)$ such that
which takes $[r_+,\infty)\to [2M,\infty)$ with smooth dependence in $a$ and such that
$r_{\rm Schw}(r,0)$ is the identity map.
In particular, if we define $\Sigma_0$ by $\mathcal{D}=\{t^*=0\}$, and define $\mathcal{R}
=\mathcal{D}\cap\{t^*\ge 0\}$,
and fix $r_{\rm Schw}$, $t^*$, $\phi^*$ Schwarzschild coordinates, then
the metric functions of $g_{M,a}$ written in terms of these coordinates as defined
previously depend smoothly on $a$ for $0\le |a|<M$ in $\mathcal{R}$, 
and, for $a=0$, reduce
to the Schwarzschild metric form in $(r,t^*,\phi,\theta)$ coordinates where $t^*$ is defined
from Schwarzschild $t$
as above.

We note that $\partial_t=\partial_{t^*}$ in the intersection of the coordinate systems.
We immediately note that $\partial_{t}$ is
spacelike on the horizon, except where 
$\theta=0,\pi$, i.e.~on the axis of symmetry. Note that we
shall often abuse notation (as we just have done)
and speak of $\partial_t$ on the horizon or at $\theta=0$, where of course the
$(r, t,\theta, \phi)$ coordinate
system breaks down, and formally, this notation is meaningless.

In general, the part of the domain of outer communications plus horizon
where $\partial_t$ is spacelike in known
as the \emph{ergoregion}. It is bounded by a hypersurface known as the
\emph{ergosphere}. The ergosphere meets the horizon on the axis of symmetry $\theta=0,\pi$.

The ergosphere allows for a particle ``process'', originally discovered by Penrose~\cite{Penr2},
for extracting energy out of a black hole. This came to be known as the
\emph{Penrose process}. In his thesis, Christodoulou~\cite{ri}
discovered the existence of a
quantity--the so-called \emph{irreducible mass} of the black  hole--which he should to be always
nondecreasing in a Penrose process. The analogy between this quantity and entropy
led later to a subject known as ``black hole thermodynamics''~\cite{bardeen, beken}.
This is currently the subject of intense investigation from the point of view
of high energy physics.

In the context of the study of $\Box_g\psi=0$, we have already 
discussed in Section~\ref{pertsec} the effect of the ergoregion:
It is precisely the presence of the ergoregion that makes
our previous proof of boundedness for Schwarzschild not 
immediately generalise for Kerr.
Moreover, in contrast to
 the Schwarzschild case, there is no ``easy result'' that one can obtain
away from the horizon analogous to Section~\ref{eazy}. 
In fact, the problem of proving any sort of boundedness  statement for general
solutions to $\Box_g\psi=0$
on Kerr had been open until very recently. We will 
describe in the next section our recent resolution~\cite{dr6} of this problem.

\subsection{Boundedness for axisymmetric stationary black holes}
\label{BP}
We will derive a rather general boundedness theorem for
a class of axisymmetric stationary
black hole exteriors near Schwarzschild. The result (Theorem~\ref{kerbnd})
will include slowly rotating 
Kerr solutions with parameters $|a|\ll M$.

We have already  explained in what sense the Kerr metric is ``close'' to 
Schwarzschild in the region $\mathcal{R}$. 
Let us note that with respect to the  coordinates $r_{\rm Schw}$, $t^*$, $\phi^*$, $\theta$ in 
$\mathcal{R}$,
then $\partial_{t^*}$ and $\partial_{\phi^*}$ are Killing for both the Schwarzschild
and the Kerr metric. The class of metrics which 
will concern us
here are metrics defined on
$\mathcal{R}$ such that the metric functions are close to Schwarzschild 
in a suitable sense\footnote{This requires moving to an auxiliary coordinate system. See~\cite{dr6}.},
and $\partial_{t^*}$, $\partial_{\phi^*}$ are Killing, where these are defined with respect
to the ambient Schwarzschild coordinates.

There is however an additional geometric assumption we shall need, and this is motivated
by a geometric property of the Kerr spacetime, to be described in the section that follows
immediately.

\subsubsection{Killing fields on the horizon}
Let us here remark a geometric property
of the Kerr spacetime itself which turns out to be of utmost importance in what
follows: Let $V$ denote a null generator of $\mathcal{H}^+$.
Then
\begin{equation}
\label{spanprop}
V\in {\rm Span}\{\partial_{t^*},\partial_{\phi^*}\}.
\end{equation}
There is a deep reason why this is true. For stationary black holes with non-degenerate
horizons, a celebrated argument of Hawking retrieves a second Killing field in the direction
of the null generator $V$. Thus, if $\partial_{t^*}$ and $\partial_{\phi^*}$ span the complete set of Killing fields,
then $V$  must evidently be in their span.

In fact,
choosing $V$ accordingly we have
\begin{equation}
\label{accordin}
V=      \partial_{t^*} + (a/2Mr_+) \partial_{\phi^*}  
\end{equation}

(For the Kerr solution, we have that there exists a timelike direction in the span 
of $\partial_{t^*}$ and $\partial_{\phi^*}$ for all points outside the horizon. 
We shall not explicitly make reference
to this property, although in view of Section~\ref{epilogue},
one can infer this property $({\bf exercise})$  for small perturbations
of Schwarzschild of the type considered here, i.e., given any point $p$ outside
the horizon, there exists
a Killing field $V$ (depending on $p$) such that $V(p)$ is timelike.)

\subsubsection{The axisymmetric case}
\label{axicase}
From $(\ref{accordin})$, it 
follows that there is a constant $\omega_0>0$, depending only on the parameters $a$ and $M$,
 such that if 
\begin{equation}
\label{if....}
|\partial_{t^*}\psi|^2\ge \omega_0|\partial_{\phi^*}\psi|^2,
\end{equation}
on $\mathcal{H}^+$,
then the 
flux satisfies
\begin{equation}
\label{STOKENTRO}
J^T_\mu (\psi)n^\mu_{\mathcal{H}^+}\ge 0.
\end{equation}
Note also that, for fixed $M$, we can take 
\begin{equation}
\label{asato0}
\omega_0\to 0,\qquad{\rm as}\qquad a\to 0.
\end{equation}

There is an immediate application of $(\ref{if....})$. Let us restrict
for the moment to axisymmetric solutions, i.e.~to $\psi$ such 
that $\partial_\phi\psi=0$. It follows that $(\ref{if....})$ trivially holds.
As a result, our argument proving boundedness is stable, i.e.~Theorem~\ref{boundedn}
holds for axisymmetric solutions of the wave equation on Kerr spacetimes with $|a|\ll M$.
(See the exercise of Section~\ref{pertsec0}.)
In fact, the restriction on $a$ can be be removed  ({\bf Exercise}, or go directly to
Section~\ref{epilogue}).

Let us note that the above considerations make sense not only for Kerr but for the more
general class of metrics on $\mathcal{R}$ close to Schwarzschild such that
$\partial_{t^*}$, $\partial_{\phi^*}$ are Killing, $\mathcal{H}^+$ is null and
$(\ref{spanprop})$ holds. In particular, $(\ref{if....})$ implies $(\ref{STOKENTRO})$,
where in  $(\ref{asato0})$, the condition $a\to 0$ is replaced by the condition that the metric
is taken suitably close to Schwarzschild. 
{\bf The discussion which follows will refer to metrics satisfying these assumptions.}\footnote{They are summarised again
in the formulation of Theorem~\ref{kerbnd}.} For simplicity, the reader can 
specialise the discussion below
to the case of a Kerr metric with $|a|\ll M$.

\subsubsection{Superradiant and non-superradiant frequencies}
\label{snsf}
There is a more general setting where we can make use of
$(\ref{if....})$. Let us suppose for the time being
that we could take the Fourier transform $\hat{\psi}(\omega)$ of our solution $\psi$
in $t^*$ and then expand in azimuthal modes $\psi_m$, i.e.~modes associated to the
Killing vector field $\partial_{\phi^*}$.

If we were to restrict $\psi$ to the frequency range 
\begin{equation}
\label{nonsuperrange}
|\omega|^2 \ge \omega_0 m^2,
\end{equation}
then $(\ref{if....})$ and thus $(\ref{STOKENTRO})$ holds after integrating
along $\mathcal{H}^+$. In view of this,  frequencies in the range $(\ref{nonsuperrange})$ 
are known as \emph{nonsuperradiant frequencies}.
The frequency range
\begin{equation}
\label{superrange}
|\omega|^2\le \omega_0 m^2
\end{equation}
determines the so-called \emph{superradiant frequencies}. In the physics literature, 
the main difficulty of this problem
has traditionally been perceived to ``lie'' with these frequencies.

Let us pretend for the time being
that using the Fourier transform, we could indeed decompose
\begin{equation}
\label{adecomp}
\psi=\psi_{\mbox{$\sharp$}}+ \psi_{\mbox {$\flat$}}
\end{equation}
where $\psi_{\mbox{$\sharp$}}$ is supported in $(\ref{nonsuperrange})$, whereas
$ \psi_{\mbox {$\flat$}}$ is supported in $(\ref{superrange})$. 

In view of the  discussion immediately above and the
 comments of Section~\ref{axicase}, it is plausible to expect that one
could indeed prove boundedness for $\psi_{\mbox{$\sharp$}}$ 
in the manner of the proof of Theorem~\ref{boundedn}. In particular,
if one could localise the integrated version of $(\ref{STOKENTRO})$
to arbitrary sufficiently large subsegments
$\mathcal{H}(\tau',\tau'')$, one could obtain
\begin{equation}
\label{sharpheur}
\int_{\Sigma_\tau}J^{n_{\Sigma_\tau}}_\mu( \psi_{\mbox{$\sharp$}} )n^\mu_{\Sigma_\tau}
\le B
\int_{\Sigma_0}J^{n_{\Sigma_0}}_\mu( \psi_{\mbox{$\sharp$}} )n^\mu_{\Sigma_0}.
\end{equation}

This would leave $\psi_{\mbox{$\flat$}}$.
Since this frequency range does not suggest a direct boundedness argument, it is natural to 
revisit the decay mechanism of Schwarzschild. We have already discussed 
(see Section~\ref{pertsec01})  the instability
of the decay argument; this instability arose from the structure of the set of
trapped null geodesics. 
At the heuristic level, however, it is
easy to see that, if one can take $\omega_0$ sufficiently small, then solutions supported 
in $(\ref{superrange})$ cannot be trapped. In particular, for $|a|\ll M$, {\bf superradiant frequencies
for $\Box_g\psi=0$ on Kerr are not trapped}. This will be the fundamental observation 
allowing for the boundedness theorem. 
Let us see how this statement can be understood from the point of view of
energy currents.

\subsubsection{A stable energy estimate for superradiant frequencies}
\label{sheur}

We continue here our heuristic point of view, where
we assume a decomposition  $(\ref{adecomp})$ where $\psi_{\mbox {$\flat$}}$
is supported in $(\ref{superrange})$.
In particular, 
one has an inequality
\begin{equation}
\label{hasan}
\int_{-\infty}^\infty \int_0^{2\pi}
\omega_0^2 (\partial_\phi\psi_{\mbox{$\flat$}})^2 \, d\phi^*\, dt^* \ge
\int_{-\infty}^\infty \int_0^{2\pi} (\partial_t\psi_{\mbox{$\flat$}})^2
 \, d\phi^*\, dt^*
\end{equation}
for all $(r,\theta)$.
We shall see below that $(\ref{hasan})$ allows us easily to construct a suitable stable current
for Schwarzschild.

It may actually be a worthwhile {\bf exercise} for the reader to come up with 
a suitable current for themselves. The choice is actually quite
flexible in comparison with the considerations of Section~\ref{sie}. Our choice (see~\cite{dr6}) 
is defined by 
\begin{equation}
\label{newmultip}
J^{\bf X} = eJ^N+ J^{X_a}+J^{X_b,w_b}
\end{equation}
where here, $N$ is the vector field of Section~\ref{thevectorfields}, 
$X_a= f_a\partial_{r^*}$, with
\begin{eqnarray*}
f_a&=& -r^{-4}(r_0)^{4}, \qquad {\rm\ for\ } r\le r_0\\
f_a&=&-1, \qquad {\rm\ for\ }  r_0 \le r\le R_1,\\
f_a&=& -1+\int_{R_1}^r \frac{d\tilde{r}}{4\tilde{r}} \qquad {\rm\ for\ } R_1\le
  r\le R_2,\\
f_a&=&0 {\rm\ for\ } r\ge R_2,
\end{eqnarray*}
$X_b=f_b\partial_{r^*}$ with
\[
f_b = \chi(r^*)\pi^{-1}\int_0^{r^*}\frac{\alpha}{x^2+\alpha^2}
\]
and $\chi(r^*)$ is a smooth cutoff with $\chi=0$ for $r^*\le0$ and $\chi=1$ for $r^*\ge1$.
Here $r$ and $r^*$ are Schwarzschild coordinates.\footnote{Since we are dealing now
with general perturbations of Schwarzschild, we shall now use $r$ for what we previously
denoted by $r_{\rm Schw}$. Note that in the special case that
our metric is Kerr, this $r$ is different from the Boyer-Lindquist $r$.}
 The function $w_b$ is given by
\[
w_b=f'_b+\frac{2}r(1-2M/r)(1-M/r) f_b.
\]
The parameters $e$, $\alpha$, $r_0$, $R_1$, $R_2$  must be chosen accordingly!

Restricting to the range $(\ref{superrange})$, using
$(\ref{hasan})$, with some computation we would obtain
\begin{equation}
\label{inallt}
\int_{-\infty}^{\infty}\int_{0}^{2\pi} K^{\bf X}(\psi_{\mbox{$\flat$}})
\,d\phi^*\, dt^*
\ge 
b\int_{-\infty}^\infty \int_{0}^{2\pi}
 \chi J^{n_\Sigma}_\mu(\psi_{\mbox{$\flat$}}) n_\Sigma^\mu\,  d\phi^* \, dt^*,
\end{equation}
for all $(r,\theta)$.

The above inequality can immediately be seen to be stable to small\footnote{Of course,
in view of the degeneration towards $i^0$, it is important that smallness is understood
in a weighted sense.} axisymmetric, stationary
perturbations of the Schwarzschild metric.
That is to say, for such metrics,
 if $\psi_{\mbox{$\flat$}}$ is supported in $(\ref{superrange})$
(where frequencies here are defined by Fourier transform
in coordinates $t^*$, $\phi^*$),
then the inequality $(\ref{inallt})$ holds as before.
 In particular, $(\ref{inallt})$ holds for
Kerr for small $|a| \ll M$.

How would $(\ref{inallt})$ give boundedness for $\psi_{\mbox{$\flat$}}$? We need
in fact to suppose  something slightly stronger,
namely that $(\ref{inallt})$ holds localised to $\mathcal{R}(0,\tau)$. 
Consider the currents
\[
J= J^T + e_2 J^{\bf X}, \qquad K=\nabla^\mu J_\mu,
\]
where $e_2$ is a positive parameters, and $J^N$ is the current of Section~\ref{thevectorfields}.
Then, for metrics $g$ close enough to Schwarzschild, and
for  $e_2$ sufficiently small,
we would have from a localised $(\ref{inallt})$ that
\[
\int_{\mathcal{R}(0,\tau)} K(\psi_{\mbox{$\flat$}}) \ge 0,
\]
\[
\int_{\mathcal{H}(0,\tau)} J_\mu(\psi_{\mbox{$\flat$}}) n^\mu_{\mathcal{H}} \ge 0,
\]
and thus
\[
\int_{\Sigma_\tau} J_\mu(\psi_{\mbox{$\flat$}})n^\mu_{\Sigma_\tau}
\le \int_{\Sigma_0} J_\mu(\psi_{\mbox{$\flat$}}) n^\mu_{\Sigma_0}.
\]
Moreover, for $g$ sufficiently close to Schwarzschild and $e_1$, $e_2$ suitably defined,
we also have ({\bf exercise})
\[
\int_{\Sigma_\tau} J_\mu^{n_{\Sigma_\tau}}(\psi_{\mbox{$\flat$}})
 n^\mu \le B \int_{\Sigma_\tau} J_\mu(\psi _{\mbox{$\flat$}})n_{\Sigma_\tau}^\mu.
\]
We thus would obtain
\begin{equation}
\label{superheur}
\int_{\Sigma_\tau} J_\mu^{n_{\Sigma_\tau}}(\psi_{\mbox{$\flat$}})
n^\mu \le B \int_{\Sigma_0} J_\mu^{n_{\Sigma_0}}(\psi_{\mbox{$\flat$}})
n^\mu.
\end{equation}
Adding $(\ref{superheur})$ and $(\ref{sharpheur})$, we would obtain
\[
\int_{\Sigma_\tau} J_\mu^{n_{\Sigma_\tau}}(\psi)
n^\mu \le B \int_{\Sigma_0} J_\mu^{n_{\Sigma_0}}(\psi)
n^\mu
\]
provided that we could also estimate say
\begin{equation}
\label{nottrue}
\int_{\Sigma_0} J_\mu^{n_{\Sigma_0}}(\psi_{\mbox{$\sharp$}})
n^\mu \le B\int_{\Sigma_0} J_\mu^{n_{\Sigma_0}}(\psi)
n^\mu.
\end{equation}

\subsubsection{Cutoff and decomposition} 
\label{cutoffs}
Unfortunately, things are not so simple!

For one thing, to take the Fourier transform necessary
to decompose in frequency, one would need to
know a priori that $\psi(t^*,\cdot)$ is in $L^2(t^*)$.
What we want to prove at this stage is much less.
A priori, $\psi$ can grow exponentially in $t^*$.
In order to apply the above, one must cut off the solution appropriately in time.

This is achieved as follows. For definiteness,
define $\Sigma_0$ to be $t^*=0$, and
$\Sigma_\tau$ as before. We will also need two auxiliary families of hypersurfaces
defined as follows. (The motivation for considering these will be discussed in
Section~\ref{bootstrapsec}.) Let  $\chi$
be  a cutoff such that $\chi(x)=0$ for $x\ge 0$ and $\chi=1$ for $x\le -1$,
and define $t^\pm$
by
\[
t^{+}= t^*- \chi(-r+R)(1+r-R)^{1/2}
\]
and 
\[
t^{-}= t^* + \chi(-r+R)(1+r-R)^{1/2}
\]
where $R$ is a large constant, which must be chosen appropriately.
Let us define then
\[
\Sigma^{+}(\tau)\doteq\{t^{+} =\tau\},
\qquad
\Sigma^{-}(\tau)\doteq\{t^{-} =\tau\}.
\]
Finally, we define
\[
\mathcal{R}(\tau_1,\tau_2) =\bigcup_{\tau_1\le \tau\le \tau_2} \Sigma(\tau),
\]
\[
\mathcal{R}^+(\tau_1,\tau_2)=\bigcup_{\tau_1\le \tau\le \tau_2} \Sigma^+(\tau),
\]
\[
\mathcal{R}^-(\tau_1,\tau_2)=\bigcup_{\tau_1\le \tau\le \tau_2} \Sigma^-(\tau).
\]

Let
$\xi$ now be a cutoff function
such that $\xi=1$ in $J^+(\Sigma^-_1)\cap J^-(\Sigma^+_{\tau-1})$,
and $\xi=0$ in $J^+(\Sigma^+_\tau)\cap J^-(\Sigma^-_0)$. 
We may finally define
\[
\psi_{\hbox{\Rightscissors}}=  \xi  \psi.
\]
The function $\psi_{\hbox{\Rightscissors}}$ 
is a solution of the inhomogeneous equation
\[
\Box_g
\psi_{\hbox{\Rightscissors}}
=F,\qquad F=2 \nabla^\a\xi\, \nabla_\a \psi + \Box_g\xi \, \psi.
\]
Note that $F$ is supported in $\mathcal{R}^-(0,1)\cup\mathcal{R}^+(\tau-1,\tau)$.

Another problem is that sharp cutoffs in frequency behave poorly under localisation. 
We thus do the following:
Let $\zeta$ be a smooth cutoff supported in $[-2,2]$ with the property
that $\zeta=1$ in $[-1,1]$, and let $\omega_0>0$ be a parameter to be determined
later. 
For an arbitrary $\Psi$ of compact support in $t^*$, define
\[
\Psi_{\mbox{$\flat$}} (t^*,\cdot)
\doteq\sum_{m\ne 0} e^{im \phi^*} \int_{-\infty}^{\infty} \zeta((\omega_0m)^{-1}\omega)\, \hat\Psi_{m}(\omega,\cdot) \, e^{i\omega t^*}d\omega,
\]
\[
\Psi_{\mbox{$\sharp$}}(t^*,\cdot) \doteq
\Psi_0+
 \sum_{m\ne 0} e^{im \phi^*} 
 \int_{-\infty}^{\infty} \left(1-\zeta((\omega_0m)^{-1}\omega)\right)\, \hat\Psi_{m}
 (\omega,\cdot)\, e^{i\omega t^*}d\omega.
\]
Note of course that $\Psi_{\mbox{$\sharp$}}+\Psi_{\mbox{$\flat$}}=\Psi$.
We shall use the notation $\psi_{\mbox{$\flat$}}$ for $(\psi_{\hbox{\Rightscissors}})_{\mbox{$\flat$}}$
and $\psi_{\mbox{$\sharp$}}$ for $(\psi_{\hbox{\Rightscissors}})_{\mbox{$\sharp$}}$.
Note that $\psi_{\mbox{$\flat$}}$, $\psi_{\mbox{$\sharp$}}$ satisfy
\begin{equation}
\label{inhomo}
\Box_g\psi_{\mbox{$\flat$}}=F_{\mbox{$\flat$}}, \qquad
\Box_g\psi_{\mbox{$\sharp$}}=F_{\mbox{$\sharp$}}.
\end{equation}

\subsubsection{The bootstrap}
\label{bootstrapsec}
With $\psi_{\mbox{$\flat$}}$, $\psi_{\mbox{$\sharp$}}$ well defined,
we now try to fill in the argument heuristically outlined before.

We wish to show the boundedness of
\begin{equation}
\label{spatial}
{\bf q}\doteq \sup_{0\le \bar\tau\le \tau} \int_{\Sigma_{\bar\tau}} J^N_\mu n^\mu.
\end{equation}
We will argue by continuity in $\tau$. We have already seen heuristically
how to obtain a bound for ${\bf q}$ in Sections~\ref{snsf} and~\ref{sheur}. When interpreted for the
$\psi_{\mbox{$\flat$}}$, $\psi_{\mbox{$\sharp$}}$ defined above, these arguments
produce error terms from:
\begin{itemize}
\item
the inhomogeneous terms $F_{\mbox{$\flat$}}$, $F_{\mbox{$\sharp$}}$ from $(\ref{inhomo})$
\item
the fact that we wish to 
localise        estimates $(\ref{STOKENTRO})$ and $(\ref{hasan})$
 to subregions  $\mathcal{H}^+(\tau',\tau'')$ and $\mathcal{R}(\tau',\tau'')$ 
 resepectively
\item
the fact that $(\ref{nottrue})$ is not exactly true.
\end{itemize}
These error terms can be controlled by ${\bf q}$ itself.
For this, one studies carefully the time-decay of
$F_{\mbox{$\flat$}}$, $F_{\mbox{$\sharp$}}$
away from the cutoff region $\mathcal{R}^-(0,1)\cup \mathcal{R}^+(\tau-1,\tau)$
using classical properties of the Fourier transform.
An important subtlety arises from the presence of $0$'th order terms
in $\psi$, and it is here that the divergence of the region   
$\mathcal{R}^\pm$ from $\mathcal{R}(0,\tau)$ is exploited to exchange decay
in $\tau$ and $r$.

To close the continuity argument, it is essential not only that the error terms be
controlled by ${\bf q}$ itself, but  that a small constant is retrieved, i.e.~that the
error terms are controlled by $\epsilon{\bf q}$, so that they can be absorbed.
For this, use is made of the fact that for metrics in the allowed class
sufficiently close to Schwarzschild (in the Kerr case,  for $|a|\ll M$),
one can control a priori the exponential 
growth rate of $(\ref{spatial})$ to be small. See~\cite{dr6}.

\subsubsection{Pointwise bounds}
Having proven the uniform boundedness of $(\ref{spatial})$,
one argues as in the proof of Theorem~\ref{boundedn}  to obtain higher order energy
and pointwise bounds.
In particular, the positivity property in the computation of Proposition~\ref{poscompu}  is stable.
(It turns out that this positivity property persists in fact for much more general
black hole spacetimes and
there is in fact a geometric reason for this! See Chapter~\ref{epilogue}.)

\subsubsection{The boundedness theorem}
We have finally
\begin{theorem}
\label{kerbnd}
Let $g$ be a metric defined on the differentiable manifold
$\mathcal{R}$ with stratified boundary $\mathcal{H}^+\cup \Sigma_0$, and let
$T$ and $\Phi=\Omega_1$ be Schwarzschild Killing fields.
Assume
\begin{enumerate}
\item
\label{closeNESS}
$g$ is sufficiently close to Schwarzschild in an appropriate sense
\item
$T$ and $\Phi$ are Killing with respect to $g$
\item
$\mathcal{H}^+$ is null with respect to $g$ and
$T$ and $\Phi$ span the null generator of $\mathcal{H}^+$.
\end{enumerate}
Then the statement of Theorem~\ref{boundedn} holds.
\end{theorem}
See~\cite{dr6} for the precise formulation of the closeness assumption~\ref{closeNESS}.
\begin{corollary}
The result applies to Kerr, and to the more general Kerr-Newman family (solving Einstein-Maxwell),
for parameters $|a|\ll M$ (and also $|Q|\ll M$ in the Kerr-Newman case).
\end{corollary}

Thus, we have quantitative pointwise and energy bounds for $\psi$ and arbitrary derivatives on
slowly rotating Kerr and Kerr-Newman exteriors.

\subsection{Decay for Kerr}
\label{KDsec}
To obtain decay results analogous to Theorem~\ref{Schdec}, one needs to 
understand trapping. For general perturbations of Schwarzschild of the class considered
in Theorem~\ref{kerbnd}, 
it is not a priori clear what stability properties one can infer about the nature of the
trapped set, and how these can be exploited. 
But for the Kerr family itself, the trapping structure can easily be understood, in view of the 
complete integrability of geodesic flow discovered by Carter~\cite{cartersep}.
The codimensionality of the trapped set persists, but in contrast to the Schwarzschild case
where trapped null geodesics all approach the codimensional-$1$ subset $r=3M$ of
spacetime, in Kerr, this codimensionality must be viewed in phase space.

\subsubsection{Separation}   
There is a convenient way of doing phase space analysis in Kerr spacetimes, namely,
as discovered by Carter~\cite{cartersep2}, 
the wave equation can be separated.  Walker and Penrose~\cite{walker} later showed that
both the complete integrability of geodesic flow and the separability of the wave equation
have their fundamental origin in the presence of a 
\emph{Killing tensor}.\footnote{See~\cite{chong, kubiznak} 
for recent higher-dimensional generalisations
of these properties.} In fact,  as we shall see, in view of its intimate relation with the 
integrability of geodesic flow, Carter's separation of $\Box_g$ immediately captures
the codimensionality of the trapped set.

The separation of the wave equation requires taking the Fourier transform, and
then expanding into oblate spheroidal harmonics. 
As before, taking the Fourier transform requires cutting off in time.
We shall here do the cutoff, however, in a somewhat different fashion.

Let ${\Sigma}_\tau$ be defined specifically as $t^*=\tau$. Given $\tau'<\tau$, 
define ${\mathcal{R}}(\tau',\tau)$ as before,
and let $\xi$ be a cutoff function as in Section~\ref{cutoffs}, but with
$\Sigma_{\tau'+1}$ replacing $\Sigma^-_1$, $\Sigma_{\tau'}$ replacing $\Sigma^-_0$,
and $\Sigma_\tau$ replacing $\Sigma^+_\tau$, $\Sigma_{\tau-1}$ replacing
$\Sigma^+_{\tau-1}$.
Define as before
\[
\psi_{\hbox{\Rightscissors}}=  \xi  \psi.
\]
The function $\psi_{\hbox{\Rightscissors}}$ 
is a solution of the inhomogeneous equation
\[
\Box_g
\psi_{\hbox{\Rightscissors}}
=F,\qquad F=2 \nabla^\a\xi\, \nabla_\a \psi + \Box_g\xi \, \psi.
\]
Note that $F$ is supported in ${\mathcal{R}}(\tau',\tau'+1)\cup
{\mathcal{R}}(\tau-1,\tau)$.

Since $\psi_{\hbox{\Rightscissors}}$ is compactly supported in $t^*$ we may
consider its Fourier transform 
$\hat\psi_{\hbox{\Rightscissors}}=\hat\psi_{\hbox{\Rightscissors}}(\omega,\cdot)$.
We may now decompose
\[
\hat\psi_{\hbox{\Rightscissors}}(\omega,\cdot)=
\sum_{m,\ell}R^\omega_{m\ell}(r) S_{m\ell}(a\omega,\cos\theta)e^{im\phi^*},
\]
\[
\hat{F}(\omega,\cdot)=
\sum_{m,\ell}F^\omega_{m\ell}(r) S_{m\ell}(a\omega,\cos\theta)e^{im\phi^*},
\]
where $S_{m\ell}$ are the oblate spheroidal harmonics.
For each $m\in \mathbb Z$, and fixed $\omega$,
these are a basis of
eigenfunctions $S_{m\ell}$ satisfying
\[
-\frac 1{\sin\theta} \frac{d}{d\theta} \left (\sin\theta \frac{d}{d\theta}S_{m\ell}\right)+\frac {m^2}{\sin^2\theta}
S_{m\ell}- a^2\omega^2 \cos^2\theta S_{m\ell} = \lambda_{m\ell} S_{m\ell} ,
\]
and, in addition,
satisfying the orthogonality conditions
with respect to the $\theta$ variable,
\[
\int_0^{2\pi} d\varphi 
\int_{-1}^1 d(\cos\theta) e^{im\phi^*} 
S_{m\ell} (a\omega,\cos\theta) \,e^{-im'\phi^*} S_{m'\ell'}(a\omega,\cos\theta)=
\de_{mm'} \de_{\ell\ell'}.
\]
Here, the $\lambda_{m\ell}(\omega)$ are the eigenvalues associated with
the harmonics $S_{m\ell}$. 
Each of the functions $R_{m\ell}^\omega(r)$ is a solution of the following problem
$$
\Delta \frac{d}{dr} \left (\Delta \frac{R_{m\ell}^\omega}{dr}\right) + \left (a^2m^2 + (r^2+a^2)^2\omega^2-\Delta(\lambda_{m\ell}+a^2\omega^2) \right) R_{m\ell}^\omega=(r^2+a^2)\Delta
F_{m\ell}^\omega.
$$

Note that if $a=0$, we typically label $S_{m\ell}$ by $\ell\ge |m|$
such that
\[
\lambda_{m\ell}(\omega)= \ell(\ell+1)/2.
\]
With this choice, $S_{m\ell}$ coincides with the standard spherical
harmonics $Y_{m\ell}$.

Given any $\omega_1>0$, $\lambda_1>0$ 
then we can choose $a$ such that
for $|\omega|\le \omega_1$, $\lambda_{ m\ell}\le \lambda_1$, then
\[
|\lambda_{m\ell}- \ell(\ell+1)/2 |\le  \epsilon.
\]

Rewriting the equation for the oblate spheroidal function
$$
-\frac 1{\sin\theta} \frac{d}{d\theta} \left (\sin\theta \frac{d}{d\theta}S_{m\ell}\right)+\frac {m^2}{\sin^2\theta}
S_{m\ell} = \lambda_{m\ell} S_{m\ell} + a^2\omega^2 \cos^2\theta S_{m\ell},
$$
the smallest eigenvalue of the operator on the left hand side of the above equation is $m(m+1)$.
This implies that 
\begin{equation}
\label{Totherestimate}
\lambda_{m\ell} \ge m(m+1)-a^2\omega^2.
\end{equation} 
This will be all that we require about $\lambda_{m\ell}$. 
For a more detailed analysis of $\lambda_{m\ell}$, see~\cite{spec}.

\subsubsection{Frequency decomposition}
Let $\zeta$ be a
sharp cutoff function such that $\zeta=1$ for $|x|\le 1$ and $\zeta=0$ for $|x|> 1$.
Note that 
\begin{equation}
\label{sharpcutoff}
\zeta^2=1.
\end{equation}
Let $\omega_1$, $\lambda_1$ be (potentially large) constants to be determined, and $\lambda_2$
be a (potentially small) constant to be determined.

Let us define
\[
\psi_{\mbox {$\flat$}}=
\int_{-\infty}^\infty
\zeta(\omega/\omega_1) 
\sum_{m,\ell:\lambda_{m\ell}(\omega)\le \lambda_1} 
R^\omega_{m\ell}(r) S_{m\ell}(a\omega,\cos\theta) e^{im\phi^*} e^{i\omega t^*} d\omega,
\]
\[
\psi_{\lessflat}=\int_{-\infty}^\infty
\zeta(\omega/\omega_1)
 \sum_{m,\ell:\lambda_{m\ell}(\omega)> \lambda_1} 
R^\omega_{m\ell}(r) S_{m\ell}(a\omega,\cos\theta) e^{im\phi^*} e^{i\omega t^*} d\omega,
\]
\[
\psi_{\mbox{$\natural$}}=\int_{-\infty}^\infty(1-\zeta(\omega/\omega_1)) \sum_{m,\ell:\lambda_{m\ell}(\omega)\ge  \lambda_2\omega^2} 
R^\omega_{m\ell}(r) S_{m\ell}(a\omega,\cos\theta) e^{im\phi^*} e^{i\omega t^*} d\omega,
\]
\[
\psi_{\mbox{$\sharp$}}=\int_{-\infty}^\infty(1-\zeta(\omega/\omega_1)) \sum_{m,\ell:\lambda_{m\ell}(\omega)< \lambda_2\omega^2} 
R^\omega_{m\ell}(r) S_{m\ell}(a\omega,\cos\theta) e^{im\phi^*} e^{i\omega t^*}  d\omega.
\]

We have clearly
\[
\psi_{\hbox{\Rightscissors}}= \psi_{\mbox {$\flat$}}+ \psi_{\lessflat} + \psi_{\mbox{$\natural$}} + \psi_{\mbox{$\sharp$}}.
\]

For quick reference, we note: 
\begin{itemize}
\item
$\psi_{\mbox {$\flat$}}$ is supported in $|\omega|\le \omega_1$, $\lambda_{m\ell} \le \lambda_1$,
\item
$\psi_{\lessflat}$ is supported in $|\omega|\le \omega_1$, $\lambda_{m\ell} > \lambda_1$,
\item
$\psi_{\mbox {$\natural$}}$ is supported in $|\omega|\ge \omega_1$, 
$\lambda_{m\ell} \ge \lambda_2\omega^2$
and
\item
$\psi_{\mbox{$\sharp$}}$ is supported in $|\omega|\ge \omega_1$, $\lambda_{m\ell}< \lambda_2\omega^2$.
\end{itemize}

\subsubsection{The trapped frequencies}
\label{kloubi}

Trapping takes place in $\psi_{\mbox {$\natural$}}$. We show here how to 
construct a multiplier for this frequency range.

Defining a coordinate $r^*$ by
\[
\frac{dr^*}{dr}=\frac{r^2+a^2}{\Delta}
\]
and setting
\[
u(r)=(r^2+a^2)^{1/2}
 R^\omega_{m\ell} (r),\qquad 
 H(r)=\frac{\Delta F^\omega_{m\ell}(r)}{(r^2+a^2)^{1/2}}, 
\]
then $u$ satisfies
\[
\frac{d^2}{(dr^*)^2}u+(\omega^2 - V^\omega_{m\ell }(r))u = H
\]
where
\[
V^\omega_{m \ell}(r)= \frac{4Mram\omega-a^2m^2+\Delta (\lambda_{m\ell}+\omega^2a^2)}{(r^2+a^2)^2}
+\frac{\Delta(3r^2-4Mr+a^2)}{(r^2+a^2)^3}
-\frac{3\Delta^2 r^2}{(r^2+a^2)^4}.
\]
Consider the following quantity 
$$
Q=
f \left ( \left|\frac {du}{dr^*}\right|^2 + (\omega^2-V ) |u|^2\right) + \frac {df}{dr^*} 
{\rm Re}\left(\frac{du}{dr^*} \bar u\right)-
\frac 12 \frac{d^2f}{{dr^*}^2} |u|^2.
$$
Then, with the notation $'=\frac{d}{dr^*}$,
\begin{equation}\label{eq:Q}
Q'= 2 f' |u'|^2  - f V' |u|^2 + {\rm Re}(2 f \bar{H} u'+f' \bar{H} u)-\frac 12 f''' |u|^2.
\end{equation}

For $\psi_{\mbox {$\natural$}}$, we have
\begin{equation}
\label{yilyil}
\lambda_{m\ell} +\omega^2 a^2 \ge (\lambda_2+a^2) \omega^2 \ge (\lambda_2 +a^2)\omega_1^2.
\end{equation}

We set 
$$
V_0=({\lambda_{m\ell} +\omega^2 a^2})\frac {r^2-2Mr}{(r^2+a^2)^2}
$$
so that 
$$
V_1=V-V_0= \frac{4Mram\omega -a^2m^2+a^2(\lambda_{m\ell}+\omega^2a^2)}{(r^2+a^2)^2}+
\frac{\Delta(3r^2-4Mr+a^2)}{(r^2+a^2)^3} -\frac{3\Delta^2 r^2}{(r^2+a^2)^4}.
$$
Using $(\ref{Totherestimate})$, $(\ref{yilyil})$, we easily see that 
\begin{eqnarray}
\label{AP}
\nonumber
r^3 |V_1'|+ \left|\left (\frac {(r^2+a^2)^4}{\Delta r^2} V_1'\right)'\right|&\le& C {\Delta} r^{-2}
\left(|a m\omega| + a^2 (\lambda_{m\ell}+a^2\omega^2)+1\right)\\
&\le& \epsilon \Delta r^{-2}(\lambda_{m\ell}+a^2\omega^2),
\end{eqnarray}
where $\epsilon$ can be made arbitrarily small, if $\omega_1$ is chosen sufficiently large, and $a$
is chosen $a<\epsilon$.
On the other hand
\begin{align}
\label{elinde}
\nonumber
V_0'&= 2\frac{\Delta}{(r^2+a^2)^4} ({\lambda_{m\ell} +\omega^2 a^2}) \left ( (r-M)(r^2+a^2) - 2r(r^2-2Mr)\right)
\\ &= -2\frac{\Delta r^2}{(r^2+a^2)^4} \left({\lambda_{m\ell} +\omega^2 a^2}\right)
\left(r-3M+a^2\frac {r-M}{r^2}
\right).
\end{align}
This computation implies that $V_0'$ has a simple zero in the $a^2$ neighborhood of $r=3M$. 
Furthermore,
$$
\left(\frac{(r^2+a^2)^4}{\Delta r^2} V_0'\right)'\le - \Delta r^{-2} ({\lambda_{m\ell} +\omega^2 a^2}).
$$
From the above and $(\ref{AP})$, it follows that for $\omega_1$
sufficiently large and $a$ sufficiently small,
we have
$$
\left(\frac{(r^2+a^2)^4}{\Delta r^2} V'\right)'\le -\frac12
\Delta r^{-2}({\lambda_{m\ell} +\omega^2 a^2}).
$$
This alone implies that $V'$ has \emph{at most} a simple zero.

To show that $V'$ indeed has a zero we examine the boundary values at $r_+$ and $\infty$.
From $(\ref{elinde})$ we see that
 \[
 \frac{(r^2+a^2)^4}{\Delta r^2}V_0'\sim
 C({\lambda_{m\ell} +\omega^2 a^2})
 \]
 for some positive constant $C$ 
on the horizon and 
\[
 \frac{(r^2+a^2)^4}{\Delta r^2}V_0'\sim
-2r({\lambda_{m\ell} +\omega^2 a^2})
\]
near $r=\infty$. 
On the other hand, from the inequality as applied to the first term on the right hand
side of $(\ref{AP})$, it follows that
\[
\left|\frac{(r^2+a^2)^4}{\Delta r^2}V_1'\right| \le
\epsilon r ({\lambda_{m\ell} +\omega^2 a^2}),
\]
where $\epsilon$ can be chosen arbitrarily small if $\omega_1$ is chosen sufficiently large and $a$ sufficiently small.
Thus, for suitable choice of $\omega_1$, 
it follows that
\begin{eqnarray*}
\frac{(r^2+a^2)^4}{\Delta r^2}V'\Big|_{r_+}&=&
\frac{(r^2+a^2)^4}{\Delta r^2}(V_0'+V_1')\Big|_{r_+}\\
&>&0
>\frac{(r^2+a^2)^4}{\Delta r^2}(V_0'+V_1')\Big|_{\infty}
=\frac{(r^2+a^2)^4}{\Delta r^2}V'\Big|_{\infty},
\end{eqnarray*}
and thus $V'$ has a unique zero. Let us denote the $r$-value of this zero by
$r_{m\ell}^\omega$.

We now choose $f$ so that
\begin{enumerate}
\item $f'\ge 0$,
\item $f\le 0$ for $r\le r_{m\ell}^\omega$ and $f\ge 0$ for $r\ge r_{m\ell}^\omega$ ,
\item
\label{proper3}
 $-fV'-\frac 12 f'''\ge c$.
\end{enumerate}
Property~\ref{proper3}
can be verified by ensuring that $f'''(r_{m\ell}^\omega)<0$ as well as requiring that
$f'''<0$ at the horizon. We may moreover normalise 
$f$ to $-1$ on the horizon. Finally, we may assume that
there exists an $R$ such that for all $r\ge R$, 
$f$ is of the form:
\[
f= \tan^{-1} \frac{r^*-\alpha-\sqrt{\alpha}}{\alpha}- \tan^{-1}(-1-\alpha^{-1/2})
\]
{\bf In particular, for $r\ge R$, the function $f$  will not depend on $\omega$, $\ell$, $m$.}

Note the similarity of this construction with that of Section~\ref{hayhay}, modulo the need
for complete separation
to centre the function $f$ appropriately.

Integrating the identity \eqref{eq:Q} and using that $u\to 0$ as $r\to \infty$ we obtain that for any compact
set $K_1$ in $r^*$ and a certain compact set $K_2$ (which in particular does not contain
$r=3M$), there exists a positive constant
$b>0$ so that
\begin{align*}
b\int_{K_1} &(|u'|^2+|u|^2) dr + b(\lambda_{m\ell}+\omega^2) 
\int_{K_2} |u|^2 dr\\
& \le \left (|u'|^2+(\omega^2-V) |u|^2\right)(r_+)+  \int {\rm Re}(2 f \bar{H} u'+f' \bar{H} u)\, dr.
\end{align*}
On the horizon $r=r_+$, we have $u'=(i\omega +(iam/2Mr_+))u$ and 
$$
V(r_+)= \frac {4Mram\omega-a^2m^2}{(r_+^2+a^2)^2}.
$$
Therefore, we obtain
\begin{align}
\nonumber
\label{beforewesum}
b&\int_{K_1} (|u'|^2+|u|^2) \, dr^* + b(\lambda_{m\ell}+\omega^2) 
\int_{K_2} |u|^2 \,  dr^* \\
&\le(\omega^2 +\epsilon m^2)|u|^2(r_+)+  \int {\rm Re} (2 f \bar H u'+f'  \bar H u)\, dr^*.
\end{align}

We now wish to reinstate the dropped indices $m,\ell,\omega$, and sum over $m$, $\ell$
and integrate over $\omega$.
Note that by the orthogonality of the $S^\omega_{m\ell}$,
it follows 
that for any functions $\alpha$ and $\beta$ with coefficients defined by
\[
\hat{\alpha}(\omega,\cdot)=
\sum_{m,\ell}\alpha^\omega_{m\ell}(r) S_{m\ell}(a\omega,\cos\theta)e^{im\phi^*}, \qquad
\hat{\beta}(\omega,\cdot)=
\sum_{m,\ell}\beta^\omega_{m\ell}(r) S_{m\ell}(a\omega,\cos\theta)e^{im\phi^*},
\]
we have
\begin{align*}
&\int \alpha^2(t^*,r,\theta,\varphi) \sin\theta d\varphi\, d\theta\, dt^*=\int_{-\infty}^\infty \sum_{m,\ell} 
|\alpha^\omega_{m\ell}(r)|^2 
d\omega,\\
&\int \alpha\cdot\beta \sin\theta d\varphi\, d\theta\, dt^*=
\int_{-\infty}^\infty \sum_{m,\ell}  \alpha^\omega_{m\ell}\cdot\bar \beta^\omega_{m\ell}
d\omega.
\end{align*}

Clearly, the summed and integrated left hand side of  $(\ref{beforewesum})$ bounds
\[
b\int_{-\infty}^\infty dt^*
\int_{K_1}\left ((\pa_{r}  \psi_{\mbox {$\natural$}})^2+ \psi_{\mbox {$\natural$}}^2\right) \, dV_g+
b\int_{K_2}\sum_{i} (\pa_i \psi_{\mbox {$\natural$}})^2 \, dV_g.
\]
Similarly, we read off immediately that the first term on the right hand 
side of $(\ref{beforewesum})$ upon summation
and integration yields precisely
\[
 \int_{{\mathcal H}_+} \left(
(T \psi_{\mbox {$\natural$}})^2+\epsilon (\partial_{\phi^*}\psi_{\mbox {$\natural$}})^2 \right).
\]
Note that we can bound 
\begin{eqnarray*}
 \int_{{\mathcal H}_+} \left(
(T \psi_{\mbox {$\natural$}})^2+\epsilon (\partial_{\phi^*}\psi_{\mbox{$\natural$}})^2 \right)
&\le&
 \int_{{\mathcal H}^+} \left(
(T \psi_{\hbox{\Rightscissors}})^2+\epsilon (\partial_{\phi^*}\psi_{\hbox{\Rightscissors}})^2 \right)\\
&\le&
B \int_{\Sigma_{\tau'} }J^N_\mu(\psi)n^\mu_{\Sigma}+
\epsilon \int_{\mathcal{H}(\tau',\tau)} (\partial_{\phi^*}\psi)^2
\end{eqnarray*}
({\bf Exercise}: Why?)

The ``error term'' of the right hand side of $(\ref{beforewesum})$ is more tricky.
To estimate the second summand of the integrand, 
note that
\begin{eqnarray*}
 \int_{-\infty}^\infty&& \sum_{m,\ell:\lambda_{m\ell}(\omega)\ge  \lambda_2\omega^2} 
(f')^\omega_{m\ell}(r)\bar F^\omega_{m\ell}(r) \psi^\omega_{m\ell}(r)  d\omega\\
&\le&  \int_{-\infty}^\infty \sum_{m,\ell:\lambda_{m\ell}(\omega)\ge  \lambda_2\omega^2} 
\delta^{-1} |(f')^\omega_{m\ell}F^\omega_{m\ell}|^2(r)+
\delta | \psi^\omega_{m\ell}|^2 d\omega\\
&\le&  \int_{-\infty}^\infty \sum_{m,\ell:\lambda_{m\ell}(\omega)\ge  \lambda_2\omega^2} 
\delta^{-1}B |F^\omega_{m\ell}|^2(r)+
\delta | \psi^\omega_{m\ell}|^2 d\omega\\
&=&\delta^{-1}B \int (F_{\mbox{$\natural$}})^2 \sin \theta \, d\phi \, d\theta \,dt^* +
\delta  \int (\psi_{\mbox{$\natural$}})^2 \sin \theta \, d\phi \, d\theta \,dt^* \\
&\le& \delta^{-1}B \int F^2 \sin \theta \, d\phi \, d\theta \,dt^* +
\delta  \int \psi^2 \sin \theta \, d\phi \, d\theta \,dt^*,
\end{eqnarray*}
where $\delta$ can be chosen arbitrarily.
In particular, this estimate holds for $r\le R$. For $r\ge R$, in view of the fact that
$f$ is independent of $\omega$, $m$, $\ell$, we have in fact
\begin{eqnarray*}
 \int_{-\infty}^\infty&& \sum_{m,\ell:\lambda_{m\ell}(\omega)\ge  \lambda_2\omega^2} 
(f')(r)\bar F^\omega_{m\ell}(r) \psi^\omega_{m\ell}(r)  d\omega\\
&=& f'(r) \int_{-\infty}^\infty \sum_{m,\ell:\lambda_{m\ell}(\omega)\ge  \lambda_2\omega^2} 
\bar F^\omega_{m\ell}(r) \psi^\omega_{m\ell}(r)  d\omega\\
&=& f'(r) \int F_{\mbox{$\natural$}}\psi_{\mbox{$\natural$}} \sin \theta \, d\phi \, d\theta \,dt^*\\
&=& f'(r) \int F_{\mbox{$\natural$}}\psi_{\hbox{\Rightscissors}} \sin \theta \, d\phi \, d\theta \,dt^*,
\end{eqnarray*}
where for the last line we have used $(\ref{sharpcutoff})$. The first summand of the
error integrand of $(\ref{beforewesum})$ can be estimated similarly.

We thus obtain 
\begin{align}
\label{a9roism}
\nonumber
b\int_{\mathcal{R}} \chi&
\left((\pa_{r}  \psi_{\mbox {$\natural$}})^2+ \psi_{\mbox {$\natural$}}^2\right) +
b\int_{\mathcal{R}} \chi h J^N_\mu (\psi_{\mbox {$\natural$}}) N^\mu \\  
\nonumber
\le& B\int_{\Sigma_{\tau'} }J^N_\mu(\psi)n^\mu_{\Sigma}+
\epsilon \int_{\mathcal{H}(\tau',\tau)} (\partial_\phi\psi)^2
+\delta^{-1} B
\int_{\mathcal{R}\cap\{r\le R\}   } 
F^2  \\
\nonumber&  +
\delta \int_{\mathcal{R}\cap\{r\le R\}   } 
\psi^2 +
(\partial_r\psi)^2  \\
\nonumber
&+ \int_{-\infty}^\infty dt^*\int_{r\ge R}
\left(2f (r^2+a^2)^{1/2} F_{\mbox{$\natural$}} \partial_{r^*} ((r^2+a^2)^{1/2} \psi_{\hbox{\Rightscissors}})
\right.\\
&\hskip2pc  \left.
+ f'(r^2+a^2)F_{\mbox{$\natural$}} 
\psi_{\hbox{\Rightscissors}}  \right)\frac{\Delta}{r^2+a^2}
 \sin \theta\,  d\phi^* \, d\theta \,  dr^* ,
\end{align}
where $\chi$ is a cutoff which degenerates at infinity and $h$ is a function $0\le h\le 1$
which vanishes in a suitable neighborhood of $r=3M$.

\subsubsection{The untrapped frequencies}
Given $\lambda_2$ sufficiently small and \emph{any} choice
of $\omega_1$, $\lambda_1$, then, for $a$ sufficiently small (where 
sufficiently small depends on these latter two constants), it follows that 
for 
$\hbox{\eighthnote}={\mbox {$\flat$}}, \, {\lessflat}, \, {\hbox{$\sharp$}}$, 
we may produce currents 
of type $J_\mu^{{\bf X}_{\hbox{\eighthnote}}}$ as in Section~\ref{sheur}   
such that
\[
b\int_{\mathcal{R}}  \chi 
J^{N}_\mu (\psi_{\hbox{\eighthnote}})
N^\mu + \tilde\chi\psi_{{\hbox{\eighthnote}}}
\le 
\int_{\mathcal{R}} K^{{\bf X}_{\hbox{\eighthnote}}}
(\psi_{\hbox{\eighthnote}}) 
\]
for $\chi$ a suitable cutoff function degenerating at infinity,
and $\tilde\chi$ a suitable cutoff function degenerating at infinity and vanishing
in a neighborhood of $\mathcal{H}^+$.
These currents can in fact be chosen independently of $a$ for such small $a$, 
and moreover, they can be chosen so that,
defining
\[
\mathcal{E}^{{\bf X}_{\hbox{\eighthnote}}}\doteq \nabla^\mu J^{{\bf X}_{\hbox{\eighthnote}}}_\mu
-K^{{\bf X}_{\hbox{\eighthnote}}},
\]
we have on the one hand
\begin{align*}
\int_{\mathcal{R}\cap \{r\ge R\}}
\mathcal{E}^{{\bf X}_{\hbox{\eighthnote}}} &=
 \int_{-\infty}^\infty dt^*\int_{r\ge R}
\left(2f (r^2+a^2)^{1/2} F_{\hbox{\eighthnote}}
\partial_{r^*} ((r^2+a^2)^{1/2} \psi_{\hbox{\Rightscissors}})
\right.\\
&\hskip2pc  \left.
+ f'(r^2+a^2)F_{\hbox{\eighthnote}}
\psi_{\hbox{\Rightscissors}}  \right)\frac{\Delta}{r^2+a^2}
 \sin \theta\,  d\phi^* \, d\theta \,  dr^* 
\end{align*}
for the $f$ of Section~\ref{kloubi}, and on the  
other hand, for the region  $r\le R$, we have
\[
\int_{\mathcal{R}\cap \{r\le R\}}
\mathcal{E}^{{\bf X}_{\hbox{\eighthnote}}} \le 
B\delta^{-1}  \int_{\mathcal{R}\cap \{r\le R\}}
F^2+
B \delta \int_{\mathcal{R}\cap \{r\le R\}}
\psi_{\hbox{\Rightscissors}}^2+(\partial_{r}\psi_{\hbox{\Rightscissors}})^2 +
\chi J^N_\mu(\psi_{\hbox{\Rightscissors}}) n^\mu
\]
where $\chi$ is supported near the horizon and away from a neighborhood of
$r=3M$.

Moreover, one can show as in Section~\ref{bootstrapsec}
that
\begin{eqnarray*}
-\int_{\mathcal{H}}J^{{\bf X}_{\hbox{\eighthnote}}}_\mu
(\psi_{\hbox{\eighthnote}}) n^\mu & \le &
-\int_{\mathcal{H}}J^{T}_\mu
(\psi_{\hbox{\eighthnote}}) n^\mu\\
&\le& -\int_{\mathcal{H}}J^{T}_\mu
(\psi_{\hbox{\Rightscissors}}) n^\mu\\
&\le &B \int_{\Sigma_{\tau'}}J^{N}_\mu(\psi) n^\mu.
\end{eqnarray*}
({\bf Exercise}: Prove the last inequality.)

From the identity
\[
\int_{\mathcal{H}^+} J^{{\bf X}_{\hbox{\eighthnote}}}_\mu
(\psi_{\hbox{\eighthnote}}) n^\mu_{\mathcal{H}}
+
\int_{\mathcal{R}}K^{{\bf X}_{\hbox{\eighthnote}}} (\psi_{\hbox{\eighthnote}})
= \int_{\mathcal{R}}\mathcal{E}^{{\bf X}_{\hbox{\eighthnote}}} (\psi_{\hbox{\eighthnote}})
\]
and the above remarks, one obtains finally an estimate
\begin{align}
\label{finaluntrapped}
\nonumber
\int_{\mathcal{R}} \chi &
(J^N_\mu(\psi_{\mbox {$\flat$}})+  J^N_\mu(\psi_{\lessflat})+J^N_\mu(\psi_{\mbox {$\sharp$}})
 )n^\mu_{\Sigma_\tau}\\
 \nonumber
\le& 
B \int_{\Sigma_{\tau'}}J^{N}_\mu(\psi) n^\mu  +  B \delta^{-1} \int_{\mathcal{R}\cap \{r\le R\}}F^2
    \\
\nonumber
&
 +   B \delta \int_{\mathcal{R}\cap \{r\le R\}}
\psi^2+(\partial_{r}\psi)^2 +\chi J^N_\mu(\psi) N^\mu\\
\nonumber
&
+ \int_{-\infty}^\infty dt^*\int_{r\ge R}
\left(2f (r^2+a^2)^{1/2} (F_{\mbox {$\flat$}}+F_{\lessflat}+F_{\mbox {$\sharp$}})
\partial_{r^*} ((r^2+a^2)^{1/2} \psi_{\hbox{\Rightscissors}})
\right.\\
&\hskip2pc  \left.
+ f'(r^2+a^2)(F_{\mbox {$\flat$}}+F_{\lessflat}+F_{\mbox {$\sharp$}})
\psi_{\hbox{\Rightscissors}}  \right)\frac{\Delta}{r^2+a^2}
 \sin \theta\,  d\phi^* \, d\theta \,  dr^* .
\end{align}

\subsubsection{The integrated decay estimates}
\label{integdecaysec}
Now, we will add $(\ref{a9roism})$, $(\ref{finaluntrapped})$ 
and 
the energy identity of $eJ^Y(\psi)$
\begin{align}
\label{addtoit}
\nonumber
\int_{\tilde\Sigma_\tau}& J^N_\mu(\psi) n^\mu_{\tilde{\Sigma}_\tau}+
\int_{\tilde{\mathcal{R}}(\tau',\tau)\cap \{r\le r_0\}}eK^Y(\psi)\\
&
=
-\int_{\mathcal{H}(\tau',\tau)}eJ^Y_\mu(\psi ) n^\mu_{\mathcal{H}}
+\int_{\tilde{\mathcal{R}}(\tau',\tau)\cap \{r_0\le r_1\le r_0\}}eK^Y(\psi)
+\int_{\tilde\Sigma_{\tau'}} J^N_\mu(\psi) n^\mu_{\tilde{\Sigma}_{\tau'}}
\end{align}
for a small $e$ with $\epsilon\ll e$, and where $r_0<r_1<3M$
are as in Corollary~\ref{Nproperties}, and $r_1$ is in the support
of $K_2$.

In the resulting inequality, 
the left hand side  bounds in particular
\begin{equation}
\label{THENEWLEFT}
\int_{\mathcal{R}(\tau'+1,\tau-1)} \chi (h
J^N_\mu(\psi)N^\mu + (\partial_{r}\psi)^2)
\end{equation}
 where $\chi$ is a cutoff decaying at infinity, 
 $\tilde\chi$ is a cutoff decaying at infinity and vanishing at $\mathcal{H}^+$
and $h$ is a function with $0\le h\le 1$ such that $h$ vanishes precisely in a neighborhood
of $r=3M$. 
(As $a\to 0$, this neighborhood can be chosen smaller and smaller in the sense
of the coordinate $r$.)

Let us examine the right hand side of the resulting inequality.

The second term of the first line of the right hand side of $(\ref{a9roism})$ is 
absorbed by the first term on the right hand side of $(\ref{addtoit})$ provided that $\epsilon \ll e$.

The third term of the first line of the right hand side of $(\ref{a9roism})$ and the second term
of $(\ref{finaluntrapped})$ are bounded by
\[
B\delta^{-1} \int_{\Sigma_{\tau'}} J^N_\mu (\psi) n^\mu_{\Sigma_{\tau'}}
\]
in view of Theorem~\ref{kerbnd}.

The second line of the right side of $(\ref{a9roism})$ and the third term
of $(\ref{finaluntrapped})$ can be 
absorbed by $(\ref{THENEWLEFT})$, provided that $\delta$ is chosen suitably small,
whereas the second term of the right hand side of $(\ref{addtoit})$ can be absorbed 
by $(\ref{THENEWLEFT})$, provided that $e$ is sufficiently small.

The fourth terms of the right hand sides of $(\ref{a9roism})$ and $(\ref{finaluntrapped})$ combine to yield
\begin{align*}
& \int_{-\infty}^\infty dt^*\int_{r\ge R}
\left(2f (r^2+a^2)^{1/2} F
\partial_{r^*} ((r^2+a^2)^{1/2} \psi_{\hbox{\Rightscissors}})
\right.\\
&\hskip2pc  \left.
+ f'(r^2+a^2)F
\psi_{\hbox{\Rightscissors}}  \right)\frac{\Delta}{r^2+a^2}
 \sin \theta\,  d\phi^* \, d\theta \,  dr^* .
\end{align*}
Note where $F$ is supported and how it decays.
Using our boundedness Theorem~\ref{kerbnd}, 
a Hardy inequality and integration by parts we may
now bound this term by
\[
B \int_{\Sigma_{\tau'}} J^N_\mu(\psi) n^\mu_{\Sigma_\tau}.
\]
But the remaining terms on the right hand side of $(\ref{a9roism})$, $(\ref{finaluntrapped})$
and $(\ref{addtoit})$ are also of this form!
We thus obtain
\begin{proposition}
\label{integdecay}
There exists a $\varphi_t$-invariant weight $\chi$, degenerating only at $i_0$,
a second $\varphi_t$-invariant weight $\tilde\chi$, degenerating at $i_0$ and vanishing
at $\mathcal{H}^+$,
a third $\varphi_t$-invariant weight $h$, which vanishes on
a neighborhood of $r=3M$, and a constant $B>0$
such that the following estimates hold for all $\tau'\le \tau$,
\[
\int_{\mathcal{R}(\tau',\tau)}\chi h J^N_\mu(\psi) N^\mu +\tilde\chi\psi^2
\le B\int_{\Sigma_{\tau'}} J_\mu^N(\psi)n^\mu_{\Sigma_{\tau'}}
\]
\[
\int_{\mathcal{R}(\tau',\tau)}\chi J^N_\mu(\psi) N^\mu+\tilde\chi \psi^2
\le B\int_{\Sigma_{\tau'}} (J_\mu^N(\psi)+J_\mu^N(T\psi))n^\mu_{\Sigma_{\tau'}}
\]
for all solutions $\Box_g\psi=0$ on Kerr.
\end{proposition}
Similar estimates could be shown on regions $\tilde{\mathcal{R}}(\tau',\tau)$, $\tilde{\Sigma}_\tau'$,
after having derived a priori suitable decay of $\psi$ in $r$.\footnote{In the section that follows, 
we shall in fact localise the above estimate in a different way applying a cutoff function.
The resulting $0$'th order terms which arise can be controlled
using the ``good'' $0$'th order term in the boundary
integrals of $J^{Z,w}$.}

\subsubsection{The $Z$-estimate}
To turn integrated decay
as in Proposition~\ref{integdecay} into decay of energy and pointwise
decay, we must adapt the argument of Section~\ref{Morawetz}.

Let $V$ be a $\phi_t$-invariant
vector field such that $V=\partial_{t^*}$ for $r\ge r_+ +c_2$
and $V=\partial_{t^*}+(a/2Mr_+) \partial_{\phi^*}$ for $f\le r_++c_1$ for some $c_1<c_2$,
and such that $V$ is timelike in $\mathcal{R}\setminus\mathcal{H}^+$. 
Note that
$V$ is Killing except
in $r_++c_1\le r\le r_++c_2$.  As $a\to 0$, we can construct such a $V$ with $c_2$ arbitrarily
small.

Now 
let us define $u$ and $v$ to be the Schwarzschild\footnote{Recall that
we are considering both the Kerr and Schwarzschild metric on the fixed
differentiable structure $\mathcal{R}$ as described in  Section~\ref{thekerrmetric}.}
 coordinates 
\[
u= t-r_{\rm Schw}^*,
\]
\[
v= t+r_{\rm Schw}^*.
\]
With respect to the coordinates $(u,v,\phi^*,\theta)$, 
defining $\underline{L}= \partial_u$, then $\underline{L}$
vanishes smoothly along the horizon. Define $\bar{L} = V-\underline{L}$.
Finally, define the vector field
\[
Z = u^2 L+v^2\underline{L}.
\]

Note that under these choices $Z$ is 
null on $\mathcal{H}^+$.
With $w$ as before, the currents
$J^{Z,w}$ together with $J^N$ can be used to control the energy fluxes
on $\Sigma_\tau$ with weights.
Use of the energy identities of $J^{Z,w}$ and $J^N$ leads to 
estimates of the form
\begin{equation}
\label{recal}
 \int_{\Sigma_\tau} \chi\psi^2+ \int_{\Sigma_\tau\cap \{r\precsim \tau\} }J^N_\mu(\psi) n_{\tilde\Sigma_\tau}^\mu \le B\,D \tau^{-2}+
B\, \tau^{-2}\int_{{\mathcal{R}}(0,\tau)} \mathcal{E},
\end{equation}
where $\chi$ is a cutoff function supported suitably, and
where $\mathcal{E}$ is an error term arising from the part of 
 $K^{Z,w}$ which has the ``wrong'' sign; $D$ arises from data.

We may  partition 
\[
\mathcal{E}=\mathcal{E}_1+\mathcal{E}_2+\mathcal{E}_3
\]
where  
\begin{itemize}
\item
$\mathcal{E}_1$ is supported in some region $r_0\le r\le R_0$,
\item
$\mathcal{E}_2$ is supported in $r\le r_0$ and
\item
$\mathcal{E}_3$ is supported in $r\ge R_0$.
\end{itemize}

Recall that $L+\underline{L}$ is Killing for $r\ge 2M+c_2$.
It follows ({\bf Exercise}) that   choosing $c_2<r_0$,
there are no terms growing quadratically in $t$ for $\mathcal{E}_1$, $\mathcal{E}_3$.
Moreover, by our construction, $Z$ depends smoothly on $a$ away from the horizon.
The behaviour near the horizon is more subtle as $Z$ itself is not smooth! We shall return
to this when discussing $\mathcal{E}_2$.

In view of our above remarks. we have that
\[
\mathcal{E}_1\le B\, t (J^N_\mu(\psi)N^\mu + \psi^2),
\]
just like in the case of Schwarzschild.
In view of Proposition~\ref{integdecay}, this leads to the following estimate:
If $\hat{\psi}=\psi$ in $\mathcal{R}(\tau',\tau'')\cap\{ r\le R_0 \}$, where $\hat\psi$ solves
again $\Box_g\hat\psi=0$, then
\begin{equation}
\label{letus1}
\int_{{\mathcal{R}}(\tau',\tau'')}\mathcal{E}_1(\psi)=
\int_{{\mathcal{R}}(\tau',\tau'')}\mathcal{E}_1(\hat\psi)
 \le B\tau'  \int_{{\Sigma}_{\tau'}}
 (J^N_\mu(\hat\psi)+J^N_\mu(T\hat\psi))n^\mu_{\tilde{\Sigma}_{\tau'}}
.
\end{equation}
The introduction of  $\hat\psi$  is related to our localisation procedure we shall carry out in what 
follows.

Recall that in the Schwarzschild case, for $R_0$ suitably chosen, there is no $\mathcal{E}_3$ term,
as the term $K^{Z,w}$ has a good sign in that region. (See Section~\ref{Morawetz}.)
Examining the $r$-decay of error terms in the smooth dependence of $Z$ in $a$, we obtain
\[
\mathcal{E}_3 \le \epsilon \, t r^{-2}  J^N_\mu(\psi)N^\mu
 \]
where $\epsilon$ can be made arbitrarily small if $a$ is small.
If $\tau''-\tau'\sim \tau'\sim t$, this leads to an estimate
\begin{eqnarray}
\nonumber
\label{letus3}
\int_{{\mathcal{R}}(\tau',\tau'')}\mathcal{E}_3(\psi)
&\le& \epsilon (\tau''-\tau')(\tau''+\tau') 
\int_{{\Sigma}_{\tau'}\cap \{r \precsim \tau''-\tau'  \} } J^N_\mu(\psi) n^\mu_{{\Sigma}_{\tau'}}\\
&&\hbox{}+\epsilon\,\log |\tau''-\tau'| \int_{\Sigma_{\tau'}} J_\mu^N(\psi) n^\mu_{\Sigma_{\tau'}} .
\end{eqnarray}

In the region $r_++c_1 \le r \le r_++c_2$, 
then, choosing $r_0$ such that $\mathcal{E}_2$ is absent in Schwarzschild,
we can argue without computation from the smooth
dependence on $a$ that
\[
\mathcal{E}_2\le \epsilon\, t^2 (J^N_\mu(\psi)N^\mu+\psi^2)
\]
where $\epsilon$ can be made arbitrarily small by choosing $a$ small.
The necessity of a quadratically growing error term arises from the fact that
$L+\underline{L}$ is not Killing in this region.\footnote{Alternatively, one can keep
$L+\underline{L}$ Killing at the expense of $Z$ failing to be causal on the horizon. This would
lead to errors of a similar nature.}

As we have already mentioned, 
an important subtely occurs near the horizon $\mathcal{H}^+$ where  
$Z$ fails to be $C^1$.  This means that $\mathcal{E}_2$ is not necessarily
small in local coordinates, and one must understand how to bound the singular terms. 
It turns out that these singular terms have a structure:
\begin{proposition}
Let $\hat{V}$, $\hat{Y}$, $E_1$, $E_2$ extend $V$ to a null frame in $r\le r_++c_1$.
We have
\[
\mathcal{E}_2 \le \epsilon v|\log (r-r_+)|^p ({\bf T}(\hat{Y},\hat{V})+{\bf T}(\hat{V},\hat{V})) 
+\epsilon v\, J^N_\mu(\psi)N^\mu.
\]
\end{proposition}

\begin{proof}
The warping function $w$ can be chosen as in Schwarzschild near $\mathcal{H}^+$, and thus,
the extra terms it generates are harmless.
For the worst behaviour, 
it suffices to examine now $K^Z$ itself. 
We must show that terms of the form:
\[
|\log (r-r_+)|^p(T(\hat{Y},\hat{Y}))
\]
do not appear in the computation for $K^{Z}$.

The relevant property follows from examining the covariant derivative
of $Z$ with respect to the null frame:
\[
\nabla_{\hat V} Z = 2u (\hat{V}u) \underline{L}+2v \hat{V}(v) L+ 
v^2 \nabla_{\hat V}V - 4r^* v \nabla_{\hat{V}} \underline{L}
+4(r^*)^2\nabla_{\hat{V}}\underline{L},
\]
\[
\nabla_{\hat{Y}} Z= 2u(\hat{Y}u)\underline{L}+
2v(\hat{Y}v) L + v^2 \nabla_{\hat{Y}} V- 4r^*v\nabla_{\hat{Y}}\underline{L}
+4 (r^*)^2 \nabla_{\hat{Y}}\underline{L}, 
\]
\[
\nabla_{E_1}Z=2u(E_1u)\underline{L}+
2v(E_1v) L + v^2 \nabla_{E_1} V- 4r^*v\nabla_{E_1}\underline{L}
+4 (r^*)^2 \nabla_{E_1}\underline{L},
\]
\[
\nabla_{E_2}Z=2u(E_2u)\underline{L}+
2v(E_2v) L + v^2 \nabla_{E_2} V- 4r^*v\nabla_{E_2}\underline{L}
+4 (r^*)^2 \nabla_{E_2}\underline{L}.
\]
\end{proof}

To estimate now $\mathcal{E}_2$, 
we first remark that with Proposition~\ref{integdecay}, we can obtain the following refinement
of the red-shift multiplier construction of Corollary~\ref{Nproperties}:
\begin{proposition}
\label{redrefine}
If we weaken the requirement that $N$ be smooth in Corollary~\ref{Nproperties}
with the statement that $N$ is $C^0$ at $\mathcal{H}^+$
and smooth away from $\mathcal{H}^+$, then given 
$p\ge 0$, we may construct an $N$ as in Corollary~\ref{Nproperties} where property~\ref{fir}
is replaced by the stronger
inequality:
\[
K^N(\psi) \ge b_p |\log (r-r_+)|^p ({\bf T}(\hat{Y},\hat{V})+{\bf T}(\hat{V},\hat{V}))
\]
for $r\le r_0$.
\end{proposition}

It now follows immediately  from Proposition~\ref{integdecay} 
that with $\psi$ and $\hat\psi$ as before, we have
\begin{equation}
\label{letus2}
\int_{{\mathcal{R}}(\tau',\tau'')}\mathcal{E}_2(\psi)
 \le \epsilon (\tau')^2 \int_{{\Sigma}_{\tau'}}
 J^N_\mu(\hat\psi) n^\mu_{{\Sigma}_{\tau'}}.
\end{equation}

To obtain energy decay from $(\ref{recal})$, $(\ref{letus3})$, $(\ref{letus1})$ and
$(\ref{letus2})$, we argue now by continuity. 
Introduce the bootstrap assumptions
\begin{equation}
\label{impr}
\int_{\Sigma_{\tau}\cap \{r\precsim \tau \}}J^N_\mu (\psi) N^\mu
+\chi\psi^2 \le C\, D \tau^{-2+2\delta},
\end{equation}
\begin{equation}
\label{impr2}
\int_{\Sigma_{\tau}\cap\{r\precsim\tau\}}J^N_\mu (T\psi) N^\mu \le C'\, D \tau^{-1+2\delta}
\end{equation}
for a $\delta>0$.

Now
dyadically decompose the interval $[0,\tau]$
by $\tau_i<\tau_{i+1}$.
Using $(\ref{letus1})$ and the above, we obtain
\begin{align}
\label{e1here}
\nonumber
\int_{{\mathcal{R}}(0,\tau)} 
 \mathcal{E}_1(\psi)
&\le \sum_i  \int_{{\mathcal{R}}(\tau_{i},\tau_{i+1})}\mathcal{E}_1(\psi) \\
\nonumber
&\le \sum_i \tau_i \int_{\Sigma_{\tau_i}}(J^N_\mu(\hat\psi)+ J^N_\mu (T\hat\psi)) N^\mu\\
\nonumber
&\le \sum_i \tau_i \int_{\Sigma_{\tau_i}\cap\{r\precsim \tau_{i+1}-\tau_i\} }(J^N_\mu(\psi)+ J^N_\mu (T\psi)) N^\mu+\chi\psi^2\\
\nonumber
&\le\sum_i \tau_i (\tau_i^{-2+2\delta}CD+ \tau_i^{-1+2\delta}C'D) \\
&\le\delta^{-1} (CD\tau^{-1+2\delta}+C'D\tau^{2\delta}).
\end{align}
Here, $\hat\psi$ is constructed separately on each dyadic region $\mathcal{R}(\tau_i,
\tau_{i+1})$ by throwing a cutoff on $\psi|_{\Sigma_{\tau_i}}$ 
eqaul to $1$ in $r\precsim \tau_{i+1}-\tau_i$ and
vanishing in $\tau_{i+1}-\tau_i
\precsim r$, solving again the initial value problem in $\mathcal{R}(\tau_i,\tau_{i+1})$,
and exploiting the domain of dependence property. 
See the original~\cite{dr3} for this localisation scheme. The
parameters of the ``dyadic'' decomposition must be chosen accordingly for the constants
to work out.
Similarly, using $(\ref{letus2})$ we obtain
\begin{align}
\int_{{\mathcal{R}}(0,\tau)} 
\label{e2here}
\nonumber
 \mathcal{E}_2(\psi)
&\le \sum_i  \int_{{\mathcal{R}}(\tau_{i},\tau_{i+1})}\mathcal{E}_2(\psi) \\
\nonumber
&\le \epsilon \sum_i \tau_i^2\int_{\Sigma_{\tau_i}}J^N_\mu (\hat\psi) N^\mu\\
\nonumber
&\le \epsilon \sum_i \tau_i^2\int_{\Sigma_{\tau_i}\cap\{r\precsim \tau_{i+1}-\tau_i\}}J^N_\mu (\psi) N^\mu+\chi\psi^2\\
\nonumber
&\le \epsilon\sum_i \tau_i^2\tau_i^{-2+2\delta}CD \\
&\le \epsilon\delta^{-1} \tau^{2\delta}C D
\end{align}
and using $(\ref{letus3})$
\begin{align}
\int_{{\mathcal{R}}(0,\tau)} 
\label{e3here}
\nonumber
 \mathcal{E}_3(\psi)
&\le \sum_i  \int_{{\mathcal{R}}(\tau_{i},\tau_{i+1})}\mathcal{E}_3(\psi) \\
\nonumber
&\le \epsilon \sum_i \left(
\tau_i^2\int_{\Sigma_{\tau_i}}J^N_\mu (\psi) N^\mu +\int_{\Sigma_{\tau_i}} 
J_\mu^N(\psi)n^\mu_{\Sigma_{\tau_i}}\right)\\
\nonumber
&\le \epsilon\sum_i (\tau_i^2\tau_i^{-2+2\delta}CD +D\log \tau')\\
&\le \epsilon\delta^{-1} \tau^{2\delta}C D.
\end{align}
For $T\psi$ we obtain
\begin{align}
\label{e1there}
\int_{{\mathcal{R}}(0,\tau)} 
 \mathcal{E}_1(T\psi)
&\le B D\tau,
\end{align}
\begin{align}
\label{e2there}
\nonumber
\int_{{\mathcal{R}}(0,\tau)} 
 \mathcal{E}_2(T\psi)
&\le \sum_i  \int_{{\mathcal{R}}(\tau_{i},\tau_{i+1})}\mathcal{E}_2(T\psi) \\
\nonumber
&\le \epsilon \sum_i \tau_i^2\int_{\Sigma_{\tau_i}}J^N_\mu (T\hat\psi) N^\mu\\
\nonumber
&\le \epsilon \sum_i \tau_i^2\int_{\Sigma_{\tau_i}\cap\{r \precsim \tau_{i+1}-\tau_i\}}J^N_\mu (T\psi) N^\mu+\chi (T\psi)^2\\
\nonumber
&\le \epsilon\sum_i \tau_i^2(\tau_i^{-1+2\delta}C'D+\tau_i^{-2+2\delta}CD) \\
&\le \epsilon\delta^{-1} \tau^{1+2\delta}C'D+\epsilon\delta^{-1}\tau^{2\delta}CD,
\end{align}
\begin{align}
\label{e3there}
\nonumber
\int_{{\mathcal{R}}(0,\tau)} 
 \mathcal{E}_3(T\psi)
&\le \sum_i  \int_{{\mathcal{R}}(\tau_{i},\tau_{i+1})}\mathcal{E}_3(T\psi) \\
\nonumber
&\le \epsilon \sum_i\left( \tau_i^2\int_{\Sigma_{\tau_i}}J^N_\mu (T\psi) N^\mu
+\int_{\Sigma_{\tau_i}}J_\mu^N(T\psi)n^\mu_{\Sigma_{\tau_i}}\right)\\
\nonumber
&\le \epsilon\sum_i (\tau_i^2\tau_i^{-1+2\delta}C'D +D\log \tau_i)\\
&\le \epsilon\delta^{-1} \tau^{1+2\delta}C'D.
\end{align}
We use here the algebra of constants where $B\epsilon= \epsilon$.
The constant
$D$ is a quantity coming from data. {\bf Exercise}: What is $D$ and why is $(\ref{e1there})$
true?

For $\epsilon\ll \delta$ and $C'$ sufficiently large, 
we see that from $(\ref{recal})$ applied to $T\psi$ in place
of $\psi$, using $(\ref{e1there})$, $(\ref{e2there})$, we improve $(\ref{impr2})$. 

On the other hand choosing $C'\ll C$ and then $\tau$ sufficiently large, we have
\[
\tau^{-2}\delta^{-1}(CD\tau^{-1+2\delta} +C' D\tau^{2\delta})\le \frac12CD\tau^{-2+2\delta}
\]
and thus, again for $\epsilon\ll \delta$, using $(\ref{e1here})$, $(\ref{e2here})$
we can improve  $(\ref{impr})$ from $(\ref{recal})$.

Once one obtains $(\ref{impr})$, 
then decay can be extended to decay in $\tilde\Sigma_\tau$
by the argument of Section~\ref{Morawetz}, by applying conservation of the $J^T$ flux
backwards.\footnote{Note that in view of the fact that we argued by continuity to obtain
$(\ref{impr})$, we could not obtain this extended decay through $\tilde\Sigma_\tau$ earlier.
This is why we have localised as in~\cite{dr3}, not as in Section~\ref{Morawetz}.}

\subsubsection{Pointwise bounds}
In any region $r\le R$, we may now obtain pointwise
decay bounds simply
by  further commutation with $T$, $N$ as in Section~\ref{asacommutator}.
To obtain the correct pointwise decay statement towards null infinity,
one must also commute the 
equation with a basis $\Omega_i$ for the Lie algebra of the
\emph{Schwarzschild} 
metric, exploiting the $r$-weights of these vector fields. 
Defining $\tilde\Omega_i=\zeta(r)\Omega_i$, where $\zeta$ is a cutoff which vanishes
for $r\le R_0$, where $3M\ll R_0$, and, setting $\tilde\psi=\tilde\Omega\psi$,
we have
\[
\Box_g \tilde\psi = F_1\partial^2\psi +F_2\partial\psi
\]
where $F_1=O(r^{-2})$ and $F_2=O(r^{-3})$.
Having estimates already for $\psi$, $T\psi$, 
one can may apply the $X$ and $Z$ estimates as before for $\tilde\psi$, 
only, in view of the $F_2$ term,
now one must exploit also the $X$-estimate in
$D^+(\Sigma_{\tau_i}\cap
\{r\precsim \tau_{i+1}-\tau_i\})\cap J^-(\Sigma_{\tau_{i+1}})$.
We leave this as an {\bf exercise}.

\subsubsection{The decay theorem}
We have obtained thus
\begin{theorem} 
\label{DT}
Let $(\mathcal{M},g)$ be Kerr for $|a|\ll M$, $\mathcal{D}$ be the closure of its
domain of dependence,
let $\Sigma_0$ be the surface $\mathcal{D}\cap\{t^*=0\}$, 
let $\uppsi$, $\uppsi'$ be initial data on $\Sigma_0$ such that
$\uppsi\in H^s_{\rm loc}(\Sigma)$, $\uppsi'\in H^{s-1}_{\rm loc}(\Sigma)$ for $s\ge 1$,
and $\lim_{x\to i^0}\uppsi=0$,
and let $\psi$ be the 
corresponding unique solution
of $\Box_g\psi=0$. Let $\varphi_\tau$ denote the $1$-parameter
family of diffeomorphisms generated by $T$,
let $\tilde\Sigma_0$ be an arbitrary spacelike hypersurface in 
$J^+(\Sigma_0 \setminus \mathcal{U})$ where $\mathcal{U}$ is an open neighborhood of the
asymptotically flat end\footnote{This
is just the assumption that $\tilde\Sigma_0$ ``terminates''
on null  infinity},
and define $\tilde\Sigma_\tau=\varphi_\tau (\tilde\Sigma_0)$.
Let $s\ge 3$ and assume
\[
E_1\doteq 
\int_{\Sigma_0} r^2 (J_\mu^{n_0}(\psi) + J_\mu^{n_0}(T\psi)+ J_\mu^{n_0}(TT\psi)) n^\mu_0<\infty.
\]
Then there exists a $\delta>0$ depending on $a$ (with $\delta\to 0$ as $a\to 0$) and a
$B$ depending only on $\tilde\Sigma_0$ such that
\[
\int_{\tilde \Sigma_\tau} J^N (\psi) n^\mu_{\tilde{\Sigma_\tau}} \le
BE_1\, \tau^{-2+2\delta}.
\]
Now let $s\ge 5$ and
assume
\[
E_2\doteq\sum_{|\alpha|\le2 } \sum_{\Gamma=\{T, N, \Omega_i\}}
\int_{\Sigma_0} r^2 (J_\mu^{n_0}(\Gamma^\alpha \psi) + J_\mu^{n_0}(\Gamma^\alpha T\psi)+ J_\mu^{n_0}(\Gamma^\alpha TT\psi)) n^\mu_0
<\infty
\]
where $\Omega_i$ are the Schwarzschild angular momentum operators.
Then
\[
\sup_{\tilde\Sigma_\tau}\sqrt{r}|\psi|\le B \sqrt{E_2}\, \tau^{-1+\delta}, \qquad
\sup_{\tilde\Sigma_\tau}r|\psi|\le B \sqrt{E_2}\, \tau^{(-1+\delta)/2}.
\]
\end{theorem}
One can obtain decay for arbitrary derivatives, including 
transversal derivatives to $\mathcal{H}^+$,
using additional commutation by $N$. 
See~\cite{dr7}.

\subsection{Black hole uniqueness}
In the context of the vacuum equations $(\ref{Evac})$, 
the Kerr solution 
plays an important role not only because it is believed to be stable, 
but because it is believed to be the only 
stationary black hole solution.\footnote{A further extrapolation leads
to the ``belief'' that all vacuum solutions eventually   decompose
into $n$ Kerr solutions moving away from each other.}
This is the celebrated \emph{no-hair ``theorem''}.
In the case of the Einstein-Maxwell equations, there is an analogous
no-hair ``theorem'' stating uniqueness for Kerr-Newman.
A general reference is~\cite{unique}.

Neither of these results is close to being a theorem in the generality which
they are often stated. Reasonably definitive statements have
only been proven in the much easier
static case, and in the case where axisymmetry is assumed a priori
and the horizon is assumed connected, i.e.~that there is one black hole. 
Axisymmetry can be inferred from stationarity under various special assumptions,
including the especially restrictive assumption of analyticity. 
See~\cite{jlc}   for  the latest on the analytic case, 
and~\cite{ionescu} for new
interesting results in the direction of removing the analyticity assumption
in inferring axisymmetry from stationarity.

Nonetheless, the expectation that black hole uniqueness is true
reasonably raises the question:
why the interest in more general black holes, allowed in Theorem~\ref{kerbnd}?

For a classical ``astrophysical'' motivation, note that
black hole solutions can in principle
 exist in the presence of persistent atmospheres. Perhaps the simplest
such constructions would
involve solutions of the Einstein-Vlasov system, where matter is described
by a distribution function on phase space invariant under geodesic flow. 
These black hole spacetimes would in general
not be Kerr even in their vacuum regions.
Recent speculations in high energy physics yield other possible motivations:
There are now a variety of ``hairy black holes'' 
solving Einstein-matter systems for non-classical matter, like Yang-Mills fields~\cite{smol}, 
and a large variety of vacuum black holes in higher dimensions~\cite{harvey}, 
many of which are currently the topic of intense study.

There is, however, a second  type of reason, which is relevant even when we restrict our 
attention to the vacuum equations $(\ref{Evac})$ in dimension 4. The less information 
one must use about the spacetime to obtain quantitative control on fields, the
better chance one has at obtaining a stability theorem.  
The essentially non-quantitative\footnote{As should
be apparent by the role of analyticity or Carleman estimates.}
aspect of our current limited understanding of black hole uniqueness
should make it clear that these arguments probably will not have a place
in a stability proof. Indeed, it would be an interesting problem
to explore the possibility of
obtaining a more quantitative version of uniqueness 
theorems (in a neighborhood of Kerr) following
ideas in  this section.

\subsection{Comments and further reading}  
\label{neocoms2}
Theorem~\ref{kerbnd} was proven in~\cite{dr6}.
In particular, this provided the first global result of any kind for general solutions
of  the Cauchy problem on a (non-Schwarzschild) Kerr background.
 Theorem~\ref{DT} was first announced at the Clay Summer School
 where these notes were lectured. Results in the direction of Proposition~\ref{integdecay} 
 are independently being studied in work in progress
 by Tataru-Tohaneanu\footnote{communication from Mihai Tohaneanu, a summer school
 participant  who attended these lectures} and Andersson-Blue\footnote{lecture of P.~Blue, 
Mittag-Leffler, September 2008}.

The best previous results concerning Kerr had been obtained by Finster
and collaborators in an important series of papers culminating in~\cite{fksy}. See also~\cite{fksy2}. 
The methods of~\cite{fksy} are spectral theoretic, with many pretty applications
of contour integration and o.d.e.~techniques.
The results of~\cite{fksy} do not
apply to general solutions of the Cauchy problem, however,
only to individual azimuthal modes,
i.e.~solutions $\psi_m$ of fixed $m$.  In addition, \cite{fksy}
imposes the restrictive assumption that
$\mathcal{H}^+\cap\mathcal{H}^-$ not be in the support of the modes.
(Recall the discussion of      
Section~\ref{discrete}.)
Under these assumptions,
the main result stated in~\cite{fksy} is that 
\begin{equation}
\label{fixedk}
\lim_{t\to\infty} \psi_m (r, t)=0
\end{equation}
for
any $r>r_+$.
Note that the reason that $(\ref{fixedk})$ did not yield any statement concerning general
solutions, i.e.~the sum over $m$--not even a non-quantitative one--is 
that one did not have
a quantitative boundedness statement as in Theorem~\ref{kerbnd}.
Moreover, one should mention that even for fixed
$m$, the results of~\cite{fksy} are  in principle
compatible with the statement
\[
\sup_{\mathcal{H}^+} \psi_m =\infty,
\]
i.e.~that the azimuthal modes blow up along
the horizon. 
See the comments in Section~\ref{heuristic}.
It is important to note, however, 
that the statement of~\cite{fksy} need not restrict to $|a|\ll M$, but concerns the entire
subextremal range $|a|<M$. Thus, the statement $(\ref{fixedk})$ of~\cite{fksy} is currently
the only known global statement about azimuthal modes on Kerr spacetimes
 for large but subextremal $a$.

There has also been interesting work on the Dirac equation~\cite{fksy0, hn},
for which superradiance does not occur, and the Klein-Gordon
equation~\cite{haf}. For the latter, see also Section~\ref{KGproblem}.

\subsection{The nonlinear stability problem for Kerr}
\label{formulation}
We have motivated these notes with the nonlinear stability problem of Kerr.
Let us give finally  a rough formulation.

\begin{conjecture}
\label{stabconj}
Let $(\Sigma, \bar{g}, K)$ be a vacuum initial data set (see Appendix~\ref{initdatasec}) sufficiently
close (in  a weighted sense) to the initial data on 
Cauchy hypersurface in the Kerr solution $(\mathcal{M},g_{M,a})$ for some
parameters $0\le |a|<M$. Then
the maximal vacuum development $(\mathcal{M},g)$ possesses a complete null infinity $\mathcal{I}^+$
such that  the metric restricted to $J^-(\mathcal{I}^+)$ approaches a Kerr solution
$(\mathcal{M},g_{M_f,a_f})$ in a uniform way (with respect to a foliation of the type
$\tilde\Sigma_\tau$ of Section~\ref{S2}) with quantitative decay rates, where $M_f$, $a_f$ are near $M$, $a$ respectively.
\end{conjecture}

Let us make some remarks concerning the above statement.
Under the assumptions of the above conjecture, $(\mathcal{M},g)$ certainly contains
a trapped surface $S$ by Cauchy stability. By Penrose's incompleteness theorem 
(Theorem~\ref{incthe}),
this implies that $(\mathcal{M},g)$ is future causally geodesically incomplete.
By the methods of the proof of Theorem~\ref{incthe}, it is
 easy to see that $S\cap J^-(\mathcal{I}^+)=\emptyset$. 
Thus, as soon as $\mathcal{I}^+$ is shown to be complete, it would follow
that the spacetime has a black hole region in the sense of Section~\ref{general?}.\footnote{Let us
also remark the obvious fact that the above conjecture implies in particular that 
weak cosmic censorship holds in a neighborhood of Kerr data.}

In view of this, one can also formulate the problem where the initial data are assumed
close to Kerr initial data on an incomplete subset of a Cauchy hypersurface with one
asymptotically flat end and bounded by a trapped surface. This is in fact the physical 
problem\footnote{Cf.~the comments on the relation between maximally-extended
Schwarzschild and Oppenheimer-Snyder.},
but in view of Cauchy stability, it is equivalent to the formulation we have given above. 
Note also the open problem described in the last paragraph of Section~\ref{vaccol}.

In the spherically symmetric analogue of this problem where the Einstein
equations are coupled with matter, or the Bianchi-triaxial IX
vacuum problem discussed in Section~\ref{higherhigherdim}, 
the completeness of null infinity can be inferred easily without detailed
understanding of the geometry~\cite{trapped, kostakis1}. 
One can view this as an ``orbital stability'' statement.
In this spherically symmetric case,
the asymptotic stability can then be studied a posteriori, as in~\cite{dr1, kostakis2}. 
This latter problem is much more difficult.

In the case of Conjecture~\ref{stabconj}, in contrast to the symmetric cases mentioned above,
one does \emph{not} expect 
to be able to show any weaker stability statement than the asymptotic stability with 
decay rates as stated. Note that it is only the Kerr family as a whole--\emph{not} the
Schwarzschild subfamily--which is expected to be asymptotically
stable: Choosing $a=0$ certainly does not imply that $a_f=0$.  On the other hand, if $|a|\ll M$,
then by the formulation of the above conjecture, 
it would follow that  $|a_f|\ll M_f$.
It is with this in mind that we have considered the
$|a|\ll M$ case in this paper.

\section{The cosmological constant $\Lambda$ and Schwarzschild-de Sitter}
\label{cosmolosec}
Another interesting setting for the study of the stability problem are black holes
within \emph{cosmological spacetimes}. 
Cosmological spacetimes--as opposed to
asymptotically flat spacetimes (See Appendix~\ref{asymptflat}), 
which model spacetime in the vicinity of an
isolated self-gravitating system--are supposed to model the whole universe. 
The working hypothesis  of classical cosmology is that the universe is 
approximately homogeneous and isotropic (sometimes known as the
\emph{Copernican principle}~\cite{he:lssst}).
In the Newtonian theory, it was not possible to formulate a cosmological model
satisfying this hypothesis.\footnote{It is possible, however, if one
geometrically reinterprets the Newtonian theory and allows space
to be--say--the torus. See~\cite{rendallrev}. These reinterpretations, of course,
postdate the formulation of general relativity.}
One of the major successes of general relativity was that the theory
allowed for such solutions, thus making cosmology into a mathematical science.

In the early years of mathematical cosmology, it was assumed that the 
universe should be static\footnote{much like in  the early studies of asymptotically
flat spacetimes discussed in Section~\ref{stars}}.
To allow  for such static cosmological
solutions, Einstein modified his equations $(\ref{Eeq})$ by adding a $0$'th order term:
 \begin{equation}
 \label{withcosmo}
 R_{\mu\nu}-\frac12g_{\mu\nu}R+ \Lambda g_{\mu\nu} = 8\pi T_{\mu\nu}.
 \end{equation}
 Here $\Lambda$ is a constant known as \emph{the cosmological constant}.
 When coupled with a perfect fluid, this system admits a static, homogeneous, isotropic 
 solution with
 $\Lambda>0$ and topology $\mathbb S^3\times \mathbb R$. This spacetime is sometimes
 called the \emph{Einstein static universe}.

Cosmological solutions with various values of the parameter $\Lambda$ were studied
by Friedmann and Lemaitre, under the hypothesis of exact homogeneity and isotropy. 
Static solutions are in fact always unstable under perturbation
of initial data. Typical homogeneous isotropic solutions expand
or contract, or both, beginning and or ending in singular configurations.
As with the early studies (referred to in Sections~\ref{orextsec})
illuminating the extensions of the Schwarzschild metric across
the horizon, these were ahead of their time.\footnote{In fact, the two are very closely
related! The interior region of the Oppenheimer-Snyder collapsing star is precisely
isometric to a region of a Friedmann universe. See~\cite{MTW}.}
(See the forthcoming book~\cite{bieri} for a history of this fascinating early
period in the history of mathematical cosmology.)  These predictions
were taken more seriously with Hubble's observational discovery of the expansion of the
universe, and the subsequent evolutionary theories of matter,
but the relevance of the solutions near where they are actually singular 
was taken seriously only after the incompleteness theorems  of Penrose and Hawking--Penrose
were proven (see Section~\ref{trapsec}).

We shall not go into a general discussion of cosmology here, 
nor tell the fascinating story of the ups and downs of $\Lambda$--from 
its adoption by Einstein to his subsequent well-known rejection of it, to its later
``triumphant'' return in current cosmological models, taking a very small positive value,
the ``explanation'' of which is widely regarded as one of the outstanding puzzles of
theoretical physics.
Rather, let us pass directly to the object of our study here,
one of the simplest examples of an inhomogeneous  ``cosmological'' spacetime,
where non-trivial small scale structure occurs in an ambient expanding cosmology.
This is the
Schwarzschild--de Sitter solution.

\subsection{The Schwarzschild-de Sitter geometry}

Again, this metric
was discovered in local coordinates early in the history of general relativity,
independently by Kottler~\cite{kottler} and Weyl~\cite{Weyl}.
Fixing $\Lambda>0$,\footnote{The expression $(\ref{metricexp})$
with $\Lambda<0$ defines
\emph{Schwarzschild--anti-de Sitter}. See Section~\ref{CFT}.} Schwarzschild-de Sitter 
is a one-parameter family of solutions of the from
\begin{equation}
\label{metricexp}
-(1-2M/r- \Lambda r^3)dt^2 + (1-2M/r -\Lambda r^3)^{-1} dr^2+ r^2 d\sigma_{\mathbb S^2}.
\end{equation}
The black hole case is the case where $0<M<\frac{1}{3\sqrt{\Lambda}}$.
A maximally-extended solution~(see~\cite{carter, GibHawk})
then has as Penrose diagram the infinitely repeating chain:
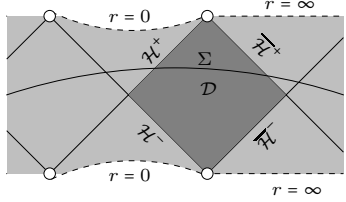
\[
\input{desitforkerr.pstex_t}
\]
To construct ``cosmological solutions'' one often takes spatially compact quotients.
(One can also glue such regions into other cosmological spacetimes. 
See~\cite{cp}. For more on the geometry of this solution, see~\cite{jb:gcc}.)

\subsection{Boundedness and decay}
The region ``analogous'' to the region studied previously for Schwarzschild and Kerr is the
darker shaded region $\mathcal{D}$ above. 
The horizon ${\overline{\mathcal{H}}}^+$ separates $\mathcal{D}$ from 
an ``expanding'' region where the spacetime is similar to the celebrated de-Sitter space.
If $\Sigma$ is a Cauchy surface such that
$\Sigma\cap\mathcal{H}^-=\Sigma\cap\overline{\mathcal{H}}^-=\emptyset$, then
let us define $\Sigma_0=\mathcal{D}\cap \Sigma$, and let us define
$\Sigma_\tau$ to be the translates of $\Sigma_0$ by the flow $\varphi_t$
generated by the Killing
field $T$ ($=\frac{\partial}{\partial t}$). Note that, in contrast to
the Schwarzschild or Kerr case, $\Sigma_0$ is compact.

We have
\begin{theorem}
\label{desitbound}
The statement of Theorem~\ref{boundedn} holds for these spacetimes,
where $\Sigma$, $\Sigma_0$, $\Sigma_\tau$
are as above, and $\lim_{x\to i^0}|\uppsi|$ 
is replaced by $\sup_{x\in\Sigma_0}|\uppsi|$.
\end{theorem}
\begin{proof}
The proof of the above theorem is as in the Schwarzschild case, 
except that in addition to the analogue
of $N$, one must use a vector field $\bar N$ which plays the role of $N$ near the
``cosmological horizon'' $\bar{\mathcal{H}}^+$. It is a good {\bf exercise} for the reader
to think about the properties required to construct such a $\bar N$. A general construction
of such a vector field applicable to all non-extremal stationary black holes
is done in Section~\ref{epilogue}.  
\end{proof}

As for decay, 
we have
\begin{theorem}
\label{dst}
For every $k\ge 0$, there exist constants $C_k$ such that the following holds. 
Let $\uppsi\in H^{k+1}_{\rm loc}$, $\uppsi'\in H^{k}_{\rm loc}$, and 
define
\begin{eqnarray*}
E_k &\doteq& \sum_{|(\alpha)|\le k } 
\sum_{\Gamma=\{\Omega_i\}}\int_{\Sigma_0} J^{n_{\Sigma_\tau}}_\mu (\Gamma^{\alpha} \psi)
n^\mu_{\Sigma_\tau}.\\
\end{eqnarray*}
Then    
\begin{equation}
\label{tote...ds}
 \int_{\Sigma_\tau} J^{n_{\Sigma_\tau}} _\mu(\psi) n^\mu_{\Sigma_\tau} \le 
C_k E_k \tau^{-k }.
\end{equation}
For $k>1$ we have
\begin{equation}
\label{afairw}
\sup_{\Sigma_{\tau}} |\psi-\psi_0| \le C_k\sqrt{ E_k} \tau^{\frac{-k+1}2},
\end{equation}
where $\psi_0$ denotes the $0$'th spherical harmonic, for which we have
for instance the estimate 
\begin{equation}
\label{for0ths}
\sup_{\Sigma_{\tau}} |\psi_0| \le \sup_{x\in\Sigma_0}\uppsi_0+C_0\sqrt{ E_0(\uppsi_0,\uppsi_0')}.
\end{equation}
\end{theorem}
The proof of this theorem uses the vector fields $T$, $Y$ 
and $\bar{Y}$ (alternatively $N$, $\bar{N}$),
together with a version of $X$ as multipliers, and requires commutation 
of the equation with $\Omega_i$ to quantify the loss caused by
trapping. (Like Schwarzschild,
the Schwarzschild-de Sitter metric has a photon sphere which is
at $r=3M$ for all values of $\Lambda$ in the allowed range. See~\cite{claire}
for a discussion of the optical geometry of this metric and its importance for gravitational lensing.)
An estimate analogous to~$(\ref{finalestimate})$ is obtained, but without the $\chi$ 
weight, in view of the compactness of $\Sigma_0$. 
The result of the Theorem follows essentially immediately, in view of 
Theorem~\ref{desitbound} and a pigeonhole argument.
No use need be made of a vector field of the type $Z$ as in Section~\ref{Morawetz}.
Note that for $\psi=\rm constant$, $E_k=0$, so removing the $0$'th spherical
harmonic in $(\ref{afairw})$ is necessary. See~\cite{dr4} for details.

Note that if $\Omega_i$ can be replaced by $\Omega_i^\epsilon$ in 
$(\ref{finalestimate})$, then it follows
that the loss in derivatives for energy decay at any polynomial rate $k$
in $(\ref{tote...ds})$ can be made arbitrarily small.
If $\Omega_i$ could be replaced by 
$\log\Omega_i$, then what would one obtain? ({\bf Exercise})

It would be a nice {\bf exercise}
to  commute with $\hat{Y}$ as in the proof of Theorem~\ref{desitbound}, to obtain pointwise
decay for arbitrary derivatives of $k$. See the related exercise in Section~\ref{pointdec}
concerning improving the statement of Theorem~\ref{Schdec}.

\subsection{Comments and further reading}
\label{neocoms3}
Theorem~\ref{dst}  was proven in~\cite{dr4}. 
Independently, the problem of the wave equation on
Schwarzschild-de Sitter has been considered in a nice paper of
Bony-H\"afner~\cite{bh} using methods of scattering theory. In that setting,
the presence of trapping is manifest by the appearance of resonances, that is to say,
the poles of the analytic continuation of the resolvent.\footnote{In the physics literature, these
are known as \emph{quasi-normal modes}. See~\cite{Kokkotas} for a nice survey, as well as
the discussion in Section~\ref{heuristic}.}  
The relevant estimates on the distribution
of these necessary for the analysis of~\cite{bh} had been obtained earlier by
S\'a Barreto and Zworski~\cite{SB-Zworski}.

In contrast to Theorem~\ref{dst},
the theorem of Bony-H\"afner~\cite{bh} makes the familiar 
restrictive assumption on the support of
initial data discussed in Section~\ref{invertin}. For these data, however, 
the results of~\cite{bh} obtain better decay than Theorem~\ref{dst} away from the horizon,
namely exponential, at the cost of only an $\epsilon$ derivative. 
The decay results of~\cite{bh} degenerate at the horizon, in particular, they
do not retrieve even boundedness for $\psi$ itself.
However, using the result of~\cite{bh} together with the analogue
of the red-shift $Y$ estimate as used in the proof of Theorem~\ref{dst}, one can prove
exponential      decay  up to and including the horizon, i.e.~exponential decay
in the parameter $\tau$ ({\bf Exercise}). This still requires, however, 
the restrictive hypothesis of~\cite{bh} concerning the support of the data.
It would be interesting to sort out whether the restrictive hypothesis can be removed
from~\cite{bh}, and whether this fast decay is stable to perturbation.
There also appears to be interesting work in progress by S\'a Barreto, Melrose and 
Vasy~\cite{vasy}
on a related problem.

One should expect that the statement of Theorem~\ref{desitbound} 
holds for the wave equation on
axisymmetric stationary perturbations of Schwarzschild-de Sitter,
in particular, slowly rotating Kerr-de Sitter, in analogy to Theorem~\ref{kerbnd}.

Finally, we note that in many context, more natural than the wave equation   is the
conformally covariant wave equation $\Box_g\psi-\frac16 R\psi =0$.  
For Schwarzschild-de Sitter, this is then a special case of
Klein-Gordon $(\ref{KGeq})$ with $\mu>0$. The analogue of Theorem~\ref{desitbound} holds by
virtue of Section~\ref{redshiftapps}.
 {\bf Exercise}: Prove the analogue of Theorem~\ref{dst} for this
equation.

\section{Epilogue: The red-shift effect for non-extremal black holes}
\label{epilogue}
We give in this section general
assumptions for the existence of vector fields $Y$ and
$N$ as in Section~\ref{thevectorfields}.
As an application, we can obtain the boundedness result of
Theorem~\ref{boundedn} or Theorem~\ref{desitbound}  for all classical non-extremal black holes
for general nonnegative cosmological constant $\Lambda\ge0$.
See~\cite{he:lssst, Townsend, carter} for discussions of these solutions.

\subsection{A general construction of vector fields $Y$ and $N$}
Recall that a \emph{Killing horizon} is a null hypersurface whose normal is 
Killing~\cite{unique,Townsend}.
Let $\mathcal{H}$ be a sufficiently regular
Killing horizon with (future-directed) generator the Killing field $V$,
which bounds a spacetime
$\mathcal{D}$. Let $\varphi_t^V$ denote the one-parameter family
of transformations generated by $V$, assumed to be globally defined
for all $t\ge0$.
Assume there exists a spatial hypersurface $\Sigma\subset \mathcal{D}$ 
transverse
to $V$, such that $\Sigma\cap\mathcal{H}=S$ is a compact $2$-surface.
Consider the region 
\[
\mathcal{R}'=\cup_{t\ge 0}\varphi_t^V(\Sigma)
\]
and assume that $\mathcal{R}'\cap \mathcal{D}$
is smoothly foliated by $\varphi_t(\Sigma)$.

Note that
\[
\nabla_VV=\kappa\, V
\]
for some function 
$\kappa:\mathcal{H}\to \mathbb R$.

\begin{theorem}
\label{eptheorem}
Let $\mathcal{H}$, $\mathcal{D}$,  $\mathcal{R}'$, $\Sigma$,
$V$, $\varphi^V_t$ be as above.
Suppose $\kappa>0$. 
Then there exists a $\phi_t^V$-invariant future-directed
timelike vector field $N$ on  $\mathcal{R}'$ 
and a constant $b>0$
such that 
\[
K^N\ge b \, J^N_\mu N^\mu
\]
in an open
$\varphi_t$-invariant (for $t\ge 0$) 
subset $\tilde{\mathcal{U}}\subset \mathcal{R}'$ containing $\mathcal{H}\cap\mathcal{R}'$.
\end{theorem}

\begin{proof}
Define $Y$ on $S$ so that $Y$ is future directed null, say 
\begin{equation}
\label{=s}
g(Y,V)=-2,
\end{equation}
and orthogonal to $S$.
Moreover, extend $Y$ off $S$ so that 
\begin{equation}
\label{covdev}
\nabla_YY = -\sigma(Y+V)
\end{equation}
on $S$.
Now push $Y$ forward by $\varphi^V_t$ to a vector field on $\mathcal{U}$. 
Note that all the above relations still hold on $\mathcal{H}$.

It is easy to see that the relations $(\ref{listten1})$--$(\ref{listten4})$ hold as before,
where $E_1$, $E_2$ are a local frame for $T_p\varphi^V_t(S)$. Now $a^1$, $a^2$ are
not necessarily $0$, hence our having included them in the original computation!
We define as before
\[
N=V+Y.
\]
Note that it is the compactness of $S$ which gives the uniformity of the choice of $b$
in the statement of the theorem.
\end{proof}

We also have the following commutation theorem
\begin{theorem}
\label{hocom}
Under the assumptions of the above theorem,
 if $\psi$ satisfies $\Box_g\psi=0$, then for all $k \ge 1$.
\[
\Box_g( Y^k\psi) = \kappa_k Y^k\psi + \sum_{0\le |m|\le k,\, 0\le m_4<k }c_m E_1^{m_1}
E_2^{m_2}T^{m_3}Y^{m_4} \psi
\]
on $\mathcal{H}^+$, where $\kappa_k>0$.
\end{theorem}
\begin{proof}
From $(\ref{listten1})$--$(\ref{listten4})$, we deduce that 
relative to the null frame (on the horizon) $V, Y, E_1, E_2$ the deformation tensor ${}^Y\pi$ 
takes the form 
\begin{align*}
&{}^Y\pi_{YY}=2\sigma,\quad {}^Y\pi_{VV}=2\kappa,\quad {}^Y\pi_{VY}=\sigma
&{}^Y\pi_{YE_i}=0,\quad {}^Y\pi_{VE_i}= a^i, \quad {}^Y\pi_{E_i E_j}=h_i^j
\end{align*}
As a result the principal part of the commutator expression--the 
term $2\, ^Y\pi^{\a\b}  D_\a D_\b \psi$
can be written as follows 
$$
 2\, ^Y\pi^{\a\b}  \nabla_\a \nabla_\b \psi=\kappa \nabla^2_{YY}\psi + 
 \sigma (\nabla^2_{VV}+\nabla^2_{YV})   \psi - a^i \nabla^2_{YE_i} \psi +
 2 h^i_j \nabla^2_{E_i E_j}  \psi.
$$
The result now follows by induction on $k$.
\end{proof}

\subsection{Applications}
\label{redshiftapps}
The proposition applies in particular to sub-extremal Kerr and Kerr-Newman,
as well as to both horizons of sub-extremal Kerr-de Sitter, Kerr-Newman-de Sitter, etc.
Let us give the following general, albeit somewhat
awkward statement:
\begin{theorem}
Let $(\mathcal{R},g)$ be a manifold with stratified boundary 
$\mathcal{H}^+\cup \Sigma$, such that
$\mathcal{R}$ is globally hyperbolic with past boundary the Cauchy hypersurface
$\Sigma$, 
where $\Sigma$ and $\mathcal{H}$ are themselves manifolds with (common) boundary $S$.
Assume 
\[
\mathcal{H}^+=\cup_{i=1}^n\mathcal{H}^+_i, \qquad  S=\cup_{i=1}^n S_i,
\]
where the unions are disjoint and each $\mathcal{H}^+_i$, $S_i$ is connected.
Assume each $\mathcal{H}^+_i$
satisfies
the assumptions of Theorem~\ref{eptheorem}
with future-directed Killing field $V_i$, some subset $\Sigma_i\subset \Sigma$,
and cross section 
a connected component $S_i$ of $S$. 
Let us assume there exists a Killing field $T$ with future complete orbits,
and $\varphi_t$ is the one-parameter family of transformations generated by $T$. 
Let
$\tilde{\mathcal{U}}_i$ be given by Theorem~\ref{eptheorem} and
assume that there exists a $\mathcal{V}$ as above such that
\[
\mathcal{R}=\varphi_t(\Sigma\setminus \mathcal{V}) \,\cup\, \cup_{i=1}^n\tilde{\mathcal{U}}_i.
\]
and
\[
-g((\varphi_t^{V_i})_* n_\Sigma, n_{\Sigma_\tau})\le B
\]
where $\Sigma_\tau=\varphi_\tau(\Sigma)$,
$\varphi_t^{V_i}$ represents the one-parameter family of transformations generated
by $V^i$, and the last inequality is assumed
for all values of $t$, $\tau$ where the left hand side can be defined.
Finally, 
let $\psi$ be a solution to the wave equation and
assume that 
for any open neighborhood
$\mathcal{V}$ of $S$ in $\Sigma$,  there exists
a positive constant $b_\mathcal{V}>0$ such that
\begin{equation}
\label{assump1edw}
J^T_\mu(T^k \psi) n_{\Sigma}^\mu \ge
b_{\mathcal{V}}
J^{n_\Sigma}_\mu(T^k \psi) n_{\Sigma}^\mu
\end{equation}
in $\Sigma\setminus \mathcal{V}$
and
\begin{equation}
\label{assump2edw}
T\psi = c_iV_i\psi 
\end{equation}
on $\mathcal{H}^+_i$.
It follows that
the first statement of Theorem~\ref{boundedn} holds for $\psi$.

Assume in addition that $\Sigma$ is  compact or asymptotically flat, in the weak 
sense of the validity of a Sobolev estimate~$(\ref{ellipt})$ near infinity.
Then 
the second statement of Theorem~\ref{boundedn} holds for $\psi$.
\end{theorem}

In the case where $T$ is assumed timelike
in $\mathcal{R}\setminus\mathcal{H}^+$,
then $(\ref{assump2edw})$ is automatic,
whereas $(\ref{assump1edw})$ holds if
\[
-g(T,T)\ge -b_{\mathcal{V}}\, g(n_\mu,T)
\]
in $\Sigma\setminus\mathcal{V}$.
Thus we have
\begin{corollary}
The above theorem applies to Reissner-Nordstr\"om, Reissner-Nordstr\"om-de Sitter,
etc, for all subextremal range of parameters.
Thus Theorem~\ref{boundedn} holds for all such
metrics.\footnote{In the $\Lambda=0$ case this range is $M>0$,
$0\le |Q|<  M$. {\bf Exercise}: What is it for $\Lambda>0$?}
\end{corollary}

On the other hand, $(\ref{assump2edw})$, $(\ref{assump1edw})$ can be easily
seen to hold for \emph{axisymmetric} solutions $\psi_0$ of $\Box_g\psi=0$
on backgrounds in the Kerr
family (see Section~\ref{BP}).
We thus have 
\begin{corollary}
The statement of Theorem~\ref{boundedn} holds for axisymmetric solutions $\psi_0$ of        for Kerr-Newman and Kerr-Newman-de Sitter for the full subextremal
range of parameters.\footnote{In the $\Lambda=0$ case this range is
$M> 0$,  $0\le |Q|<   \sqrt{M^2-a^2}$. {\bf Exercise}: What is it for $\Lambda>0$?}
\end{corollary}

Let us also mention that
the the theorems of this section apply
to the Klein-Gordon equation $\Box_g\psi= \mu^2\psi$, as well as
to the Maxwell equations ({\bf Exercise}).

\section{Open problems}
\label{acik}
We end these notes with a discussion of open problems. Some of these are related
to Conjecture~\ref{stabconj}, but all have independent interest. 

\subsection{The wave equation}
\label{waves}
The decay rates of Theorem~\ref{Schdec} are sharp as uniform decay rates in $v$ for any nontrivial
class of initial data. On the other hand, it would be nice to  obtain more decay in the interior,
possibly under a stronger assumption on initial data.

\begin{problem} 
\label{MORE}
Show that there exists a $\delta>0$ such that $(\ref{tote...})$ holds with 
$\tau$ replaced with
$\tau^{-2(1+\delta)}$, for a suitable redefinition of $E_1$. 
Show the same thing for Kerr spacetimes with $|a|\ll M$.
\end{problem}

At the very least, it would be nice to obtain this result for the energy restricted
to $\tilde{\Sigma}_\tau\cap \{r\le R\}$.

Recall how the algebraic structure of the Kerr solution is used in a fundamental way
in the proof of Theorem~\ref{DT}. 
On the other hand, one would think that the validity of the results should
depend only on the robustness of the trapping structure.
This suggests the following
\begin{problem}
Show the analogue of Theorem~\ref{DT} for the wave equation on metrics close to
Schwarzschild with as few as possible geometric assumptions on the metric.
\end{problem}
For instance, can Theorem~\ref{DT} be proven under the assumptions of Theorem~\ref{kerbnd}?
Under even weaker assumptions?

Our results for Kerr require $|a|\ll M$. Of course, this is a ``valid'' assumption in the context
of the nonlinear stability problem, in the sense that if this condition is assumed on the parameters
of the initial reference Kerr solution, one expects it holds for the final Kerr solution. 
Nonetheless, one certainly would like a result for all cases.
See the discussion in Section~\ref{neocoms2}.
\begin{problem}
Show the analogue of Theorem~\ref{DT} for Kerr solutions in the entire subextremal range
$0\le |a|<M$.
\end{problem}

The extremal case $|a|=M$ may be quite different in view of the fact that Section~\ref{epilogue}
cannot apply:
\begin{problem}
Understand the behaviour of solutions to the wave equation on 
extremal Reissner-Nordstr\"om, extremal Schwarzschild-de
Sitter, and extremal Kerr.
\end{problem}

Turning to the case of $\Lambda>0$, we have already remarked that the 
analogue of Theorems~\ref{desitbound} and~\ref{dst}  
should certainly hold in the case of Kerr-de Sitter.
In the case of both Schwarzschild-de Sitter and Kerr-de Sitter, 
another interesting problem is to understand the behaviour in the region
$\mathcal{C}=J^+(\overline{\mathcal{H}}^+_A)\cap J^+(\overline{\mathcal{H}}^+_B)$,
where $\mathcal{H}^+_A$, $\mathcal{H}^+_B$ are  two cosmological horizons
meeting at a sphere:
\begin{problem}
Understand the behaviour of solutions to the wave equation in region $\mathcal{C}$ of
Schwarzschild-de Sitter and Kerr-de Sitter, in particular, their behaviour along $r=\infty$ as $i^+$ is approached.
\end{problem}

Let us add that in the case of cosmological constant, in some contexts it is appropriate to 
replace $\Box_g$ with the conformally covariant wave operator
$\Box_g - \frac16R$. In view of the fact that $R$ is constant, this is a special case of
the Klein-Gordon equation discussed in Section~\ref{KGproblem} below.

\subsection{Higher spin}
The wave equation is a ``poor man's'' linearisation of the Einstein equations $(\ref{Evac})$. 
The role of linearisation in  the mathematical theory
of nonlinear partial differential equations is of a different nature than that which one might
imagine from the formal ``perturbation'' theory which one still encounters in the physics literature. 
Rather than linearising the equations, one considers 
the solution of the 
nonlinear equation from the point of view of a related linear equation that it itself satisfies.

In the case of the simplest nonlinear equations (say $(\ref{powernl})$ discussed in 
Section~\ref{notlinear} below), 
typically this means freezing the right hand side, 
i.e.~treating it as a given inhomogeneous term.
In the case of the Einstein equations, the proper analogue of this procedure
is much more geometric. Specifically, it 
amounts to looking at the so called Bianchi
equations
\begin{equation}
\label{spin2}
\nabla_{[\mu}R_{\nu\lambda]\rho\sigma} = 0,
\end{equation}
which are already linear as equations for the curvature tensor when $g$ is regarded
as fixed. For more on this point of view, see~\cite{book}. The above equations
for a field $S_{\lambda\mu\nu\rho}$
with the symmetries
of the Riemann curvature tensor 
are in general known as the spin-$2$ equations.
This motivates:
\begin{problem}
State and prove the spin-$2$ version of Theorems~\ref{kerbnd} or~\ref{DT} (or 
Open problem~\ref{MORE}) on
Kerr metric backgrounds or more generally, metrics settling down to Kerr.
\end{problem}
In addition to~\cite{book},
a good reference for these problems is~\cite{ck1}, where this problem is resolved just for
Minkowski space. 
In contrast to the case of Minkowski space,
an additional difficulty in the above problem for the black hole setting
arises from the presence of nontrivial stationary solutions provided by
the curvature tensor of the solutions themselves. 
This will have
to be accounted for in the statement of any decay theorem.
From the ``linearisation'' point of view, the 
existence of stationary solutions 
is of course related to the fact that it is the $2$-parameter Kerr family which is 
expected to be stable,  not an individual solution.

\subsection{The Klein-Gordon equation}
\label{KGproblem}
Another important problem is the Klein-Gordon equation
\begin{equation}
\label{KGeq}
\Box_g\psi = \mu \psi.
\end{equation}
A large body of heuristic studies suggest the existence of a sequence of quasinormal modes
(see Section~\ref{heuristic})
approaching the real axis from below in the Schwarzschild case. When the metric is perturbed
to Kerr, it is thought that essentially this sequence ``moves up'' and produces exponentially
growing solutions. See~\cite{zouros, detweiler}.
This suggests
\begin{problem}
Construct an exponentially growing solution of  $(\ref{KGeq})$ on Kerr, for arbitrarily
small $\mu>0$
and arbitrary small $a$.
\end{problem}
Interestingly, if one fixed $m$, then adapting the proof of Section~\ref{BP}, one can show that
for $\mu>0$ sufficiently small and $a$ sufficiently small, depending on $m$, the statement
of Theorem~\ref{kerbnd} holds for $(\ref{KGeq})$ for such Kerr's. This is consistent with the
quasinormal mode picture, as one must take $m\to\infty$ for the modes to approach
the real axis in Schwarzschild. This shows how misleading fixed-$m$ results can be
when compared to the actual physical problem.

\subsection{Asymptotically anti-de Sitter spacetimes}
\label{CFT}
In discussing the cosmological constant we have considered only the case $\Lambda>0$. 
This is the case of current interest in cosmology. On the other hand, from the completely different
viewpoint of high energy physics, there has been intense interest in the case
$\Lambda<0$. See~\cite{gibbons}.

The expression $(\ref{metricexp})$ 
for $\Lambda<0$ defines a solution known as \emph{Schwarzschild-anti-de Sitter}.
A Penrose diagramme of this solution is given below. 
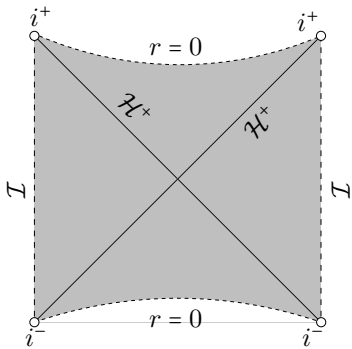
\[
\input{SchwarzschildADS.pstex_t}
\]
The timelike character of infinity means
that this solution is not globally hyperbolic. As with Schwarzschild-de Sitter,
Schwarzschild-anti-de Sitter can be viewed as a subfamily of
a larger Kerr-anti de Sitter family, with similar properties.

Again, as with Schwarzschild-de Sitter, the role of the wave equation is in some contexts replaced by
the conformally covariant wave equation. Note that this
corresponds to $(\ref{KGeq})$ with a negative $\mu=2\Lambda/3<0$. 

Even in the case of anti-de Sitter space itself (set $M=0$ in $(\ref{metricexp})$), the question
of the existence and uniqueness of dynamics is subtle in view of the timelike character of
the ideal boundary $\mathcal{I}$. 
It turns out that dynamics are unique  for $(\ref{KGeq})$
only if the $\mu\ge  5\Lambda/12$, whereas for the total energy to be nonnegative one must have
$\mu\ge 3\Lambda/4$. Under our conventions, the conformally covariant
 wave equation
lies between these values.
See~\cite{bachelotads, dzf}.

\begin{problem}
For suitable ranges of $\mu$, understand the boundedness and blow-up properties for
solutions  of $(\ref{KGeq})$ on Schwarzschild-anti de Sitter and  Kerr-anti de Sitter.
\end{problem}
See~\cite{hari,  cardoso} for background.

\subsection{Higher dimensions}
All the black hole solutions described above have higher dimensional analogues.
See \cite{harvey, malcolm}.
These are currently of great interest from the point of view of high energy physics.
\begin{problem}
Study all the problems of Sections~\ref{waves}--\ref{CFT} in dimension greater than $4$.
\end{problem}

Higher dimensions also brings a wealth of explicit black hole
solutions such that the topology of spatial sections of $\mathcal{H}^+$ is no longer spherical.
In particular, in $5$ spacetime dimensions
there exist ``black string'' solutions, and much more interestingly,
asymptotically flat  ``black ring'' solutions with horizon topology $S^1\times S^2$. 
See~\cite{harvey}.

\begin{problem}
Investigate the dynamics of
the wave equation $\Box_g\psi=0$ and related equations on black ring
backgrounds.
\end{problem}

\subsection{Nonlinear problems}
\label{notlinear}
The eventual goal of this subject is to study the global dynamics of the Einstein equations
$(\ref{Evac})$ themselves and in particular,
to resolve Conjecture~\ref{stabconj}. 

It may be interesting, however, to first look at simpler non-linear equations on 
fixed black hole backgrounds and ask whether decay results  of the type 
proven here   are sufficient to show non-linear stability.

The simplest non-linear perturbation of the wave equation is 
\begin{equation}
\label{powernl}
\Box_g\psi = V'(\psi)
\end{equation}
where $V= V(x)$ is a potential function. Aspects of this problem on a Schwarzschild
background have been studied by~\cite{vikolios, dr2, BlueSter, marzuola}.

\begin{problem}
Investigate the problem $(\ref{powernl})$ on  Kerr backgrounds.
\end{problem}
In particular, in view of the discussion of Section~\ref{KGproblem},
one may be able to construct solutions of $(\ref{powernl})$ with
$V= \mu \psi^2  + |\psi|^p$, for $\mu>0$ and for arbitrarily large $p$,
arising from arbitrarily small, decaying initial data,
which blow up in finite time. This would be quite interesting.

A nonlinear problem with a  stronger relation 
to $(\ref{Evac})$ is the wave map problem. 
Wave maps are maps $\Phi: \mathcal{M}\to \mathcal{N}$
where $\mathcal{M}$ is Lorentzian and $\mathcal{N}$ is Riemannian,
which are critical points of the Lagrangian
\[
\mathcal{L}(\Phi)=\int |d\Phi|^2_{g_N}
\]
In local coordinates, the equations take the form
\[
\Box_{g_M}\Phi^k = -\Gamma_{ij}^k g^{\alpha\beta}_M
(\partial_\alpha \Phi^i \partial_\beta \Phi^j),
\]
where $\Gamma^{k}_{ij}$ denote the Christoffel symbols of $g_N$.
See the lecture notes of Struwe~\cite{Struwe} for a nice introduction. 

\begin{problem}
Show global existence in the domain of outer communications
for small data solutions of the wave map problem, for arbitrary target manifold
$\mathcal{N}$,
on Schwarzschild and  Kerr backgrounds.
\end{problem}

All the above problems concern fixed black hole backgrounds. 
One of the essential difficulties in proving Conjecture~\ref{stabconj}
is dealing with a black hole background which is not known a priori, and whose geometry
must thus be recovered in a  bootstrap setting. It would be nice to have more tractable
model problems which address this difficulty. 
One can arrive at such problems by passing to symmetry classes.
The closest analogue to Conjecture~\ref{stabconj} in such a context is perhaps 
provided by the results of Holzegel~\cite{kostakis2},  which concern the dynamic stability of
the $5$-dimensional Schwarzschild as a solution of $(\ref{Evac})$, restricted under
Triaxial Bianchi IX symmetry. 
See Section~\ref{higherhigherdim}.
In the symmetric setting,
one can perhaps attain more insight  on the geometric difficulties by 
attempting a large-data
problem. For instance
\begin{problem} 
\label{kostakis;}
Show that the maximal development of asymptotically flat triaxial Bianchi IX vacuum initial data 
for the $5$-dimensional  vacuum equations 
containing a trapped surface settles down to Schwarzschild. 
\end{problem}

The analogue of the above statement has in fact been proven
for the Einstein-scalar field system under spherical symmetry~\cite{theory, dr1}.
In the direction of the above, another interesting set of problems is 
provided by the Einstein-Maxwell-charged
scalar field system under spherical symmetry. For both the charged-scalar field system and the
Bianchi IX vacuum system,
even more ambitious than Open problem~\ref{kostakis;} would be
to study the strong and weak cosmic censorship conjectures, possibly
unifying the analysis of~\cite{wcs, the, cbh}. 
Discussion of these open problems, however,
is beyond the scope of the present notes.

\section{Acknowledgements}
These notes were the basis for a course at the Clay Summer School 
on Evolution Equations which took place at ETH, Z\"urich from June 23--July 18, 2008.
A shorter version of this course was presented
as a series of lectures at the Mittag-Leffler Institute  in September 2008.

The authors thank ETH for hospitality while these notes were written, as well as the
Clay Mathematics Institute.  M.~D.~thanks in addition the Mittag-Leffler Institute in
Stockholm. M.~D.~is supported in part by a grant from the European Research Council. 
I.~R.~is supported in part by NSF grant DMS-0702270.

\appendix

\section{Lorentzian geometry}
\label{Lorge}
The reader who wishes a formal introduction to Lorentzian geometry can consult~\cite{he:lssst}.
For the reader familiar with the concepts and notations of Riemannian geometry, the following
remarks should suffice for a quick introduction.

\subsection{The Lorentzian signature}
Lorentzian geometry is defined as in Riemannian geometry, except that the
metric $g$ is not assumed positive definite, but of signature $(-,+,\ldots ,+)$. 
That is to say,
we assume that at each point $p\in\mathcal{M}^{n+1}$,\footnote{It is conventional 
to denote the dimension of the manifold by $n+1$.} we may find a basis ${\bf e}_i$ of the tangent
space $T_p\mathcal{M}$, $i=0, \ldots, n$, such that
\[
g = -{\bf e}_0\otimes {\bf e}_0+{\bf e}_1\otimes {\bf e}_1+\cdots {\bf e}_n\otimes +{\bf e}_n.
\]
In Riemannian geometry, the $-$ in the first term on the right hand side would by $+$.

A non-zero vector $v\in T_p\mathcal{M}$ is called
\emph{timelike}, \emph{spacelike}, or \emph{null},
according to whether $g(v,v)<0$, $g(v,v)>0$, or $g(v,v)=0$. 
Null and timelike vectors collectively are known as \emph{causal}. There are various
conventions for the $0$-vector. Let us not concern ourselves with such issues here.

The appellations timelike, spacelike, null 
are inherited by vector fields and immersed curves by their tangent vectors,
i.e.~a vector field $V$ is timelike if $V(p)$ is timelike, etc., and a curve
$\gamma$ is timelike if $\dot\gamma$ is timelike, etc. On the other hand,
a submanifold $\Sigma\subset\mathcal{M}$ is said to be spacelike if
its induced geometry is Riemannian, timelike if its induced geometry is
Lorentzian, and null if its induced geometry is degenerate.  (Check that these two 
definitions coincide for embedded curves.)
For a codimension-$1$ submanifold $\Sigma\subset\mathcal{M}$,
at every $p\in M$, 
there exists a non-zero normal $n^\mu$, i.e.~a vector in $T_p\mathcal{M}$
such that $g(n,v)=0$ for all $v\in T_p\Sigma$. 
It is easily seen that $\Sigma$ is spacelike iff $n$ is timelike, 
$\Sigma$ is timelike iff $n$ is spacelike, and $\Sigma$ is null iff
$n$ is null. Note that in the latter case $n\in T_p\Sigma$. The normal of $\Sigma$
is thus tangent to $\Sigma$. 

\subsection{Time-orientation and causality}
A \emph{time-orientation} on $(\mathcal{M},g)$ is defined by an equivalence class $[K]$ where $K$ is a continuous timelike vector field, where $K_1\sim K_2$ if $g(K_1,K_2)<0$.
A Lorentzian manifold admitting a time-orientation is called \emph{time-orientable}, and a triple 
$(\mathcal{M},g,[K])$ is said to be a \emph{time-oriented} Lorentzian manifold. Sometimes
one reserves the use of the word ``spacetime'' for such triples. In any case,
we shall
always consider time-oriented Lorentzian manifolds and often drop explicit mention
of the time orientation.

Given this, we may further partition causal vectors as follows. A causal vector $v$ is said to
be \emph{future-pointing} if $g(v,K)<0$, otherwise \emph{past-pointing}, 
where $K$ is a representative 
for the time orientation. As before, these names are inherited by causal curves, i.e.~we may
now talk of a \emph{future-directed} timelike curve, etc.
Given $p$, we define the \emph{causal future} $J^+(p)$ by
\[
J^+(p) = p\cup\{q\in\mathcal{M}: \exists \gamma:[0,1]\to \mathcal{M}:
\dot\gamma \hbox{\rm\ future-pointing, causal}\}
\]
Similarly, we define $J^-(p)$ where future is replaced by past in the above.
If $S\subset \mathcal{M}$ is a set, then we define
\[
J^\pm (S)=\cup_{p\in S} J^\pm(p).
\]

\subsection{Covariant derivatives, geodesics, curvature}
The standard local notions of Riemannian geometry carry over.
In particular, one defines the Christoffel symbols
\[
\Gamma^\mu_{\nu\lambda}= \frac12 g^{\mu\alpha}(\partial_\nu g_{\alpha\lambda}+
\partial_\lambda g_{\nu\alpha} -\partial_\alpha g_{\nu\lambda}),
\]
and
geodesics $\gamma(t)=(x^\alpha(t))$ are defined
as solutions to 
\[
\ddot x^\mu +\Gamma^\mu_{\nu\lambda}\dot x^\nu \dot x^\lambda = 0.
\]
Here $g_{\mu\nu}$ denote the components of $g$ with respect to 
a local coordinate system $x^\mu$, $g^{\mu\nu}$ denotes
the components of the inverse metric, 
and we are applying the Einstein summation
convention where repeated upper and lower indices are summed.
The Christoffel symbols allow us to define the \emph{covariant derivative} on
$(k,l)$ tensor fields by
\[
\nabla_\lambda  A^{\nu_1\ldots \nu_k}_{\mu_1\ldots \mu_\ell}
=\partial_\lambda A^{\nu_1\ldots \nu_k}_{\mu_1\ldots \mu_\ell}
+\sum_{i=1}^k \Gamma_{\lambda\rho}^{\nu_i} A^{\nu_1\ldots \rho\ldots  \nu_k}_{\mu_1\ldots \mu_\ell}
- \sum_{i=1}^l \Gamma_{\lambda\mu_i}^{\rho} A^{\nu_1\ldots \nu_k}_{\mu_1\ldots\rho\ldots \mu_\ell}
\]
where it is understood that $\rho$ replaces $\nu_i$, $\mu_i$, respectively in the
two terms on the right.
This defines $(k,l+1)$ tensor. 
As usual, if we contract this with a vector $v$ at $p$, then we will denote this operator
as $\nabla_v$ and we note that this can be defined in the case that the tensor field
is defined only on a curve tangent to $v$ at $p$.
We may thus express the geodesic equation as 
\[
\nabla_{\dot\gamma} \dot\gamma =0.
\]

The \emph{Riemann curvature tensor} is given by
\[
R^{\mu}_{\nu\lambda\rho}\doteq
\partial_\lambda \Gamma^{\mu}_{\rho\nu}-
\partial_\rho \Gamma^\mu_{\lambda\nu}
+\Gamma^\alpha_{\rho\nu}\Gamma^{\mu}_{\lambda\alpha}
-\Gamma^\alpha_{\lambda\nu}\Gamma^{\mu}_{\rho\alpha},
\]
and the \emph{Ricci} and \emph{scalar} curvatures by    
\[
R_{\mu\nu}\doteq R^{\alpha}_{\mu\alpha\nu}
 ,\qquad R \doteq g^{\mu\nu}R_{\mu\nu}.
\]
Using the same letter $R$ to denote all these tensors is conventional in relativity,
the number of indices indicating which tensor is being referred to.
For this reason we will avoid  writing  ``the tensor $R$''. The expression $R$ 
without indices will 
always denote the scalar curvature.
As usual, we shall also use the letter $R$ with indices to denote the
various manifestations of these tensors with indices raised and lowered
by the inverse metric and metric, e.g.
\[
R_{\mu\nu\lambda\rho} =  g_{\mu\alpha}R^{\alpha}_{\nu\lambda\rho}
\]
Note the important formula
\[
\nabla_\a \nabla_\b Z_\mu -\nabla_\a \nabla_\b Z_\mu = R_{\si \mu\a\b} Z^\si 
\]

We say that an immersed curve $\gamma:I\to \mathcal{M}$ is \emph{inextendible}
if there does not exist an immersed curve $\tilde\gamma: J\to \mathcal{M}$ where $J\supset I$
and $\tilde\gamma|_I=\gamma$.

We say that $(\mathcal{M},g)$ is
\emph{geodesically complete} if for all inextendible geodesics $\gamma:I\to \mathcal{M}$,
then $I=\mathbb R$. Otherwise, we say that it is \emph{geodesically incomplete}.
We can similarly define the notion of spacelike geodesic (in)completeness,
timelike geodesic completeness,
causal geodesic completeness, etc, by restricting the definition to such geodesics.
In the latter two cases, we may further specialise, e.g.~to the notion of \emph{future causal
geodesic completeness}, by replacing the condition $I=\mathbb R$ with $I\supset (a,\infty)$
for some $a$. 

We say that a spacelike hypersurface $\Sigma\subset \mathcal{M}$ is \emph{Cauchy} if
every inextendible causal curve in $\mathcal{M}$ intersects it precisely once.
A spacetime $(\mathcal{M},g)$ admitting such a hypersurface is called
\emph{globally hyperbolic}. This notion was first introduced by Leray~\cite{leray}.

\section{The Cauchy problem for the Einstein equations}
\label{cauchyproblem}
We outline here for reference the basic framework of
the Cauchy problem for the Einstein equations
\begin{equation}
\label{Einstcp}
R_{\mu\nu}-\frac12g_{\mu\nu}R+\Lambda g_{\mu\nu} = 8\pi T_{\mu\nu}.
\end{equation}
Here $\Lambda$ is a constant known as the \emph{cosmological constant}
and $T_{\mu\nu}$ is the so-called energy momentum tensor of matter.
We will consider mainly the
vacuum case
\begin{equation}
\label{Einstvaccp}
R_{\mu\nu}=\Lambda g_{\mu\nu},
\end{equation}
where the system closes in itself.  If the reader wants to set $\Lambda=0$, he should
feel free to do so.
To illustrate the case of matter, we will consider the example of a scalar field.

\subsection{The constraint equations}
Let $\Sigma$ be a spacelike hypersurface in $(\mathcal{M},g)$, with future directed unit
timelike normal $N$.
By definition, $\Sigma$ inherits a Riemannian metric from $g$.
On the other hand, we can define the so-called second fundamental form of $\Sigma$
to be the
symmetric covariant $2$-tensor in $T\Sigma$ defined by
\[
K(u,v)= 
-g(\nabla_u V,N)
\]
where  $V$ denotes an arbitrary extension of $v$ to a vector field along $\Sigma$,
and $\nabla$ here denotes the connection of $g$.
As in Riemannian geometry, one easily shows that the above indeed defines a tensor
on $T\Sigma$, and that it is symmetric.

Suppose now $(\mathcal{M},g)$ satisfies $(\ref{Einstcp})$ with some tensor $T_{\mu\nu}$.
With $\Sigma$ as above, let $\bar{g}_{ab}$, $\bar{\nabla}$, $K_{ab}$ denote
the induced metric, connection, and second fundamental form, respectively, of
$\Sigma$. Let barred quantities and Latin indices refer to tensors, curvature, etc.,
on $\Sigma$, and let $\Pi^\nu_a (p)$ denote the components of the
pullback map $T^*\mathcal{M} \to T^*\Sigma$.
It follows that
\begin{equation}
\label{constrai1}
\bar{R}+ (K^a_a)^2- K^a_bK^b_a = 16\pi\, T_{\mu\nu}n^\mu n^\nu  +2\Lambda,
\end{equation}
\begin{equation}
\label{constrai2}
\nabla_b K^{b}_a-\nabla_{a}K^b_b=16\pi\,  \Pi^\nu_a T_{\mu \nu}n^\mu.
\end{equation}
To see this, one
derives as in Riemannian geometry the Gauss and Codazzi equations, take traces,
and apply $(\ref{Einstcp})$.

\subsection{Initial data}
\label{initdatasec}
It is clear that $(\ref{constrai1})$, $(\ref{constrai2})$ are \emph{necessary conditions}
on the induced geometry of a spacelike hypersurface $\Sigma$ so as to arise as
a hypersurface in a spacetime satisfying $(\ref{Einstcp})$. 
As we shall see, immediately, they will also be sufficient conditions
for solving the initial value problem.

\subsubsection{The vacuum case}
Let $\Sigma$ be a $3$-manifold, $\bar{g}$ a Riemannian metric on $\Sigma$,
and $K$ a symmetric covariant $2$-tensor. We shall call $(\Sigma, \bar{g}, K)$
a \emph{vacuum initial data set with cosmological constant $\Lambda$} 
if $(\ref{constrai1})$--$(\ref{constrai2})$ are satisfied
with $T_{\mu\nu}=0$.
Note that in this case, equations $(\ref{constrai1})$--$(\ref{constrai2})$ refer
only to $\Sigma$, $\bar{g}$, $K$.

\subsubsection{The case of matter}
\label{ESF}
Let us here provide only the case for the Einstein-scalar field case.
Here, the system is $(\ref{Einstcp})$ coupled with
\begin{equation}
\label{Einst-sf1}
\Box_g\psi =0,
\end{equation}
\begin{equation}
\label{Einst-sf2}
T_{\mu\nu}= \partial_\mu\psi\partial_\nu\psi -\frac12 g_{\mu\nu} \nabla^\alpha\psi \nabla_\alpha\psi.
\end{equation}

First note that were $\Sigma$ a spacelike hypersurface in a spacetime $(\mathcal{M},g)$
satisfying the Einstein-scalar field system with massless scalar field $\psi$,
and $n^\mu$ were the future-directed normal,
then setting $\uppsi'=n^\mu\partial_\mu\phi$, $\uppsi=\phi|_\Sigma$ we have
\[
T_{\mu\nu}n^\mu n^\nu = \frac12((\uppsi')^2+\bar\nabla^a\uppsi
\bar\nabla_a\uppsi),
\]
\[
\Pi^\nu_a T_{\mu \nu}n^\mu= \uppsi'\bar\nabla_a\uppsi,
\]
where latin indices and barred quantities refer to $\Sigma$ and its induced metric
and connection.

This motivates the following:
Let $\Sigma$ be a $3$-manifold, $\bar{g}$ a Riemannian metric on $\Sigma$,
$K$ a symmetric covariant $2$-tensor, and $\uppsi:\Sigma\to \mathbb R$, $\uppsi':\Sigma\to
\mathbb R$ functions. We shall call $(\Sigma, \bar{g}, K)$
an \emph{Einstein-scalar field initial data set with cosmological constant $\Lambda$} 
if $(\ref{constrai1})$--$(\ref{constrai2})$ are satisfied
replacing $T_{\mu\nu}n^\mu n^\nu$ with $\frac12((\uppsi')^2+\bar\nabla^a\uppsi
\bar\nabla_a\uppsi)$,
and replacing $\Pi^\nu_a T_{\mu \nu}n^\mu$ with $\uppsi'\bar\nabla_a\uppsi$.

Note again that with the above replacements the equations $(\ref{constrai1})$--$(\ref{constrai2})$
do not refer to an ambient spacetime $\mathcal{M}$. See~\cite{gluing} for the 
construction of solutions
to this system.

\subsubsection{Asymptotic flatness and the positive mass theorem}
\label{asymptflat}
The study of the Einstein constraint equations is non-trivial!

Let us refer in this section to a triple $(\Sigma, \bar{g}, K)$ where $\Sigma$ is
a $3$-manifold, $\bar{g}$ a Riemannian metric, and $K$ a symmetric two-tensor
on $\Sigma$ as an \emph{initial data set}, even though we have not specified
a particular closed system of equations. 
An initial data set $(\Sigma, \bar{g}, K)$ is \emph{strongly asymptotically flat
with one end} if
there exists a compact set $\mathcal{K}\subset \Sigma$ and a coordinate
chart on $\Sigma\setminus \mathcal{K}$ which is a diffeomorphism to the complement
of a ball in $\mathbb R^3$, and for which
\[
g_{ab} = \left(1+\frac{2M}r\right)\delta_{ab} + o_2(r^{-1}), \qquad k_{ab}=o_1(r^{-2}),
\]
where $\delta_{ab}$ denotes the Euclidean metric and $r$ denotes
the Euclidean polar coordinate.

In appopriate units, $M$ is the ``mass'' measured by asymptotic observers, 
when comparing to Newtonian motion in the frame $\delta_{ab}$.
On the other hand, under the assumption of a global coordinate system well-behaved
at infinity, $M$ can be computed by integration of the $t_{0}^0$ component
of a certain pseudotensor\footnote{This is subtle: The Einstein vacuum
equations arise from the Hilbert Lagrangian $\mathcal{L}(g)=\int R$
which is $2$nd order in the metric. In local coordinates,
the highest order term is a divergence,
and the Lagrangian can thus be replaced by a new Lagrangian which is $1$st order in the metric.
The resultant
Lagrangian density, however, is no longer coordinate invariant. The quantity $t_0^0$ now arises from
``Noether's theorem''~\cite{noether}. See~\cite{notes} for a nice discussion.}   added to $T_0^0$. In this manifestation, the quantity $E= M$ is known as 
the \emph{total energy}.\footnote{With the above asymptotics, the so-called linear
momentum
 vanishes. Thus, in this case ``mass'' and energy are
equivalent.}
This relation was first studied by Einstein and is discussed 
in Weyl's book Raum-Zeit-Materie~\cite{weyl}. In one looks at $E$ for a family
of hypersurfaces with the above asymptotics, then $E$ is conserved.

A celebrated theorem of Schoen-Yau~\cite{S-Y, S-Y2} (see also~\cite{witten}) states
\begin{theorem}
\label{PMth}
Let $(\Sigma, \bar{g}, K)$ be strongly asymptotically flat with one end and
satisfy $(\ref{constrai1})$, $(\ref{constrai2})$ with $\Lambda=0$,
and where $T_{\mu\nu}n^\mu n^\nu$, $\Pi^\nu_a T_{\mu\nu}n^\mu$ are replaced by
the scalar $\mu$ and the tensor $J_a$, respectively, defined on $\Sigma$, such that moreover
$\mu\ge \sqrt{J^a J_a}$. Suppose moreover the asymptotics are strengthened
by replacing $o_2(r^{-1})$ by $O_4(r^{-2})$ and $o_1(r^{-2})$ by $O_3(r^{-3})$.
Then $M\ge0$ and $M=0$ iff $\Sigma$ embeds isometrically into $\mathbb R^{3+1}$
with induced metric $\bar{g}$ and second fundamental form $K$.
\end{theorem}

The assumption $\mu\ge \sqrt{J^a J_a}$ holds if the 
matter satisfies the \emph{dominant energy condition}~\cite{he:lssst}. 
In particular, it holds for the Einstein
scalar field system of Section~\ref{ESF}, and (of course) for the vacuum case.
The statement we have given above is weaker than the full strength of the Schoen-Yau result.
For the most general assumptions under which mass can be defined, see~\cite{bartnik}.

One can define the notion of  \emph{strongly asymptotically flat with $k$ ends}
by assuming that there exists a compact $\mathcal{K}$ such that
$\Sigma\setminus\mathcal{K}$ is a disjoint union of $k$ regions possessing
a chart as in the above definition.
The Cauchy surface $\Sigma$ of Schwarzschild of Kerr with $0\le |a| <M$,
can be chosen to be strongly asymptotically flat with $2$-ends.  The mass
of both ends
coincides with the parameter $M$ of the solution.

The above theorem applies to this case as well for the parameter $M$ associated
to any end. If $M=0$ for one end, then it follows by the rigidity statement that there is only one end.
Note why Schwarzschild with $M<0$ does not provide a counterexample.

The association of ``naked singularities'' with negative mass Schwarzschild gave
the impression that the positive energy theorem protects against 
naked singularities. This is not true! See the examples discussed in Section~\ref{nakedsings}.

In the presence of black holes, one expects a strengthening of the lower bound
on mass in Theorem~\ref{PMth} to include a term related to the square root of
the area of a cross section of the
horizon. Such inequalities were first discussed by 
Penrose~\cite{Penr2} with the Bondi mass
in place of the mass defined above. All inequalities of this type are often called
\emph{Penrose inequalities}.   It is not clear what this 
term should be, as the horizon is only identifyable after
global properties of the maximal development have been understood. Thus,
one often replaces this area in the conjectured inequality
with the area of a suitably defined apparent horizon.
Such a statement
 has indeed been obtained
in the so-called Riemannian case (corresponding to $K=0$)
where the relevant notion of apparent horizon coincides with that of 
minimal surface. See the important papers of
Huisken-Ilmanen~\cite{huiskenilm} and Bray~\cite{hugh}.

\subsection{The maximal development}
\label{maxdevy}
Let $(\Sigma, \bar{g}, K)$ denote a smooth vacuum initial data set with cosmological
constant $\Lambda$.
We say that  a smooth spacetime $(\mathcal{M}, g)$ is 
a smooth \emph{development of initial data} if
\begin{enumerate}
\item
\label{sateq}
$(\mathcal{M},g)$ satisfies the Einstein vacuum equations $(\ref{Evac})$
with cosmological constant
$\Lambda$. 
\item
\label{cau}
There exists a smooth embedding $i:\Sigma\to \mathcal{M}$ such that
$(\mathcal{M},g)$ is globally hyperbolic with Cauchy surface
$i(\Sigma)$, and $\bar{g}$, $K$ are the induced metric and second fundamental
form, respectively.
\end{enumerate}

The original local existence and uniqueness theorems were proven in 1952 by 
Choquet-Bruhat~\cite{choquet2}.\footnote{Then called Four\`es-Bruhat.}
In modern language, they can be formulated as follows
\begin{theorem}
\label{exist}
Let $(\Sigma, \bar{g}, K)$ be as in the statement of the above theorem. 
Then there exists a smooth development $(\mathcal{M},g)$ of initial data.
\end{theorem}
\begin{theorem}
\label{unique}
Let $\mathcal{M}$, $\widetilde{\mathcal{M}}$ be two smooth developments
of initial data.
Then there exists a third development $\mathcal{M}'$ and isometric embeddings
$j:\mathcal{M}'\to \mathcal{M}$, $\tilde{j}:\mathcal{M}'\to\widetilde{\mathcal{M}}$
commuting with $i$, $\tilde{i}$.
\end{theorem}

Application of
Zorn's lemma,
the above two theorems
and simple facts about Lorentzian causality yields:
\begin{theorem}~(Choquet-Bruhat--Geroch \cite{chge:givp})
\label{Maxdev}
Let $(\Sigma, \bar{g}, K)$ denote a smooth vacuum initial data set with cosmological
constant $\Lambda$.
Then there exists a unique development of initial data $(\mathcal{M},g)$
satisfying
the following maximality statement:
If $(\widetilde{\mathcal{M}},\widetilde{g})$ satisfies $(\ref{sateq})$, $(\ref{cau})$ with embedding
$\tilde{i}$, then
there exists an isometric embedding $j:\widetilde{\mathcal{M}}\to \mathcal{M}$
such that $j$ commutes with $\tilde{i}$.
\end{theorem}

The spacetime $(\mathcal{M},g)$ is known as the \emph{maximal development}
of $(\Sigma, \bar{g}, K)$. The spacetime $\mathcal{M}\cap J^+(\Sigma)$
is known as the \emph{maximal future development} and $\mathcal{M}\cap J^-(\Sigma)$
the \emph{maximal past development}.

We have formulated the above theorems in the class of smooth initial data. They are of course
proven in classes of finite regularity. There has been much recent work in 
proving a version of 
Theorem~\ref{exist} under minimal regularity assumptions. The current state of the art
requires only $\bar{g}\in H^{2+\epsilon}$, $K\in H^{1+\epsilon}$. See~\cite{rough}.

We leave as an {\bf exercise} formulating the analogue of Theorem~\ref{Maxdev}
for the Einstein-scalar field system $(\ref{Einstcp})$, $(\ref{Einst-sf1})$, $(\ref{Einst-sf2})$,
where the notion of initial data set is that given in Section~\ref{ESF}.

\subsection{Harmonic coordinates and the proof of local existence}
\label{harmonic}
The statements of Theorems~\ref{exist}  and~\ref{unique} are coordinate
independent. Their proofs, however, require fixing a gauge which determines
the form of the metric functions in coordinates from initial data. The classic
gauge is the so-called \emph{harmonic} gauge\footnote{also known as wave coordinates}.
Here the coordinates $x^\mu$ are required to satisfy
\begin{equation}
\label{harcor1}
\Box_g x^\mu =0.
\end{equation}
Equivalently, this gauge is characterized by the condition
\begin{equation}
\label{harcor2}
g^{\mu\nu}\Gamma^\alpha_{\mu\nu}=0.
\end{equation}

A linearised version of these coordinates
was used by Einstein~\cite{waves} to predict gravitational waves. 
It appears that  de Donder~\cite{deDonder}
was the first to consider harmonic coordinates in general.
These coordinates are discussed extensively in the book of Fock~\cite{fock}.

The result of Theorem~\ref{unique} 
actually predates Theorem~\ref{exist}, and in some
form was first proven by  Stellmacher~\cite{Stell}.
Given two developments $(\mathcal{M},g)$, $(\widetilde{\mathcal{M}},\tilde g)$
one constructs for each harmonic coordinates $x^\mu$, $\tilde{x}^\mu$
adapted to $\Sigma$, such that $g_{\mu\nu}=\tilde{g}_{\mu\nu}$,
$\partial_\lambda g_{\mu\nu} =\partial_{\lambda} \tilde{g}_{\mu\nu}$
along $\Sigma$.
In these coordinates, the Einstein vaccum equations can be expressed as
\begin{equation}
\label{reduced}
\Box_g g^{\mu\nu} = Q^{\mu\nu, \alpha\beta}_{\iota\kappa \lambda\rho\sigma\tau} \,g^{\iota\kappa}
\, \partial_\alpha g^{\lambda\rho}
\, \partial_\beta g^{\sigma\tau}
\end{equation}
for which uniqueness follows from   general results of Schauder~\cite{Schauder}.
This theorem gives in addition a domain of dependence property.\footnote{There
is even earlier work on uniqueness in the analytic category going back to Hilbert,
appealing to Cauchy-Kovalevskaya.  Unfortunately, nature is not analytic; in particular, one
cannot infer the domain of dependence property from those considerations.}

Existence for solutions of the system $(\ref{reduced})$ with smooth initial data
would also follow
from the results of Schauder~\cite{Schauder}.
This does not immediately yield a proof of  Theorem~\ref{exist}, because one does not
have a priori the spacetime metric $g$ so as to impose $(\ref{harcor1})$ or $(\ref{harcor2})$!
The crucial observation is that if $(\ref{harcor2})$ is true ``to first order'' on $\Sigma$,
and $g$ is defined to be the unique solution to $(\ref{reduced})$, then 
$(\ref{harcor2})$ will hold, and thus, $g$ will solve $(\ref{Einstcp})$. 
Thus, to prove Theorem~\ref{exist}, it suffices to show that one can
arrange for $(\ref{harcor2})$ to be true ``to first order'' initially. 
Choquet-Bruhat~\cite{choquet2} showed that this can be done
precisely when the constraint equations $(\ref{constrai1})$--$(\ref{constrai2})$ are 
satisfied with vanishing right hand side.
Interestingly, to obtain existence for $(\ref{reduced})$,
Choquet-Bruhat's proof~\cite{choquet2}  does not in fact
appeal to the techniques of Schauder~\cite{Schauder},
but, following Sobolev, 
rests on a Kirchhoff formula representation of the solution. 
Recently, new representations of
this type have found applications to refined extension criteria~\cite{parametrix}.

An interesting feature of the classical existence and uniqueness proofs is that
Theorem~\ref{unique} requires more regularity than Theorem~\ref{exist}. This is 
because solutions of $(\ref{harcor1})$ are a priori only as regular as the metric. 
This difficulty has recently been overcome in~\cite{Fabrod}.

\subsection{Stability of Minkowski space}
\label{stabsect}
The most celebrated global result on the Einstein equations 
is the stability of Minkowski space, first proven in monumental work of
Christodoulou and Klainerman~\cite{book}:
\begin{theorem}
\label{monumental}
Let $(\Sigma,\bar{g}, K)$ be a strongly asymptotically flat vacuum initial data set, assumed
sufficiently close to Minkowski space in a weighted sense. Then the maximal development
is geodesically complete, and the spacetime approaches Minkowski space (with quantitative
decay rates) in all directions. Moreover, a complete future null infinity $\mathcal{I}^+$ can be attached
to the spacetime such that $J^-(\mathcal{I}^+)=\mathcal{M}$.
\end{theorem}
The above theorem also allows one to rigorously define the laws of gravitational radiation.
These laws are nonlinear even at infinity. 
Theorem~\ref{monumental} led to the discovery
 of Christodoulou's memory effect~\cite{memory}.

A new proof of a version of stability of Minkowski space using
harmonic coordinates has been given in~\cite{LR}.
This has now been extended in various directions in~\cite{loizelet}.
The original result~\cite{book} was
extended to the Maxwell case in the Ph.D.~thesis of Zipser~\cite{zipser}.
 Bieri~\cite{bieri0} has very recently
given a proof of a version of stability of Minkowski space
under weak asymptotics and regularity
assumptions, following the basic setup of~\cite{book}. 

There was an earlier semi-global result 
of Friedrich~\cite{Fr2} where initial data were prescribed
on a hyperboloidal initial hypersurface meeting $\mathcal{I}^+$.

A common misconception is that it is the positivity of mass which is somehow
responsible for the stability
of Minkowski space. 
The results of~\cite{LR} for this are very telling, for they apply
not only to the Einstein-vacuum equations, but also to the Einstein-scalar field 
system of Section~\ref{ESF},
including the case where the definition of the energy-momentum tensor $(\ref{Einst-sf2})$ 
is replaced with its negative. 
Minkowski space is then not even a local minimizer for the mass 
functional in the class of perturbations
allowed! Nonetheless, by the results of~\cite{LR}, Minkowski space 
is still stable in this context.

Another point which cannot be overemphasized: It is essential that the smallness
in $(\ref{monumental})$ concern a weighted norm. Compare with the results of
Section~\ref{vaccol}. 
 
Stability of Minkowski space is the only truly global result on the maximal development
which has been obtained for asymptotically flat initial data without symmetry. 
There are a number of important results 
applicable in cosmological settings, due to Friedrich~\cite{Fr2}, 
Andersson-Moncrief~\cite{larsvince}, and most recently Ringstrom~\cite{ringstrom}.
 
Other than this, our current global understanding of solutions to the Einstein equations 
(in particular all work on the cosmic censorship conjectures)
has been confined to solutions under symmetry. We have given many such references
in the asymptotically flat setting in the course of Section~\ref{introsec}. 
The cosmological setting is beyond the scope of these notes, but we refer the
reader to the recent review article and
book of Rendall~\cite{rendallrev, rendall}  for an overview and many references.

\section{The divergence theorem}
\label{divthesec}
Let $(\mathcal{M},g)$ be a spacetime, 
and let $\Sigma_0$, $\Sigma_1$ be homologous spacelike hypersurfaces with 
common boundary, bounding a spacetime region $\mathcal{B}$,
with $\Sigma_1 \subset J^+(\Sigma_0)$. 
Let $n^\mu_0$, $n^\mu_1$ denote the future unit normals of $\Sigma_0$, $\Sigma_1$
respectively, and let $P_\mu$ denote a one-form. 
Under our convention on the signature,
the divergence theorem takes the form
\begin{equation}
\label{divthe}
\int_{\Sigma_1} P_\mu n^\mu_1  + \int_{\mathcal{B}}
\nabla^\mu P_\mu =
\int_{\Sigma_0} P_\mu n^\mu_0,
\end{equation}
where all integrals are with respect to the \emph{induced volume form}.

This is defined as follows. The volume form of spacetime is
\[
\sqrt{-\det g}  dx^0\ldots dx^n
\]
where $\det g$ denotes the determinant of the matrix $g_{\alpha\beta}$ in the above 
coordinates. The induced volume form of a spacelike hypersurface
is defined as in Riemannian geometry.  

We will also consider the case where (part of) $\Sigma_1$ is null.
Then, we choose  arbitrarily a future directednull generator $n^\Sigma_1$ for $\Sigma_1$ arbitrarily and define the
volume element so that the divergence theorem applies.
For instance the divergence theorem in the 
region $\mathcal{R}(\tau',\tau'')$ (described in the lectures) for an arbitrary
current $P_\mu$
then takes the form
\[
\int_{\Sigma_{\tau''}} P_\mu n^\mu_{\Sigma_{\tau''}}+
\int_{\mathcal{H}(\tau',\tau'')} P_\mu n^\mu_{\mathcal{H}}
+\int_{\mathcal{R}(\tau',\tau'')} \nabla^\mu P_\mu =
\int_{\Sigma_{\tau'}} P_\mu n^\mu_{\Sigma_{\tau'}},
\]
where the volume elements are as described.

Note how the form of this theorem can change depending on sign conventions regarding
the directions of the normal, the definition of the divergence and the signature of the metric.

\section{Vector field multipliers and their currents}
\label{VFM}
Let
$\psi$ be a solution of
\begin{equation}
\label{WE}
\Box_g\psi=0
\end{equation}
on a Lorentzian manifold $(\mathcal{M},g)$.
Define 
\begin{equation}
\label{noec}
T_{\mu\nu}(\psi)=\partial_\mu\psi\partial_\nu\psi -\frac12 g_{\mu\nu}\partial^\alpha\psi
\partial_\alpha\psi
\end{equation}
We call $T_{\mu\nu}$ the \emph{energy-momentum} tensor of $\psi$.\footnote{Note that this
is the same expression that appears on the right hand side of $(\ref{Einstcp})$ in the
Einstein-scalar field system. See Section~\ref{ESF}.} Note the 
symmetry property 
\[
T_{\mu\nu}=T_{\nu\mu}.
\]
The wave equation $(\ref{WE})$ implies
\begin{equation}
\label{divfree}
\nabla^\mu T_{\mu\nu}=0.
\end{equation}

Given a vector field $V^\mu$, we may define
the associated currents
\begin{equation}
\label{assoc1}
J^V_\mu(\psi) = V^\nu T_{\mu\nu}(\psi)
\end{equation}
\begin{equation}
\label{assoc2}
K^V= {}^V\pi_{\mu\nu}T^{\mu\nu}(\psi)
\end{equation}
where $\pi^X$ is the deformation tensor defined by
\[
{}^X\pi_{\mu\nu}= \frac12\nabla_{(\mu} X_{\nu)}=\frac12 (\mathcal{L}_Xg)_{\mu\nu}.
\]
The identity $(\ref{divfree})$ gives
\[
\nabla^\mu J^V_\mu (\psi)= K^V(\psi).
\]

Note that $J^V_\mu(\psi)$ and $K^V(\psi)$ both depend only on the $1$-jet of $\psi$,
yet the latter is the divergence of the former.         
Applying the divergence theorem $(\ref{divthe})$, this allows one to relate
quantities of the same order. 

The existence of a tensor $T_{\mu\nu}(\psi)$  satisfying 
$(\ref{divfree})$ follows from the
fact that equation~$(\ref{WE})$ derives from a Lagrangian of a specific type.
These issues were first systematically studied by Noether~\cite{noether}.
For more general such Lagrangian theories, two
currents $J_\mu$, $K$ with $\nabla^\mu J_\mu = K$, both depending only on the $1$-jet,
but not necessarily arising from $T_{\mu\nu}$ as above, are
known as  \emph{compatible currents}. These have been introduced and classified
by Christodoulou~\cite{book2}.

\section{Vector field commutators}
\label{commutat}

\begin{proposition}
Let $\psi$ be a solution of the equation of the scalar equation
$$
\Box_g\psi=f,
$$
and $X$ be a vectorfield. Then
$$
\Box_g (X\psi) = X(f) - 2\, ^X\pi^{\a\b}  \nabla_\a \nabla_\b \psi - 2 \left (2 (\nabla^\a\, ^X\pi_{\a\mu})   - 
(\nabla_\mu\, ^X\pi^\a_{\a} ) \right) \nabla^\mu\psi.
$$
\end{proposition}

\begin{proof}
To  show this we write 
$$
X(\Box_g \psi) ={\mathcal L}_X (g^{\a\b} \nabla_\a \nabla_\b \psi)= 2\, ^X\pi^{\a\b}  
\nabla_\a \nabla_\b \psi + 
g^{\a\b} {\mathcal L}_X (\nabla_\a \nabla_\b \psi).
$$
Furthermore,
$$
{\mathcal L}_X (\nabla_\a \nabla_\b \psi)-\nabla_\a {\mathcal L}_X 
\nabla_\b\psi=2 \left ((\nabla_\b\, ^X\pi_{\a\mu})   - 
(\nabla_\mu\, ^X\pi_{\b\a} ) +(\nabla_\a \,^X\pi_{\mu\b}) \right) \nabla^\mu\psi
$$
and 
$$
{\mathcal L}_X \nabla_\b\psi = \nabla_X \nabla_\b\psi +
\nabla_\b X^\mu \nabla_\mu \psi= \nabla_\b (X\psi).
$$
\end{proof}

\section{Some useful Schwarzschild computations}
In this section, $(\mathcal{M},g)$ refers to maximal Schwarzschild with $M>0$,
$\mathcal{Q}=\mathcal{M}/SO(3)$, $\mathcal{I}^\pm$, $J^\mp (\mathcal{I}^\pm)$ are as defined
in Section~\ref{pendiag}.
\label{computs}
\subsection{Schwarzschild coordinates $(r,t)$}
The coordinates are $(r,t)$ and 
the metric takes the form
\[
-(1-2M/r)dt^2+(1-2M/r)^{-1} dr^2 +r^2 d\sigma_{\mathbb S^2}
\]
These coordinates can be used to cover any of the four connected components
of  $\mathcal{Q}\setminus\mathcal{H}^\pm$.
In particular, the region $J^-(\mathcal{I}^+_A)\cap J^+(\mathcal{I}^-_A)$ (where 
$\mathcal{I}^\pm_A$ correspond to a pair of connected components
of $\mathcal{I}^\pm$ sharing a limit point in the embedding) is covered
by a Schwarzschild coordinate system where 
$2M<r<\infty$, $-\infty<t<\infty$. Note that $r$
has an invariant characterization
namely $r(x)=\sqrt{\rm Area(S)/4\pi}$ where $S$ is the unique group orbit 
of the ${\rm SO}(3)$ action containing $x$.\footnote{Compare with the Minkowski case
$M=0$ where the ${\rm SO}(3)$ action is of course not unique.}

The hypersurface $\{t=c\}$ in the Schwarzschild coordinate region
$J^-(\mathcal{I}^+_A)\cap J^+(\mathcal{I}^-_A)$ 
 extends regularly to a hypersurface with boundary
in $\mathcal{M}$
where the boundary is precisely $\mathcal{H}^+\cap\mathcal{H}^-$.

The coordinate vector field $\partial_t$ is Killing (and extends to the globally defined
Killing field $T$).

In a slight abuse of notation, we will often
extend Schwarzschild coordinate notation to $\mathcal{D}$, the closure of 
$J^-(\mathcal{I}^+_A)\cap J^+(\mathcal{I}^-_A)$. For instance, 
we may talk of the vector 
field $\partial_t$ ``on'' $\mathcal{H}^\pm$, or of $\{t=c\}$ having boundary 
$\mathcal{H}^+\cap\mathcal{H}^-$, etc.

\subsection{Regge-Wheeler coordinates $(r^*,t)$}
\label{RWcs}
Here  $t$ is as before and
\begin{equation}
\label{burada}
r^*= r+2M\log (r-2M)-3M-2M\log M
\end{equation}
and the metric takes the form
\[
-(1-2M/r)(-dt^2+(dr^*)^2) + r^2 d\sigma_{\mathbb S^2}
\]
where $r$ is defined implictly by $(\ref{burada})$.
A coordinate
chart defined in $-\infty<r^*<\infty$, $-\infty<t<\infty$
covers $J^-(\mathcal{I}^+_A)\cap J^+(\mathcal{I}^-_A)$.

The constant renormalisation of the coordinate is taken so that $r^*=0$ at the photon sphere,
where $r=3M$.

Note the explicit form of the wave operator
\[
\Box_{g}\psi = -(1-2M/r)^{-1} (\partial_t^2\psi-r^{-2}\partial_{r^*} (r^2\partial_{r^*}\psi))
+\nabb^A\nabb_A\psi
\]
where $\nabb$ denotes the induced covariant derivative on the group orbit spheres.

Similar warnings of abuse of notation apply, for instance, we may write
$\partial_t=\partial_{r*}$ on $\mathcal{H}^+$.

\subsection{Double null coordinates $(u,v)$}
\label{dnc}
Our convention is to define
\[
u=\frac12(t-r^*),
\]
\[
v=\frac12(t+r^*).
\]
The metric takes the form
\[
-4(1-2M/ r)du dv + r^2 d\sigma_{\mathbb S^2}
\]
and $J^-(\mathcal{I}^+_A)\cap J^+(\mathcal{I}^-_A)$ is covered
by a chart $-\infty<u<\infty$, $-\infty<v<\infty$.

The usual comments about abuse of notation hold, in particular, 
we may now parametrize $\mathcal{H}^+\cap \mathcal{D}$ with
$\{\infty\}\times [-\infty,\infty)$ and similarly
$\mathcal{H}^-\cap \mathcal{D}$ with $(-\infty,\infty]\times\{-\infty\}$, and write
$\partial_v(-\infty,v)=\partial_t(-\infty,v)$, $\partial_u(-\infty,v)=0$.

Note that the vector field
$(1-2M/r)^{-1}\partial_u$ extends to a regular vector null field across 
$\mathcal{H}^+\setminus\mathcal{H}^-$.
Thus, with the basis $\partial_v$, $(1-2M/r)^{-1}\partial_u$,
one can choose regular vector fields near $\mathcal{H}^+\setminus \mathcal{H}^-$ without
changing to regular coordinates. In practice, this can be convenient.

\end{document}

%% file: schw0.pstex_t
\begin{picture}(0,0)%
\includegraphics{schw0.pstex}%
\end{picture}%
\setlength{\unitlength}{2763sp}%
\begingroup\makeatletter\ifx\SetFigFont\undefined%
\gdef\SetFigFont#1#2#3#4#5{%
  \reset@font\fontsize{#1}{#2pt}%
  \fontfamily{#3}\fontseries{#4}\fontshape{#5}%
  \selectfont}%
\fi\endgroup%
\begin{picture}(2298,3268)(3393,-4348)
\put(3526,-3661){\rotatebox{90.0}{\makebox(0,0)[lb]{\smash{{\SetFigFont{8}{9.6}{\rmdefault}{\mddefault}{\updefault}{\color[rgb]{0,0,0}$r=0$}%
}}}}}
\put(4500,-3961){\rotatebox{90.0}{\makebox(0,0)[lb]{\smash{{\SetFigFont{8}{9.6}{\rmdefault}{\mddefault}{\updefault}{\color[rgb]{0,0,0}$r=R_0$}%
}}}}}
\end{picture}%

%% file: Schw1.pstex_t
\begin{picture}(0,0)%
\includegraphics{schw1.pstex}%
\end{picture}%
\setlength{\unitlength}{2368sp}%
\begingroup\makeatletter\ifx\SetFigFont\undefined%
\gdef\SetFigFont#1#2#3#4#5{%
  \reset@font\fontsize{#1}{#2pt}%
  \fontfamily{#3}\fontseries{#4}\fontshape{#5}%
  \selectfont}%
\fi\endgroup%
\begin{picture}(2309,3249)(3393,-4198)
\put(3526,-3661){\rotatebox{90.0}{\makebox(0,0)[lb]{\smash{{\SetFigFont{7}{8.4}{\rmdefault}{\mddefault}{\updefault}{\color[rgb]{0,0,0}$r=2M$}%
}}}}}
\end{picture}%

%% file: Schw2.pstex_t
\begin{picture}(0,0)%
\includegraphics{schw2.pstex}%
\end{picture}%
\setlength{\unitlength}{2763sp}%
\begingroup\makeatletter\ifx\SetFigFont\undefined%
\gdef\SetFigFont#1#2#3#4#5{%
  \reset@font\fontsize{#1}{#2pt}%
  \fontfamily{#3}\fontseries{#4}\fontshape{#5}%
  \selectfont}%
\fi\endgroup%
\begin{picture}(2309,3256)(3393,-4205)
\put(3526,-3661){\rotatebox{90.0}{\makebox(0,0)[lb]{\smash{{\SetFigFont{8}{9.6}{\rmdefault}{\mddefault}{\updefault}{\color[rgb]{0,0,0}$r=0$}%
}}}}}
\put(4500,-3961){\rotatebox{90.0}{\makebox(0,0)[lb]{\smash{{\SetFigFont{8}{9.6}{\rmdefault}{\mddefault}{\updefault}{\color[rgb]{0,0,0}$r=2M$}%
}}}}}
\end{picture}%

%% file: krusk.pstex_t
\begin{picture}(0,0)%
\includegraphics{krusk.pstex}%
\end{picture}%
\setlength{\unitlength}{2763sp}%
\begingroup\makeatletter\ifx\SetFigFont\undefined%
\gdef\SetFigFont#1#2#3#4#5{%
  \reset@font\fontsize{#1}{#2pt}%
  \fontfamily{#3}\fontseries{#4}\fontshape{#5}%
  \selectfont}%
\fi\endgroup%
\begin{picture}(3302,2574)(2925,-7048)
\put(4276,-4786){\makebox(0,0)[lb]{\smash{{\SetFigFont{8}{9.6}{\rmdefault}{\mddefault}{\updefault}{\color[rgb]{0,0,0}$r=0$}%
}}}}
\put(3651,-5171){\rotatebox{315.0}{\makebox(0,0)[lb]{\smash{{\SetFigFont{8}{9.6}{\rmdefault}{\mddefault}{\updefault}{\color[rgb]{0,0,0}$T=-R$}%
}}}}}
\put(4721,-5446){\rotatebox{45.0}{\makebox(0,0)[lb]{\smash{{\SetFigFont{8}{9.6}{\rmdefault}{\mddefault}{\updefault}{\color[rgb]{0,0,0}$T=R$}%
}}}}}
\put(4252,-6861){\makebox(0,0)[lb]{\smash{{\SetFigFont{8}{9.6}{\rmdefault}{\mddefault}{\updefault}{\color[rgb]{0,0,0}$r=0$}%
}}}}
\end{picture}%

%% file: schw3.pstex_t
\begin{picture}(0,0)%
\includegraphics{schw3.pstex}%
\end{picture}%
\setlength{\unitlength}{2763sp}%
\begingroup\makeatletter\ifx\SetFigFont\undefined%
\gdef\SetFigFont#1#2#3#4#5{%
  \reset@font\fontsize{#1}{#2pt}%
  \fontfamily{#3}\fontseries{#4}\fontshape{#5}%
  \selectfont}%
\fi\endgroup%
\begin{picture}(6040,3181)(1576,-7385)
\put(1576,-5836){\makebox(0,0)[lb]{\smash{{\SetFigFont{8}{9.6}{\rmdefault}{\mddefault}{\updefault}{\color[rgb]{0,0,0}$i^0$}%
}}}}
\put(2401,-5236){\rotatebox{45.0}{\makebox(0,0)[lb]{\smash{{\SetFigFont{8}{9.6}{\rmdefault}{\mddefault}{\updefault}{\color[rgb]{0,0,0}$\mathcal{I}^+$}%
}}}}}
\put(5776,-4336){\makebox(0,0)[lb]{\smash{{\SetFigFont{8}{9.6}{\rmdefault}{\mddefault}{\updefault}{\color[rgb]{0,0,0}$i^+$}%
}}}}
\put(3001,-4336){\makebox(0,0)[lb]{\smash{{\SetFigFont{8}{9.6}{\rmdefault}{\mddefault}{\updefault}{\color[rgb]{0,0,0}$i^+$}%
}}}}
\put(3001,-7336){\makebox(0,0)[lb]{\smash{{\SetFigFont{8}{9.6}{\rmdefault}{\mddefault}{\updefault}{\color[rgb]{0,0,0}$i^-$}%
}}}}
\put(5626,-7336){\makebox(0,0)[lb]{\smash{{\SetFigFont{8}{9.6}{\rmdefault}{\mddefault}{\updefault}{\color[rgb]{0,0,0}$i^-$}%
}}}}
\put(7201,-5836){\makebox(0,0)[lb]{\smash{{\SetFigFont{8}{9.6}{\rmdefault}{\mddefault}{\updefault}{\color[rgb]{0,0,0}$i^0$}%
}}}}
\put(4721,-5446){\rotatebox{45.0}{\makebox(0,0)[lb]{\smash{{\SetFigFont{8}{9.6}{\rmdefault}{\mddefault}{\updefault}{\color[rgb]{0,0,0}$r=2M$}%
}}}}}
\put(4276,-7261){\makebox(0,0)[lb]{\smash{{\SetFigFont{8}{9.6}{\rmdefault}{\mddefault}{\updefault}{\color[rgb]{0,0,0}$r=0$}%
}}}}
\put(4351,-4411){\makebox(0,0)[lb]{\smash{{\SetFigFont{8}{9.6}{\rmdefault}{\mddefault}{\updefault}{\color[rgb]{0,0,0}$r=0$}%
}}}}
\put(3651,-5171){\rotatebox{315.0}{\makebox(0,0)[lb]{\smash{{\SetFigFont{8}{9.6}{\rmdefault}{\mddefault}{\updefault}{\color[rgb]{0,0,0}$r=2M$}%
}}}}}
\put(2251,-6361){\rotatebox{315.0}{\makebox(0,0)[lb]{\smash{{\SetFigFont{8}{9.6}{\rmdefault}{\mddefault}{\updefault}{\color[rgb]{0,0,0}$\mathcal{I}^-$}%
}}}}}
\put(6301,-6736){\rotatebox{45.0}{\makebox(0,0)[lb]{\smash{{\SetFigFont{8}{9.6}{\rmdefault}{\mddefault}{\updefault}{\color[rgb]{0,0,0}$\mathcal{I}^-$}%
}}}}}
\put(6226,-4861){\rotatebox{315.0}{\makebox(0,0)[lb]{\smash{{\SetFigFont{8}{9.6}{\rmdefault}{\mddefault}{\updefault}{\color[rgb]{0,0,0}$\mathcal{I}^+$}%
}}}}}
\end{picture}%

%% file: compl.pstex_t
\begin{picture}(0,0)%
\includegraphics{compl.pstex}%
\end{picture}%
\setlength{\unitlength}{2763sp}%
\begingroup\makeatletter\ifx\SetFigFont\undefined%
\gdef\SetFigFont#1#2#3#4#5{%
  \reset@font\fontsize{#1}{#2pt}%
  \fontfamily{#3}\fontseries{#4}\fontshape{#5}%
  \selectfont}%
\fi\endgroup%
\begin{picture}(1782,2660)(5325,-7091)
\put(6226,-4861){\rotatebox{315.0}{\makebox(0,0)[lb]{\smash{{\SetFigFont{8}{9.6}{\rmdefault}{\mddefault}{\updefault}{\color[rgb]{0,0,0}$\mathcal{I}^+$}%
}}}}}
\put(6301,-6736){\rotatebox{45.0}{\makebox(0,0)[lb]{\smash{{\SetFigFont{8}{9.6}{\rmdefault}{\mddefault}{\updefault}{\color[rgb]{0,0,0}$\mathcal{I}^-$}%
}}}}}
\end{picture}%

%% file: Minko.pstex_t
\begin{picture}(0,0)%
\includegraphics{Minko.pstex}%
\end{picture}%
\setlength{\unitlength}{2763sp}%
\begingroup\makeatletter\ifx\SetFigFont\undefined%
\gdef\SetFigFont#1#2#3#4#5{%
  \reset@font\fontsize{#1}{#2pt}%
  \fontfamily{#3}\fontseries{#4}\fontshape{#5}%
  \selectfont}%
\fi\endgroup%
\begin{picture}(1539,2660)(5568,-7091)
\put(5701,-5986){\rotatebox{90.0}{\makebox(0,0)[lb]{\smash{{\SetFigFont{8}{9.6}{\rmdefault}{\mddefault}{\updefault}{\color[rgb]{0,0,0}$r=0$}%
}}}}}
\put(6226,-4861){\rotatebox{315.0}{\makebox(0,0)[lb]{\smash{{\SetFigFont{8}{9.6}{\rmdefault}{\mddefault}{\updefault}{\color[rgb]{0,0,0}$\mathcal{I}^+$}%
}}}}}
\put(6301,-6736){\rotatebox{45.0}{\makebox(0,0)[lb]{\smash{{\SetFigFont{8}{9.6}{\rmdefault}{\mddefault}{\updefault}{\color[rgb]{0,0,0}$\mathcal{I}^-$}%
}}}}}
\end{picture}%

%% file: collapse.pstex_t
\begin{picture}(0,0)%
\includegraphics{collapse.pstex}%
\end{picture}%
\setlength{\unitlength}{2763sp}%
\begingroup\makeatletter\ifx\SetFigFont\undefined%
\gdef\SetFigFont#1#2#3#4#5{%
  \reset@font\fontsize{#1}{#2pt}%
  \fontfamily{#3}\fontseries{#4}\fontshape{#5}%
  \selectfont}%
\fi\endgroup%
\begin{picture}(2195,1421)(4912,-5816)
\put(5045,-5342){\rotatebox{90.0}{\makebox(0,0)[lb]{\smash{{\SetFigFont{8}{9.6}{\rmdefault}{\mddefault}{\updefault}{\color[rgb]{0,0,0}$r=0$}%
}}}}}
\put(5252,-4527){\makebox(0,0)[lb]{\smash{{\SetFigFont{8}{9.6}{\rmdefault}{\mddefault}{\updefault}{\color[rgb]{0,0,0}$r=0$}%
}}}}
\put(5438,-5077){\rotatebox{45.0}{\makebox(0,0)[lb]{\smash{{\SetFigFont{8}{9.6}{\rmdefault}{\mddefault}{\updefault}{\color[rgb]{0,0,0}$\mathcal{H}^+$}%
}}}}}
\put(6226,-4861){\rotatebox{315.0}{\makebox(0,0)[lb]{\smash{{\SetFigFont{8}{9.6}{\rmdefault}{\mddefault}{\updefault}{\color[rgb]{0,0,0}$\mathcal{I}^+$}%
}}}}}
\end{picture}%

%% file: negmass.pstex_t
\begin{picture}(0,0)%
\includegraphics{negmass.pstex}%
\end{picture}%
\setlength{\unitlength}{2763sp}%
\begingroup\makeatletter\ifx\SetFigFont\undefined%
\gdef\SetFigFont#1#2#3#4#5{%
  \reset@font\fontsize{#1}{#2pt}%
  \fontfamily{#3}\fontseries{#4}\fontshape{#5}%
  \selectfont}%
\fi\endgroup%
\begin{picture}(1539,2660)(5568,-7091)
\put(5701,-5986){\rotatebox{90.0}{\makebox(0,0)[lb]{\smash{{\SetFigFont{8}{9.6}{\rmdefault}{\mddefault}{\updefault}{\color[rgb]{0,0,0}$r=0$}%
}}}}}
\put(6301,-6736){\rotatebox{45.0}{\makebox(0,0)[lb]{\smash{{\SetFigFont{8}{9.6}{\rmdefault}{\mddefault}{\updefault}{\color[rgb]{0,0,0}$\mathcal{I}^-$}%
}}}}}
\put(6226,-4861){\rotatebox{315.0}{\makebox(0,0)[lb]{\smash{{\SetFigFont{8}{9.6}{\rmdefault}{\mddefault}{\updefault}{\color[rgb]{0,0,0}$\mathcal{I}^+$}%
}}}}}
\end{picture}%

%% file: Cauchy.pstex_t
\begin{picture}(0,0)%
\includegraphics{Cauchy.pstex}%
\end{picture}%
\setlength{\unitlength}{2763sp}%
\begingroup\makeatletter\ifx\SetFigFont\undefined%
\gdef\SetFigFont#1#2#3#4#5{%
  \reset@font\fontsize{#1}{#2pt}%
  \fontfamily{#3}\fontseries{#4}\fontshape{#5}%
  \selectfont}%
\fi\endgroup%
\begin{picture}(1539,2660)(5568,-7091)
\put(6125,-5762){\makebox(0,0)[lb]{\smash{{\SetFigFont{8}{9.6}{\rmdefault}{\mddefault}{\updefault}{\color[rgb]{0,0,0}$\Sigma$}%
}}}}
\put(6301,-6736){\rotatebox{45.0}{\makebox(0,0)[lb]{\smash{{\SetFigFont{8}{9.6}{\rmdefault}{\mddefault}{\updefault}{\color[rgb]{0,0,0}$\mathcal{I}^-$}%
}}}}}
\put(6226,-4861){\rotatebox{315.0}{\makebox(0,0)[lb]{\smash{{\SetFigFont{8}{9.6}{\rmdefault}{\mddefault}{\updefault}{\color[rgb]{0,0,0}$\mathcal{I}^+$}%
}}}}}
\put(5701,-5986){\rotatebox{90.0}{\makebox(0,0)[lb]{\smash{{\SetFigFont{8}{9.6}{\rmdefault}{\mddefault}{\updefault}{\color[rgb]{0,0,0}$r=0$}%
}}}}}
\end{picture}%

%% file: nsing.pstex_t
\begin{picture}(0,0)%
\includegraphics{nsing.pstex}%
\end{picture}%
\setlength{\unitlength}{2368sp}%
\begingroup\makeatletter\ifx\SetFigFont\undefined%
\gdef\SetFigFont#1#2#3#4#5{%
  \reset@font\fontsize{#1}{#2pt}%
  \fontfamily{#3}\fontseries{#4}\fontshape{#5}%
  \selectfont}%
\fi\endgroup%
\begin{picture}(2181,1381)(4926,-5816)
\put(6226,-4786){\rotatebox{315.0}{\makebox(0,0)[lb]{\smash{{\SetFigFont{7}{8.4}{\rmdefault}{\mddefault}{\updefault}{\color[rgb]{0,0,0}$\mathcal{I}^+$}%
}}}}}
\put(5085,-5702){\rotatebox{90.0}{\makebox(0,0)[lb]{\smash{{\SetFigFont{7}{8.4}{\rmdefault}{\mddefault}{\updefault}{\color[rgb]{0,0,0}$r=0$}%
}}}}}
\end{picture}%

%% file: kerr2.pstex_t
\begin{picture}(0,0)%
\includegraphics{kerr2.pstex}%
\end{picture}%
\setlength{\unitlength}{2763sp}%
\begingroup\makeatletter\ifx\SetFigFont\undefined%
\gdef\SetFigFont#1#2#3#4#5{%
  \reset@font\fontsize{#1}{#2pt}%
  \fontfamily{#3}\fontseries{#4}\fontshape{#5}%
  \selectfont}%
\fi\endgroup%
\begin{picture}(4730,4224)(3036,-7123)
\put(5851,-4711){\makebox(0,0)[lb]{\smash{{\SetFigFont{8}{9.6}{\rmdefault}{\mddefault}{\updefault}{\color[rgb]{0,0,0}$\Sigma$}%
}}}}
\put(4501,-4186){\rotatebox{315.0}{\makebox(0,0)[lb]{\smash{{\SetFigFont{8}{9.6}{\rmdefault}{\mddefault}{\updefault}{\color[rgb]{0,0,0}$\mathcal{H}_{B}^+$}%
}}}}}
\put(5701,-4336){\rotatebox{45.0}{\makebox(0,0)[lb]{\smash{{\SetFigFont{8}{9.6}{\rmdefault}{\mddefault}{\updefault}{\color[rgb]{0,0,0}$\mathcal{H}_{A}^+$}%
}}}}}
\put(7351,-5011){\makebox(0,0)[lb]{\smash{{\SetFigFont{8}{9.6}{\rmdefault}{\mddefault}{\updefault}{\color[rgb]{0,0,0}$i_0$}%
}}}}
\put(4501,-3436){\rotatebox{45.0}{\makebox(0,0)[lb]{\smash{{\SetFigFont{8}{9.6}{\rmdefault}{\mddefault}{\updefault}{\color[rgb]{0,0,0}$\mathcal{CH}^+$}%
}}}}}
\put(5476,-3061){\rotatebox{315.0}{\makebox(0,0)[lb]{\smash{{\SetFigFont{8}{9.6}{\rmdefault}{\mddefault}{\updefault}{\color[rgb]{0,0,0}$\mathcal{CH}^+$}%
}}}}}
\put(6676,-4336){\makebox(0,0)[lb]{\smash{{\SetFigFont{8}{9.6}{\rmdefault}{\mddefault}{\updefault}{\color[rgb]{0,0,0}$\mathcal{I}_{A}^+$}%
}}}}
\put(3151,-4486){\makebox(0,0)[lb]{\smash{{\SetFigFont{8}{9.6}{\rmdefault}{\mddefault}{\updefault}{\color[rgb]{0,0,0}$\mathcal{I}_{B}^+$}%
}}}}
\put(6826,-5686){\makebox(0,0)[lb]{\smash{{\SetFigFont{8}{9.6}{\rmdefault}{\mddefault}{\updefault}{\color[rgb]{0,0,0}$\mathcal{I}_{A}^-$}%
}}}}
\put(3226,-5611){\makebox(0,0)[lb]{\smash{{\SetFigFont{8}{9.6}{\rmdefault}{\mddefault}{\updefault}{\color[rgb]{0,0,0}$\mathcal{I}_{B}^-$}%
}}}}
\end{picture}%

%% file: regpast.pstex_t
\begin{picture}(0,0)%
\includegraphics{regpast.pstex}%
\end{picture}%
\setlength{\unitlength}{2763sp}%
\begingroup\makeatletter\ifx\SetFigFont\undefined%
\gdef\SetFigFont#1#2#3#4#5{%
  \reset@font\fontsize{#1}{#2pt}%
  \fontfamily{#3}\fontseries{#4}\fontshape{#5}%
  \selectfont}%
\fi\endgroup%
\begin{picture}(1536,2006)(5568,-7091)
\put(6646,-5281){\rotatebox{315.0}{\makebox(0,0)[lb]{\smash{{\SetFigFont{8}{9.6}{\rmdefault}{\mddefault}{\updefault}{\color[rgb]{0,0,0}$\mathcal{I}^+$}%
}}}}}
\put(5701,-5986){\rotatebox{90.0}{\makebox(0,0)[lb]{\smash{{\SetFigFont{8}{9.6}{\rmdefault}{\mddefault}{\updefault}{\color[rgb]{0,0,0}$r=0$}%
}}}}}
\put(6301,-6736){\rotatebox{45.0}{\makebox(0,0)[lb]{\smash{{\SetFigFont{8}{9.6}{\rmdefault}{\mddefault}{\updefault}{\color[rgb]{0,0,0}$\mathcal{I}^-$}%
}}}}}
\put(6315,-5676){\rotatebox{45.0}{\makebox(0,0)[lb]{\smash{{\SetFigFont{8}{9.6}{\rmdefault}{\mddefault}{\updefault}{\color[rgb]{0,0,0}$\mathcal{H}^+$}%
}}}}}
\put(6073,-5232){\makebox(0,0)[lb]{\smash{{\SetFigFont{8}{9.6}{\rmdefault}{\mddefault}{\updefault}{\color[rgb]{0,0,0}$r=0$}%
}}}}
\put(5960,-5605){\makebox(0,0)[lb]{\smash{{\SetFigFont{8}{9.6}{\rmdefault}{\mddefault}{\updefault}{\color[rgb]{0,0,0}$p$}%
}}}}
\end{picture}%

%% file: wCauchy.pstex_t
\begin{picture}(0,0)%
\includegraphics{wCauchy.pstex}%
\end{picture}%
\setlength{\unitlength}{2763sp}%
\begingroup\makeatletter\ifx\SetFigFont\undefined%
\gdef\SetFigFont#1#2#3#4#5{%
  \reset@font\fontsize{#1}{#2pt}%
  \fontfamily{#3}\fontseries{#4}\fontshape{#5}%
  \selectfont}%
\fi\endgroup%
\begin{picture}(5239,3031)(1895,-7310)
\put(4201,-5311){\makebox(0,0)[lb]{\smash{{\SetFigFont{8}{9.6}{\rmdefault}{\mddefault}{\updefault}{\color[rgb]{0,0,0}$\Sigma$}%
}}}}
\put(3451,-5011){\rotatebox{315.0}{\makebox(0,0)[lb]{\smash{{\SetFigFont{8}{9.6}{\rmdefault}{\mddefault}{\updefault}{\color[rgb]{0,0,0}$r=2M$}%
}}}}}
\put(4876,-5311){\rotatebox{45.0}{\makebox(0,0)[lb]{\smash{{\SetFigFont{8}{9.6}{\rmdefault}{\mddefault}{\updefault}{\color[rgb]{0,0,0}$r=2M$}%
}}}}}
\put(2401,-5236){\rotatebox{45.0}{\makebox(0,0)[lb]{\smash{{\SetFigFont{8}{9.6}{\rmdefault}{\mddefault}{\updefault}{\color[rgb]{0,0,0}$\mathcal{I}^+$}%
}}}}}
\put(6226,-4861){\rotatebox{315.0}{\makebox(0,0)[lb]{\smash{{\SetFigFont{8}{9.6}{\rmdefault}{\mddefault}{\updefault}{\color[rgb]{0,0,0}$\mathcal{I}^+$}%
}}}}}
\put(6301,-6736){\rotatebox{45.0}{\makebox(0,0)[lb]{\smash{{\SetFigFont{8}{9.6}{\rmdefault}{\mddefault}{\updefault}{\color[rgb]{0,0,0}$\mathcal{I}^-$}%
}}}}}
\put(2251,-6361){\rotatebox{315.0}{\makebox(0,0)[lb]{\smash{{\SetFigFont{8}{9.6}{\rmdefault}{\mddefault}{\updefault}{\color[rgb]{0,0,0}$\mathcal{I}^-$}%
}}}}}
\put(4351,-4411){\makebox(0,0)[lb]{\smash{{\SetFigFont{8}{9.6}{\rmdefault}{\mddefault}{\updefault}{\color[rgb]{0,0,0}$r=0$}%
}}}}
\put(4276,-7261){\makebox(0,0)[lb]{\smash{{\SetFigFont{8}{9.6}{\rmdefault}{\mddefault}{\updefault}{\color[rgb]{0,0,0}$r=0$}%
}}}}
\put(5401,-6136){\makebox(0,0)[lb]{\smash{{\SetFigFont{8}{9.6}{\rmdefault}{\mddefault}{\updefault}{\color[rgb]{0,0,0}$\mathcal{D}$}%
}}}}
\end{picture}%

%% file: wfol.pstex_t
\begin{picture}(0,0)%
\includegraphics{wfol.pstex}%
\end{picture}%
\setlength{\unitlength}{2368sp}%
\begingroup\makeatletter\ifx\SetFigFont\undefined%
\gdef\SetFigFont#1#2#3#4#5{%
  \reset@font\fontsize{#1}{#2pt}%
  \fontfamily{#3}\fontseries{#4}\fontshape{#5}%
  \selectfont}%
\fi\endgroup%
\begin{picture}(5239,3031)(1895,-7310)
\put(4201,-5311){\makebox(0,0)[lb]{\smash{{\SetFigFont{7}{8.4}{\rmdefault}{\mddefault}{\updefault}{\color[rgb]{0,0,0}$\Sigma$}%
}}}}
\put(3451,-5011){\rotatebox{315.0}{\makebox(0,0)[lb]{\smash{{\SetFigFont{7}{8.4}{\rmdefault}{\mddefault}{\updefault}{\color[rgb]{0,0,0}$r=2M$}%
}}}}}
\put(4876,-5311){\rotatebox{45.0}{\makebox(0,0)[lb]{\smash{{\SetFigFont{7}{8.4}{\rmdefault}{\mddefault}{\updefault}{\color[rgb]{0,0,0}$\mathcal{H}^+(0,\tau)$}%
}}}}}
\put(2401,-5236){\rotatebox{45.0}{\makebox(0,0)[lb]{\smash{{\SetFigFont{7}{8.4}{\rmdefault}{\mddefault}{\updefault}{\color[rgb]{0,0,0}$\mathcal{I}^+$}%
}}}}}
\put(6226,-4861){\rotatebox{315.0}{\makebox(0,0)[lb]{\smash{{\SetFigFont{7}{8.4}{\rmdefault}{\mddefault}{\updefault}{\color[rgb]{0,0,0}$\mathcal{I}^+$}%
}}}}}
\put(6301,-6736){\rotatebox{45.0}{\makebox(0,0)[lb]{\smash{{\SetFigFont{7}{8.4}{\rmdefault}{\mddefault}{\updefault}{\color[rgb]{0,0,0}$\mathcal{I}^-$}%
}}}}}
\put(2251,-6361){\rotatebox{315.0}{\makebox(0,0)[lb]{\smash{{\SetFigFont{7}{8.4}{\rmdefault}{\mddefault}{\updefault}{\color[rgb]{0,0,0}$\mathcal{I}^-$}%
}}}}}
\put(4351,-4411){\makebox(0,0)[lb]{\smash{{\SetFigFont{7}{8.4}{\rmdefault}{\mddefault}{\updefault}{\color[rgb]{0,0,0}$r=0$}%
}}}}
\put(4276,-7261){\makebox(0,0)[lb]{\smash{{\SetFigFont{7}{8.4}{\rmdefault}{\mddefault}{\updefault}{\color[rgb]{0,0,0}$r=0$}%
}}}}
\put(5401,-6136){\makebox(0,0)[lb]{\smash{{\SetFigFont{7}{8.4}{\rmdefault}{\mddefault}{\updefault}{\color[rgb]{0,0,0}$\mathcal{D}$}%
}}}}
\put(5551,-4936){\makebox(0,0)[lb]{\smash{{\SetFigFont{7}{8.4}{\rmdefault}{\mddefault}{\updefault}{\color[rgb]{0,0,0}$\Sigma_\tau$}%
}}}}
\end{picture}%

%% file: redshift.pstex_t
\begin{picture}(0,0)%
\includegraphics{redshift.pstex}%
\end{picture}%
\setlength{\unitlength}{2368sp}%
\begingroup\makeatletter\ifx\SetFigFont\undefined%
\gdef\SetFigFont#1#2#3#4#5{%
  \reset@font\fontsize{#1}{#2pt}%
  \fontfamily{#3}\fontseries{#4}\fontshape{#5}%
  \selectfont}%
\fi\endgroup%
\begin{picture}(3763,2002)(4950,-6673)
\put(6826,-5911){\makebox(0,0)[lb]{\smash{{\SetFigFont{7}{8.4}{\rmdefault}{\mddefault}{\updefault}{\color[rgb]{0,0,0}$B$}%
}}}}
\put(5776,-5311){\makebox(0,0)[lb]{\smash{{\SetFigFont{7}{8.4}{\rmdefault}{\mddefault}{\updefault}{\color[rgb]{0,0,0}$\mathcal{H}^+$}%
}}}}
\put(7726,-5611){\makebox(0,0)[lb]{\smash{{\SetFigFont{7}{8.4}{\rmdefault}{\mddefault}{\updefault}{\color[rgb]{0,0,0}$\mathcal{I}^+$}%
}}}}
\put(5326,-6511){\makebox(0,0)[lb]{\smash{{\SetFigFont{7}{8.4}{\rmdefault}{\mddefault}{\updefault}{\color[rgb]{0,0,0}$A$}%
}}}}
\end{picture}%

%% file: locredshift.pstex_t
\begin{picture}(0,0)%
\includegraphics{locredshift.pstex}%
\end{picture}%
\setlength{\unitlength}{2368sp}%
\begingroup\makeatletter\ifx\SetFigFont\undefined%
\gdef\SetFigFont#1#2#3#4#5{%
  \reset@font\fontsize{#1}{#2pt}%
  \fontfamily{#3}\fontseries{#4}\fontshape{#5}%
  \selectfont}%
\fi\endgroup%
\begin{picture}(3763,2002)(4950,-6673)
\put(5776,-5311){\makebox(0,0)[lb]{\smash{{\SetFigFont{7}{8.4}{\rmdefault}{\mddefault}{\updefault}{\color[rgb]{0,0,0}$\mathcal{H}^+$}%
}}}}
\put(7726,-5611){\makebox(0,0)[lb]{\smash{{\SetFigFont{7}{8.4}{\rmdefault}{\mddefault}{\updefault}{\color[rgb]{0,0,0}$\mathcal{I}^+$}%
}}}}
\put(5326,-6511){\makebox(0,0)[lb]{\smash{{\SetFigFont{7}{8.4}{\rmdefault}{\mddefault}{\updefault}{\color[rgb]{0,0,0}$A$}%
}}}}
\put(6526,-5911){\makebox(0,0)[lb]{\smash{{\SetFigFont{7}{8.4}{\rmdefault}{\mddefault}{\updefault}{\color[rgb]{0,0,0}$B$}%
}}}}
\end{picture}%

%% file: newfol.pstex_t
\begin{picture}(0,0)%
\includegraphics{newfol.pstex}%
\end{picture}%
\setlength{\unitlength}{2763sp}%
\begingroup\makeatletter\ifx\SetFigFont\undefined%
\gdef\SetFigFont#1#2#3#4#5{%
  \reset@font\fontsize{#1}{#2pt}%
  \fontfamily{#3}\fontseries{#4}\fontshape{#5}%
  \selectfont}%
\fi\endgroup%
\begin{picture}(2645,2660)(4489,-7091)
\put(4876,-5311){\rotatebox{45.0}{\makebox(0,0)[lb]{\smash{{\SetFigFont{8}{9.6}{\rmdefault}{\mddefault}{\updefault}{\color[rgb]{0,0,0}$\mathcal{H}^+(0,\tau)$}%
}}}}}
\put(6226,-4861){\rotatebox{315.0}{\makebox(0,0)[lb]{\smash{{\SetFigFont{8}{9.6}{\rmdefault}{\mddefault}{\updefault}{\color[rgb]{0,0,0}$\mathcal{I}^+$}%
}}}}}
\put(6301,-6736){\rotatebox{45.0}{\makebox(0,0)[lb]{\smash{{\SetFigFont{8}{9.6}{\rmdefault}{\mddefault}{\updefault}{\color[rgb]{0,0,0}$\mathcal{I}^-$}%
}}}}}
\put(4951,-5911){\makebox(0,0)[lb]{\smash{{\SetFigFont{8}{9.6}{\rmdefault}{\mddefault}{\updefault}{\color[rgb]{0,0,0}$t=0$}%
}}}}
\put(5326,-5461){\makebox(0,0)[lb]{\smash{{\SetFigFont{8}{9.6}{\rmdefault}{\mddefault}{\updefault}{\color[rgb]{0,0,0}$\tilde\Sigma_0$}%
}}}}
\put(5401,-6136){\makebox(0,0)[lb]{\smash{{\SetFigFont{8}{9.6}{\rmdefault}{\mddefault}{\updefault}{\color[rgb]{0,0,0}$\mathcal{D}$}%
}}}}
\put(5551,-5011){\makebox(0,0)[lb]{\smash{{\SetFigFont{8}{9.6}{\rmdefault}{\mddefault}{\updefault}{\color[rgb]{0,0,0}$\tilde\Sigma_\tau$}%
}}}}
\end{picture}%

%% file: kerr.pstex_t
\begin{picture}(0,0)%
\includegraphics{kerr.pstex}%
\end{picture}%
\setlength{\unitlength}{2763sp}%
\begingroup\makeatletter\ifx\SetFigFont\undefined%
\gdef\SetFigFont#1#2#3#4#5{%
  \reset@font\fontsize{#1}{#2pt}%
  \fontfamily{#3}\fontseries{#4}\fontshape{#5}%
  \selectfont}%
\fi\endgroup%
\begin{picture}(4730,4224)(3036,-7123)
\put(6076,-4411){\makebox(0,0)[lb]{\smash{{\SetFigFont{8}{9.6}{\rmdefault}{\mddefault}{\updefault}{\color[rgb]{0,0,0}$\mathcal{D}$}%
}}}}
\put(5701,-4336){\rotatebox{45.0}{\makebox(0,0)[lb]{\smash{{\SetFigFont{8}{9.6}{\rmdefault}{\mddefault}{\updefault}{\color[rgb]{0,0,0}$\mathcal{H}_{A}^+$}%
}}}}}
\put(7351,-5011){\makebox(0,0)[lb]{\smash{{\SetFigFont{8}{9.6}{\rmdefault}{\mddefault}{\updefault}{\color[rgb]{0,0,0}$i_0$}%
}}}}
\put(4501,-3436){\rotatebox{45.0}{\makebox(0,0)[lb]{\smash{{\SetFigFont{8}{9.6}{\rmdefault}{\mddefault}{\updefault}{\color[rgb]{0,0,0}$\mathcal{CH}^+$}%
}}}}}
\put(5476,-3061){\rotatebox{315.0}{\makebox(0,0)[lb]{\smash{{\SetFigFont{8}{9.6}{\rmdefault}{\mddefault}{\updefault}{\color[rgb]{0,0,0}$\mathcal{CH}^+$}%
}}}}}
\put(5851,-4711){\makebox(0,0)[lb]{\smash{{\SetFigFont{8}{9.6}{\rmdefault}{\mddefault}{\updefault}{\color[rgb]{0,0,0}$\Sigma$}%
}}}}
\put(6676,-4336){\makebox(0,0)[lb]{\smash{{\SetFigFont{8}{9.6}{\rmdefault}{\mddefault}{\updefault}{\color[rgb]{0,0,0}$\mathcal{I}_{A}^+$}%
}}}}
\put(3151,-4486){\makebox(0,0)[lb]{\smash{{\SetFigFont{8}{9.6}{\rmdefault}{\mddefault}{\updefault}{\color[rgb]{0,0,0}$\mathcal{I}_{B}^+$}%
}}}}
\put(6826,-5686){\makebox(0,0)[lb]{\smash{{\SetFigFont{8}{9.6}{\rmdefault}{\mddefault}{\updefault}{\color[rgb]{0,0,0}$\mathcal{I}_{A}^-$}%
}}}}
\put(3226,-5611){\makebox(0,0)[lb]{\smash{{\SetFigFont{8}{9.6}{\rmdefault}{\mddefault}{\updefault}{\color[rgb]{0,0,0}$\mathcal{I}_{B}^-$}%
}}}}
\put(4501,-4186){\rotatebox{315.0}{\makebox(0,0)[lb]{\smash{{\SetFigFont{8}{9.6}{\rmdefault}{\mddefault}{\updefault}{\color[rgb]{0,0,0}$\mathcal{H}_{B}^+$}%
}}}}}
\end{picture}%

%% file: desitforkerr.pstex_t
\begin{picture}(0,0)%
\includegraphics{desitforkerr.pstex}%
\end{picture}%
\setlength{\unitlength}{2368sp}%
\begingroup\makeatletter\ifx\SetFigFont\undefined%
\gdef\SetFigFont#1#2#3#4#5{%
  \reset@font\fontsize{#1}{#2pt}%
  \fontfamily{#3}\fontseries{#4}\fontshape{#5}%
  \selectfont}%
\fi\endgroup%
\begin{picture}(3549,2164)(2914,-4694)
\put(5626,-2686){\makebox(0,0)[lb]{\smash{{\SetFigFont{7}{8.4}{\rmdefault}{\mddefault}{\updefault}{\color[rgb]{0,0,0}$r=\infty$}%
}}}}
\put(5701,-4636){\makebox(0,0)[lb]{\smash{{\SetFigFont{7}{8.4}{\rmdefault}{\mddefault}{\updefault}{\color[rgb]{0,0,0}$r=\infty$}%
}}}}
\put(3991,-2821){\makebox(0,0)[lb]{\smash{{\SetFigFont{7}{8.4}{\rmdefault}{\mddefault}{\updefault}{\color[rgb]{0,0,0}$r=0$}%
}}}}
\put(3991,-4486){\makebox(0,0)[lb]{\smash{{\SetFigFont{7}{8.4}{\rmdefault}{\mddefault}{\updefault}{\color[rgb]{0,0,0}$r=0$}%
}}}}
\put(4276,-3961){\rotatebox{315.0}{\makebox(0,0)[lb]{\smash{{\SetFigFont{7}{8.4}{\rmdefault}{\mddefault}{\updefault}{\color[rgb]{0,0,0}$\mathcal{H}^-$}%
}}}}}
\put(5476,-3061){\rotatebox{315.0}{\makebox(0,0)[lb]{\smash{{\SetFigFont{7}{8.4}{\rmdefault}{\mddefault}{\updefault}{\color[rgb]{0,0,0}$\overline{\mathcal{H}}^+$}%
}}}}}
\put(4951,-3586){\makebox(0,0)[lb]{\smash{{\SetFigFont{7}{8.4}{\rmdefault}{\mddefault}{\updefault}{\color[rgb]{0,0,0}$\mathcal{D}$}%
}}}}
\put(4413,-3281){\rotatebox{45.0}{\makebox(0,0)[lb]{\smash{{\SetFigFont{7}{8.4}{\rmdefault}{\mddefault}{\updefault}{\color[rgb]{0,0,0}$\mathcal{H}^+$}%
}}}}}
\put(4920,-3266){\makebox(0,0)[lb]{\smash{{\SetFigFont{7}{8.4}{\rmdefault}{\mddefault}{\updefault}{\color[rgb]{0,0,0}$\Sigma$}%
}}}}
\put(5626,-4186){\rotatebox{45.0}{\makebox(0,0)[lb]{\smash{{\SetFigFont{7}{8.4}{\rmdefault}{\mddefault}{\updefault}{\color[rgb]{0,0,0}$\overline{\mathcal{H}}^-$}%
}}}}}
\end{picture}%

%% file: SchwarzschildADS.pstex_t
\begin{picture}(0,0)%
\includegraphics{schwarzschildADS.pstex}%
\end{picture}%
\setlength{\unitlength}{1579sp}%
\begingroup\makeatletter\ifx\SetFigFont\undefined%
\gdef\SetFigFont#1#2#3#4#5{%
  \reset@font\fontsize{#1}{#2pt}%
  \fontfamily{#3}\fontseries{#4}\fontshape{#5}%
  \selectfont}%
\fi\endgroup%
\begin{picture}(5575,5451)(1762,-6034)
\put(2101,-5911){\makebox(0,0)[lb]{\smash{{\SetFigFont{10}{12.0}{\rmdefault}{\mddefault}{\updefault}{\color[rgb]{0,0,0}$i^-$}%
}}}}
\put(3526,-2086){\rotatebox{315.0}{\makebox(0,0)[lb]{\smash{{\SetFigFont{10}{12.0}{\rmdefault}{\mddefault}{\updefault}{\color[rgb]{0,0,0}$\mathcal{H}^+$}%
}}}}}
\put(5701,-2686){\rotatebox{45.0}{\makebox(0,0)[lb]{\smash{{\SetFigFont{10}{12.0}{\rmdefault}{\mddefault}{\updefault}{\color[rgb]{0,0,0}$\mathcal{H}^+$}%
}}}}}
\put(6451,-5911){\makebox(0,0)[lb]{\smash{{\SetFigFont{10}{12.0}{\rmdefault}{\mddefault}{\updefault}{\color[rgb]{0,0,0}$i^-$}%
}}}}
\put(6376,-961){\makebox(0,0)[lb]{\smash{{\SetFigFont{10}{12.0}{\rmdefault}{\mddefault}{\updefault}{\color[rgb]{0,0,0}$i^+$}%
}}}}
\put(7201,-3586){\rotatebox{90.0}{\makebox(0,0)[lb]{\smash{{\SetFigFont{10}{12.0}{\rmdefault}{\mddefault}{\updefault}{\color[rgb]{0,0,0}$\mathcal{I}$}%
}}}}}
\put(2176,-886){\makebox(0,0)[lb]{\smash{{\SetFigFont{10}{12.0}{\rmdefault}{\mddefault}{\updefault}{\color[rgb]{0,0,0}$i^+$}%
}}}}
\put(2101,-3586){\rotatebox{90.0}{\makebox(0,0)[lb]{\smash{{\SetFigFont{10}{12.0}{\rmdefault}{\mddefault}{\updefault}{\color[rgb]{0,0,0}$\mathcal{I}$}%
}}}}}
\put(4051,-1336){\makebox(0,0)[lb]{\smash{{\SetFigFont{10}{12.0}{\rmdefault}{\mddefault}{\updefault}{\color[rgb]{0,0,0}$r=0$}%
}}}}
\put(4051,-5611){\makebox(0,0)[lb]{\smash{{\SetFigFont{10}{12.0}{\rmdefault}{\mddefault}{\updefault}{\color[rgb]{0,0,0}$r=0$}%
}}}}
\end{picture}%